\documentclass[pra,twocolumn,superscriptaddress]{revtex4-2}
\usepackage{amsmath,amssymb,graphicx,bm}
\usepackage{amsthm}
\usepackage[colorlinks=true,citecolor=blue,linkcolor=blue,urlcolor=blue]{hyperref}
\usepackage{braket}
\usepackage{qcircuit}
\usepackage{bbm}

\newtheorem{lemma}{Lemma}

\newtheorem{theorem}{Theorem}

\begin{document}

\title{Convex combinations of bosonic pure-loss channels}
\author{Giuseppe Catalano}
\email{giuseppe.catalano@sns.it} 
\affiliation{NEST-CNR Scuola Normale Superiore, Piazza dei Cavalieri 7, I-56126 Pisa, Italy}

\author{Marco Fanizza}
\affiliation{Inria, Institut Polytechnique de Paris, 1 Rue Honoré d'Estienne d'Orves, 91120 Palaiseau, France}

\author{Francesco Anna Mele}
\affiliation{NEST-CNR Scuola Normale Superiore, Piazza dei Cavalieri 7, I-56126 Pisa, Italy}

\author{Giacomo De Palma}
\affiliation{Dipartimento di Matematica, Università di Bologna, Piazza di Porta San Donato 5, I-40126 Bologna, Italy}

\author{Vittorio Giovannetti}
\affiliation{NEST-CNR Scuola Normale Superiore, Piazza dei Cavalieri 7, I-56126 Pisa, Italy}

\begin{abstract}

\noindent The pure-loss channel is a fundamental model for describing noise in bosonic quantum platforms. It is characterised by a single parameter, the transmissivity, which quantifies the fraction of the input energy that reaches the output of the channel. In realistic scenarios, however, such as free-space quantum communication, the transmissivity is not fixed but fluctuates from one channel use to another. In this setting, the overall channel is effectively described as a convex combination of pure-loss channels, known as a fading channel. Despite its practical relevance, the quantum Shannon theory of the fading channel has remained largely unexplored. Here, we address this gap, 
specifically investigating degradability, anti-degradability, entanglement breakingness, and capacities of the fading channel. 
Of particular relevance to practical quantum-internet applications, we prove that entanglement distribution and quantum key distribution can always be achieved at a strictly positive rate over any fading channel, no matter how noisy it is or how strongly the transmissivity fluctuates, provided the channel is not completely noisy. When the transmissivity takes a finite set of values, we determine the energy-unconstrained two-way quantum and secret-key capacities exactly, as the averages of those of the pure-loss components. Moreover, we prove that thermal states, which are optimal for a broad class 
of static bosonic Gaussian channels, do not in general achieve 
the entanglement-assisted classical capacity of fading channels: for a binary fading model we derive the capacity-achieving state in closed form, and we exhibit channels for which non-Gaussian Fock-diagonal states strictly outperform every Gaussian encoding. For the quantum capacity, we give a simple sufficient condition on the transmissivity distribution under which weak thermal inputs already yield a strictly positive rate. Outside this condition, we numerically identify parameter regions where no thermal input compatible with the energy constraint has a positive coherent information, while optimized non-Gaussian inputs do, so that quantum communication is possible there. For general fading distributions, 
we design an iterative variational algorithm to optimize the coherent and 
mutual information, which may be of independent interest. Overall, our work unveils fundamental properties of this practically relevant channel and advances the study of quantum communication in the non-Gaussian regime.
\end{abstract}

\maketitle

\section{Introduction}

\noindent The development of a global quantum network~\cite{quantum_internet_kimble,quantum_internet_Wehner} is one of the central goals of quantum information science~\cite{NC}. While optical fibers are the standard medium for short-range connections, their exponential loss limits their use for long-distance quantum communication~\cite{Caves,SerafiniBook,Record1,Record2,Record3,Record4,Record5,Pirandola_2020,usenko2025continuousvariablequantumcommunication,PLOB,MeleNature}. Free-space optical links, particularly ground-to-satellite channels, have emerged as the enabling technology for global-scale quantum connectivity~\cite{deForgesDeParny2023,RevModPhys.94.035001,Krzic2023,Sidhu_2021}. Traditionally, the fundamental model used to describe attenuation in such links is the pure-loss channel, which assumes a constant, fixed transmissivity between sender and receiver. However, real-world free-space communication is inherently dynamic. Atmospheric turbulence induces stochastic fluctuations in the channel's physical parameters, most notably its transmissivity~\cite{Pirandola2021, Usenko2018OE_WeibullExperiment}. This paradigmatic real-world noise, formally known as fading, is caused by a combination of effects, including beam broadening, scintillation, and degradation of coherence~\cite{Andrews2005_LaserBeam, Vasylyev2016PRL_EllipticBeam}. In many relevant scenarios, such as mid-range links or weak turbulence, the dominant effect causing fluctuating losses is beam wandering, where the beam spot is randomly deflected and travels around the receiving aperture~\cite{ Vasylyev2012PRL_LogNegativeWeibull}.\\
\noindent Given the ubiquity of these fluctuation effects in realistic free-space communication settings, understanding the ultimate communication limits of these fluctuating channels is crucial. Shannon’s noisy channel coding theorem~\cite{Shannon1948} established these limits for classical channels.
\noindent In the quantum setting, the communication potential of a channel is captured by a richer family of capacities. We are primarily interested in the \emph{Entanglement-Assisted Classical Capacity} ($C_\text{E}$), which quantifies the optimal rate of classical information transmission when sender and receiver share unlimited pre-existing entanglement~\cite{Bennett1999PRLEACC,Bennett2002IEEEEACC}, and the unassisted \emph{Quantum Capacity} ($Q$), which defines the rate at which quantum states can be reliably transmitted~\cite{DevetakShor2005_ChannelCapacities}. Furthermore, since practical quantum communication protocols often allow for classical feedback, it is fundamentally relevant to consider the \emph{Two-Way Assisted Quantum Capacity} ($Q_2$), which represents the highest rate of reliable quantum communication achievable when unlimited classical communication is allowed between sender and receiver, and the \emph{Secret Key Capacity} ($K$), which sets the ultimate limit for generating a secure cryptographic key~\cite{Pirandola2021,WildeBook}. Notably, the fundamental hierarchy $K \ge Q_2 \ge Q$ implies that establishing a strictly positive $Q_2$ automatically guarantees the feasibility of both entanglement distribution and Quantum Key Distribution (QKD) over the given channel~\cite{HolevoGiovannetti2012}.\\

\noindent Continuous-variable (CV) systems are a key component of quantum technologies as they model quantum optical platforms and bosonic systems. Among CV states, Gaussian states play a crucial role, given their ubiquity in nature and the fact that they are relatively simple to analyze analytically. More specifically, they also play a fundamental role because of their extremality properties~\cite{SerafiniBook,Weedbrook2012}. It is a well-known result that, for a fixed energy constraint, Gaussian states maximize the von Neumann entropy~\cite{wolf_PRL_gaussian_extramality}. Furthermore, regarding channel dynamics, it has been shown that Gaussian inputs (specifically, thermal states) minimize the output entropy of one-mode phase-insensitive Gaussian channels among all the input states with a given entropy~\cite{PhysRevLett.118.160503}. 

\noindent These extremality properties underpin a well-established body of results showing 
that Gaussian encodings, specifically thermal states, achieve the optimal 
communication rates for a broad class of bosonic Gaussian channels under energy 
constraints~\cite{Holevo2001PRAGaussian, giovannetti_solution_gaussian_optimality_conjecture_2015, 
PhysRevLett.118.160503}. This optimality, however, relies fundamentally on the 
Gaussian nature of the channel itself. The fading channel, being a convex 
combination of pure-loss maps, is generally non-Gaussian, and therefore falls 
outside the scope of these results. As we demonstrate, this extension fails: there are fading channels for which thermal states are strictly sub-optimal for the entanglement-assisted classical capacity, where non-Gaussian Fock-diagonal states achieve strictly higher rates, and thermal states are outperformed by non-Gaussian states also at the level of the single-shot coherent information $Q_1$.\\

\noindent A fading channel is mathematically described as a convex combination (mixture) of pure-loss  channels $\mathcal{E}_{\lambda}$ with fluctuating transmissivity $\lambda$~\cite{Pirandola2021}:
\begin{equation}
    \Phi_{p(\lambda)}(\rho) = \int_0^1 d\lambda\, p(\lambda)\, \mathcal{E}_{\lambda}(\rho),
    \label{eq:fading-channel}
\end{equation}
where $p(\lambda)$ is the probability distribution of the transmissivity $\lambda$. The starting point of this work is the observation that the set of Gaussian channels is not convex. Consequently, fading channels, being convex combinations of pure-loss maps, are in general non-Gaussian, and the optimality of Gaussian encodings is no longer guaranteed. While the quantum Shannon-theoretic properties of pure-loss channels, including degradability, anti-degradability, entanglement breakingness, and communication capacities, have been extensively investigated~\cite{Holevo2001PRAGaussian,Wolf2007,PLOB,MMMM,TGW,Mark2012,Pirandola2009,giovannetti_solution_gaussian_optimality_conjecture_2015,LossyECEAC1,LossyECEAC2,Caruso_weak,Giovadd,mele2025achievableratesnonasymptoticbosonic,Winnel,KHATRI}, much less is known about fading channels.
In this work, we make progress towards filling this gap by providing a comprehensive analysis of the quantum Shannon-theoretic properties of the fading channel. Our main contributions are as follows:
\begin{itemize}
    \item \textbf{Fundamental Channel Properties:} We rigorously prove that any non-trivial convex combination of lossy channels (Eq.~\ref{eq:fading-channel}) is \textbf{non-degradable} (i.e., its complementary map cannot be simulated by post-processing the channel's output) by showing its Choi matrix rank exceeds the requirement for degradability. We also derive a sufficient condition for \textbf{non-antidegradability}, $\langle\sqrt{\lambda}\rangle^2 + \langle\lambda\rangle > 1$, based on the fidelity of complementary maps. Non-degradability implies that fading channels generally lack the simpler structure that guarantees additivity for capacity measures, i.e. the property that entangling inputs across multiple channel uses provides no communication advantage, meaning the total quantum capacity simply scales linearly. We show that the same condition guarantees a strictly positive quantum capacity, already witnessed by weak thermal inputs.

    \item \textbf{Distillability and Two-Way Capacities:} We prove that fading channels are \textbf{distillable} i.e., their two-way quantum capacity is strictly positive, $Q_2(\Phi) > 0$. This analytical result guarantees that entanglement distribution and quantum key distribution can always be achieved at a strictly positive rate over any fading channel, no matter how noisy it is or how strongly the transmissivity fluctuates, provided the channel is not completely noisy. Furthermore, when the transmissivity takes a finite set of values, we determine the two-way quantum and secret-key capacities exactly: with no constraint on the input energy, both equal $\sum_n p_n\, Q_2(\mathcal{E}_{\lambda_n})$, the average of the capacities of the individual pure-loss channels, which were determined in Ref.~\cite{PLOB}; this closes the gap between the bounds of Ref.~\cite{PLB19}. At a fixed input energy, we derive an explicit rate by adapting the multi-rail protocol of Ref.~\cite{Mele2024LossDephasing} to fading channels, which is strictly positive whenever the mean transmissivity is.
    
    \item  \textbf{Sub-optimality of Gaussian Encodings:} We demonstrate that thermal 
states, which are optimal for static bosonic Gaussian channels~\cite{Holevo2001PRAGaussian, 
giovannetti_solution_gaussian_optimality_conjecture_2015}, are in general not optimal
for the entanglement-assisted classical capacity of fading channels. In fact, we exhibit fading channels for which the entire Gaussian class fails to achieve $C_\text{E}$: this is proven analytically for a binary fading model, and for the two other binary families we studied numerically the thermal encoding is strictly sub-optimal at every point we computed. Furthermore, using the same numerical framework, we show that thermal-state encodings can be strictly sub-optimal also for the coherent information: our optimized non-Gaussian states achieve a strictly positive single-shot lower bound $Q_1$ in regimes where thermal inputs give zero or negative values.
    
    \item \textbf{Iterative Optimization Algorithm:} We develop an iterative variational algorithm to construct optimized Fock-diagonal input states by directly optimizing the relevant information-theoretic functionals of the fading channel. This method constructs increasingly expressive non-Gaussian ansätze and provides non-decreasing lower bounds for the capacities, revealing gains over the thermal benchmark at physically relevant energies; since every non-thermal Gaussian state is strictly sub-optimal for the quantum mutual information (Appendix~\ref{appendix: diagonal optimizer}), any such gain in the mutual information proves that the optimal encoding is non-Gaussian. In the case of the entanglement-assisted classical capacity, the resulting hierarchy of exact maxima is moreover guaranteed to converge to the true optimum, the values we report being lower bounds on them.
    
\item \textbf{Activation of Quantum Capacity:} We analyze the single-shot quantum capacity lower bound $Q_1(\Phi) = \max_\rho I_c(\Phi, \rho)$. We numerically find regimes where no thermal input compatible with the energy constraint has a positive coherent information, while our optimized non-Gaussian states yield a strictly positive one. Since $Q(\Phi) \ge Q_1(\Phi)$ for any channel, the channel can transmit quantum information there. Because fading channels are non-degradable, $Q_1$ may underestimate the true capacity $Q$, which we do not compute.
\end{itemize}

\noindent These contributions are organized along two complementary directions. 
The structural properties of the channel are developed in 
Sec.~\ref{sec:fundamental}, which establishes non-degradability, 
non-antidegradability, and distillability, and determines the two-way capacities for a finite set of transmissivities, together with an explicit multi-rail rate at a fixed input energy. The operational consequences are developed in 
Secs.~\ref{sec:erasure_channel}--\ref{sec:numerical}: 
Sec.~\ref{sec:erasure_channel} provides the analytical counterexample to 
Gaussian optimality for the erasure-like channel, Sec.~\ref{sec:optimization_method} 
introduces the iterative variational algorithm for general fading distributions, 
and Sec.~\ref{sec:numerical} presents the numerical results, including the 
activation of quantum communication through non-Gaussian encodings. Together, 
these results show that the capacity landscape of fading channels is 
fundamentally richer than that of static pure-loss channels, and suggest that
non-Gaussian state engineering deserves attention in the design of free-space quantum communication protocols.

\section{Preliminaries}\label{sec:prelim}
\noindent We consider quantum communication over a single-mode bosonic system, described by annihilation and creation operators ${a}$ and ${a}^\dagger$ satisfying $[{a}, {a}^\dagger] = \mathbbm{1}$. 
In this section, we introduce the fundamental Gaussian building blocks, define the fading channel model along with its Stinespring dilation, and review the relevant capacity measures.

\subsection{Gaussian Channels and States}
\noindent A bosonic quantum state $\rho$ is \emph{Gaussian} if its characteristic function $\chi(\xi) = \mathrm{Tr}[\rho \, {D}(\xi)]$  is a Gaussian function of the displacement $\xi \in \mathbb{C}$, i.e., it is fully characterized by its first and second moments. Here, $D(\xi)$ denotes the displacement operator:
\begin{equation}
    D(\xi) = \exp(\xi {a}^\dagger - \xi^* {a}).
\end{equation}
Gaussian channels are completely positive trace-preserving (CPTP) maps preserving Gaussianity. The fundamental model for attenuation is the \emph{pure-loss channel} $\mathcal{E}_\lambda$, which mixes the input mode with an environmental mode, initially in the vacuum state $\ket{0}_E$, via a beam-splitter of transmissivity $\lambda \in [0,1]$. Formally, its action is defined via the Stinespring dilation:
\begin{equation}\label{eq:pure_loss_def}
    \mathcal{E}_\lambda(\rho) = \mathrm{Tr}_E \left[ U_\lambda (\rho \otimes |0\rangle_E \langle 0|) U_\lambda^\dagger \right],
\end{equation}
where the unitary operator that defines the beam splitter interaction
\begin{equation}\label{eq:beam_splitter_def}
    U_\lambda= \exp[\cos^{-1}(\sqrt{\lambda}) ({a}^\dagger {b} - {a} {b}^\dagger)] 
\end{equation}
implements the coupling ${a} \to \sqrt{\lambda}{a} + \sqrt{1-\lambda}{b}$ between the signal and environment annihilation operators.\\

\subsection{Bosonic Fading Channels}

\noindent While more general models such as the thermal attenuator account for 
environmental noise, the pure-loss channel suffices for optical and 
free-space communications where the thermal background is negligible~\cite{Weedbrook2012}. 
In realistic free-space scenarios, however, the transmissivity $\lambda$ 
is rarely a static parameter; rather, it undergoes stochastic fluctuations 
due to atmospheric turbulence or pointing errors~\cite{Pirandola2021, 
Usenko2018OE_WeibullExperiment, Vasylyev2016PRL_EllipticBeam}. This physical 
reality leads us to model the propagation as a \emph{fading channel}, defined in Eq.~\eqref{eq:fading-channel} as a convex combination of pure-loss maps weighted by a probability distribution 
$p(\lambda)$.
Throughout the paper, $p(\lambda)$ can be any probability distribution on $[0,1]$, for instance a probability density, as in the log-negative Weibull model introduced below, or a discrete distribution, which we write as a sum of Dirac deltas.
A discrete distribution $p(\lambda) = \sum_{n} p_n\, \delta(\lambda-\lambda_n)$ gives the channel
\begin{equation}\label{eq:discretized_fading_channel}
    \Phi_{\{p_n, \lambda_n\}}(\rho) = \sum_{n=1}^d p_n \, \mathcal{E}_{\lambda_n}(\rho),
\end{equation}
where each $\mathcal{E}_{\lambda_n}$ is a pure-loss channel of transmissivity
$\lambda_n \in [0,1]$ and $\{p_n\}$ is a probability distribution. Throughout the paper we write $\Phi_{p(\lambda)}$, as in Eq.~\eqref{eq:fading-channel}, for a fading channel with an arbitrary distribution $p(\lambda)$, and $\Phi_{\{p_n,\lambda_n\}}$, as in Eq.~\eqref{eq:discretized_fading_channel}, when $p(\lambda)$ is discrete. Results written for $\Phi_{p(\lambda)}$ hold for every distribution. Theorem~\ref{thm:exact_two_way} (Sec.~\ref{sec:improved Q2}) requires a discrete distribution: it is proven for a finite set of transmissivities; the numerical method of Sec.~\ref{sec:optimization_method} also works with discrete distributions, to which continuous ones are reduced by a quadrature rule (App.~\ref{app:general_fading}).  Crucially, while Gaussian
channels are closed under composition, they are not closed under convex 
combinations~\cite{SerafiniBook}, as formalized in the following lemma.
\begin{lemma}\label{lem:convex combinations non-gaussian}
A fading channel $\Phi_{p(\lambda)}$ is non-Gaussian whenever $p(\lambda)$ is not a Dirac delta. In particular, any convex combination $\Phi_{\{p_n,\lambda_n\}} = \sum_n p_n \mathcal{E}_{\lambda_n}$ with at 
least two distinct transmissivities $\lambda_i \neq \lambda_j$, is a 
non-Gaussian channel.
\end{lemma}
\begin{proof}
	Since Gaussian channels map Gaussian states into Gaussian
	states~\cite{Weedbrook2012}, it suffices to show that
	$\Phi_{p(\lambda)}$ maps a coherent state $\ket{\alpha}$, $\alpha\neq0$,
	into a non-Gaussian state. Using
	$\mathcal{E}_\lambda(\ket{\alpha}\bra{\alpha})
	=\ket{\sqrt{\lambda}\alpha}\bra{\sqrt{\lambda}\alpha}$, the output is the
	classical mixture of coherent states
	$\rho_\alpha=\int d\lambda\,p(\lambda)\ket{\sqrt{\lambda}\alpha}
	\bra{\sqrt{\lambda}\alpha}$, whose Glauber--Sudarshan $P$-function is the
	probability distribution $P(\beta)=\int d\lambda\,p(\lambda)\,
	\delta^{(2)}(\beta-\sqrt{\lambda}\alpha)$, supported on the segment
	$\{\sqrt{\lambda}\alpha:\lambda\in[0,1]\}$. A classical mixture of
	coherent states is Gaussian if and only if its $P$-function is a (possibly
	degenerate) Gaussian distribution: indeed its Wigner function is the
	convolution of the vacuum Wigner function with $P(\beta)$; since the
	Fourier transform of the vacuum Wigner function is a Gaussian that never
	vanishes, dividing by it shows that $P(\beta)$ is Gaussian whenever the
	Wigner function is. A Gaussian distribution with
	bounded support is necessarily a Dirac delta, hence $P(\beta)$, and thus
	$p(\lambda)$, must be a Dirac delta for $\rho_\alpha$ to be Gaussian. For
	any other $p(\lambda)$ the output $\rho_\alpha$ is non-Gaussian and so is
	$\Phi_{p(\lambda)}$; the case $\Phi_{\{p_n,\lambda_n\}}$ corresponds to
	$p(\lambda)=\sum_n p_n\delta(\lambda-\lambda_n)$.
\end{proof}

\noindent As we shall see in the following sections, this departure from the Gaussian framework is not merely a mathematical curiosity, but leads to qualitatively new behaviors, most notably the fact that thermal-state encodings can fail to achieve capacity.\\

\noindent We consider two representative distributions for $p(\lambda)$:
\begin{itemize}

    \item \textbf{Log-negative Weibull distribution}: a physically realistic fading model derived from wave propagation in turbulent media, often used in free-space quantum optics~\cite{Pirandola2021}.

    The probability distribution of the transmissivity $\lambda$ is modeled according to the log-negative Weibull distribution, which accounts for the fluctuations induced by atmospheric turbulence and beam wandering. Given the maximum transmissivity $\lambda_0$, the probability density function is defined as:
\begin{equation}
    p_{R,\gamma}(\lambda) = \frac{R^2}{\gamma \lambda} \left( \ln \frac{\lambda_0}{\lambda} \right)^{\frac{2}{\gamma} - 1} \exp \left[ -\frac{R^2}{2} \left( \ln \frac{\lambda_0}{\lambda} \right)^{\frac{2}{\gamma}} \right],
\end{equation}
for $0 < \lambda < \lambda_0\le 1$, and $p(\lambda) = 0$ otherwise. Here, $R = r_0/\sigma$ is the ratio between the scale parameter $r_0$ of the pointing loss, set by the receiver's aperture, and the standard deviation $\sigma$ of the beam wandering, and $\gamma$ is a shape parameter~\cite{Pirandola2021}.

    \item \textbf{Binary distribution}: $\Phi_{p,\lambda_1,\lambda_2} = p \mathcal{E}_{\lambda_1} + (1-p) \mathcal{E}_{\lambda_2}$, a physically simple yet meaningful model.
\end{itemize}

\noindent The binary distribution serves as a physically motivated model that captures the essential features of fluctuating media while remaining analytically tractable. Beyond its mathematical simplicity, this model effectively approximates several realistic communication scenarios, such as free-space links subject to intermittent obstructions, like clouds or foliage, that abruptly switch the channel between a high-transmissivity state and a high-loss regime. In the context of atmospheric science, it can further represent a simplified two-state turbulence model where the channel alternates between stable conditions and bursts of strong beam wandering. Specifically, when one of the components is a total loss channel ($\lambda_2=0$), the model reduces to the \emph{erasure-lossy channel}, which constitutes the fundamental benchmark for studying the breakdown of Gaussian-optimality. By focusing on this binary mixture, we can isolate the effects of statistical uncertainty in the channel parameters, providing a clear operational framework to test non-Gaussian strategies before addressing more complex continuous distributions.\\
Understanding the capacity of such channels requires dedicated analytical and numerical tools, which we develop in the next sections.

\begin{figure}[t]
	\centering
	
	$$ 
	\Qcircuit @C=1em @R=.7em {
		\lstick{\rho_S}	& \multigate{1}{U_{\lambda_1}} & \multigate{1}{U_{\lambda_2}} & \multigate{1}{\dots} & \multigate{1}{U_{\lambda_d}} & \qw & \rstick{\Phi_{\{p_n,\lambda_n\}}(\rho_S)} \\
		\lstick{\ket{0}_E}	& \ghost{U_{\lambda_1}} & \ghost{U_{\lambda_2}} & \ghost{\dots} & \ghost{U_{\lambda_d}} & \qw \\
		\lstick{\ket{\psi_p}_A}	& \ctrl{-1} & \ctrl{-1} & \ctrl{-1} & \ctrl{-1} & \qw \gategroup{2}{6}{3}{4}{0.8em}{\}}& \rstick{\tilde{\Phi}_{\{p_n,\lambda_n\}}(\rho_S)} \\
	}
	$$
	\caption{Quantum circuit representation of a fading channel with a discrete distribution $\{p_n,\lambda_n\}$, Eq.~\eqref{eq:discretized_fading_channel}. This gives us a graphical intuition for the structure of a complementary channel of a fading channel.}
	\label{fig:stinespring_dilation}
\end{figure}

\subsection{Complementary channel construction}\label{sec:complementary}

\noindent From the Stinespring dilation, any quantum channel $\Phi(\rho_S) = \text{Tr}_E[U (\rho_S \otimes \ket{0}_E\bra{0}) U^\dagger]$ has a corresponding complementary channel, $\tilde{\Phi}(\rho_S) = \text{Tr}_S[U (\rho_S \otimes \ket{0}_E\bra{0}) U^\dagger]$. This map $\tilde{\Phi}$ describes the information that is leaked to the environment $E$ instead of being transmitted to the receiver $S$. 
A crucial subtlety arises when we want to compute the complementary channel of a convex combination of channels as in our fading channel model (Eq.~\ref{eq:discretized_fading_channel}). The complementary channel of a convex combination is not, in general, the convex combination of the complementary channels. To find the correct complementary channel $\tilde{\Phi}$ for the fading channel, we must first construct a single, unified Stinespring dilation for the entire mixture $\Phi_{\{p_n,\lambda_n\}}$. This is achieved by introducing an ancillary ``control" system $A$ (with $\dim(A)=d$) which correlates the channel's action with the probability distribution $\{p_n\}$. We define an initial pure state for this control system:
\begin{equation}
	\ket{\psi_p}_A = \sum_{n=1}^{d} \sqrt{p_n}\ket{n}_A,
\end{equation}
where $\{\ket{n}_A\}$ is an orthonormal basis for $A$.\\
\noindent The global interaction $V$ is now a controlled unitary that couples the input system $S$, the vacuum environment $E$, and the control ancilla $A$. This operator is defined as:
\begin{equation}
	V = \sum_{n=1}^{d} U_{\lambda_n}^{(SE)} \otimes \ket{n}_A \bra{n},
\end{equation}
where $U_{\lambda_n}^{(SE)}$ is the beam-splitter unitary from Eq.~\eqref{eq:beam_splitter_def} with transmissivity $\lambda_n$. This unitary $V$ applies the specific beam-splitter interaction $U_{\lambda_n}$ if and only if the control ancilla is in the state $\ket{n}_A$.\\
\noindent The Stinespring dilation for the fading channel $\Phi_{\{p_n,\lambda_n\}}$ is therefore given by first preparing the joint state $\rho_S \otimes \ket{0}_E\bra{0}\otimes \ket{\psi_p}_A\bra{\psi_p}$, applying the global unitary $V$, and tracing out both the environment $E$ and the ancilla $A$:
\begin{equation}
	\Phi_{\{p_n,\lambda_n\}}(\rho_S) = \text{Tr}_{EA}\left[ V (\rho_S \otimes \ket{0}_E\bra{0}\otimes \ket{\psi_p}_A\bra{\psi_p}) V^\dagger \right].
\end{equation}
Consequently, the true complementary channel $\tilde{\Phi}_{\{p_n,\lambda_n\}}$ is a map that describes the information leaked to the total environment, which consists of both $E$ and $A$. We find it by tracing over the system $S$ instead:
\begin{equation}\label{eq:complementary}
	\tilde{\Phi}_{\{p_n,\lambda_n\}}(\rho_S) = \text{Tr}_{S}\left[ V (\rho_S \otimes \ket{0}_E\bra{0}\otimes \ket{\psi_p}_A\bra{\psi_p}) V^\dagger \right].
\end{equation}
This construction, which is essential for correctly calculating the exchange entropy $S(\tilde{\Phi}_{\{p_n,\lambda_n\}}(\rho))$ and for assessing the antidegradability region of $\Phi_{\{p_n,\lambda_n\}}$, is visualized in Fig.~\ref{fig:stinespring_dilation}.
			The same construction applies to an arbitrary distribution $p(\lambda)$ by replacing the
	$d$-dimensional control ancilla with a continuous register $A$ spanned by the improper basis
	$\{\ket{\lambda}_A\}_{\lambda\in[0,1]}$, $\braket{\lambda|\lambda'}=\delta(\lambda-\lambda')$.
	The register is not an additional bosonic mode: it merely purifies the classical randomness of
	the transmissivity. The control
	state and the controlled unitary become
	\begin{equation}
	\begin{aligned}
		& \ket{\psi_p}_A = \int_0^1 d\lambda\,\sqrt{p(\lambda)}\,\ket{\lambda}_A, \\
		& V = \int_0^1 d\lambda\, U_\lambda^{(SE)} \otimes \ket{\lambda}_A\bra{\lambda},
	\end{aligned}
	\end{equation}
	so that $V$ applies the beam splitter of transmissivity $\lambda$ conditioned on the value
	$\lambda$ stored in the register. For a discrete distribution the construction reduces to the
	finite-dimensional one above.

\section{Fundamental Properties of Fading Channels}\label{sec:fundamental}

\noindent In this section, we analyze the fundamental information-theoretic structure of bosonic fading channels. Beyond evaluating communication rates, it is crucial to establish the qualitative features of the map $\Phi$, as they dictate the relationships between different capacity measures and the limits of quantum error correction. 
Specifically, we investigate channel \emph{degradability}, which ensures the additivity of coherent information; \emph{antidegradability}, which sets a zero-threshold for quantum communication; and \emph{distillability}, which serves as a witness for non-zero two-way quantum capacity $Q_2$. As we will prove, convex combinations of pure-loss channels generally belong to a non-trivial class of maps: they are non-degradable and, under certain conditions, avoid antidegradability. Most notably, we prove that bosonic fading channels are distillable for every transmissivity distribution, with the single exception of the complete-loss channel. This fundamental property implies that the two-way quantum capacity $Q_2$ remains strictly positive even for distributions concentrated around arbitrarily low transmissivity values. In addition, for a finite set of transmissivities,\ Sec.~\ref{sec:improved Q2} determines $Q_2$ and the secret-key capacity exactly when the input energy is unconstrained, and gives an explicit rate at a fixed energy through a multi-rail protocol. The picture that emerges is one of fundamental resilience of entanglement distribution over fading channels.

\subsection{Region of degradability}

\noindent A quantum channel $\Phi$ is defined as degradable if and only if there exists a degrading channel $\mathcal{D}$ such that:
\begin{equation}\label{eq:degradability_def}
	\tilde{\Phi} = \mathcal{D} \circ \Phi,
\end{equation}
where $\tilde{\Phi}$ is the complementary channel of $\Phi$.
Let us define the subspace $\mathcal{H}_d = \text{span}\{\ket{0},\dots,\ket{d-1}\}$ generated by the first $d$ Fock states. By identifying with $\sigma(\mathcal{H})$ the space of the quantum states defined on the Hilbert space $\mathcal{H}$ we can define the restriction $\Phi_d$ of $\Phi$ to $\sigma(\mathcal{H}_d)$ such that:
\begin{equation}
	\Phi_d(\rho) = \Phi(\rho) \quad \forall \rho \in \sigma(\mathcal{H}_d).
\end{equation}
Similarly, we define the restricted complementary channel $\tilde{\Phi}_d$ such that:
\begin{equation}
	\tilde{\Phi}_d(\rho) = \tilde{\Phi}(\rho) \quad \forall \rho \in \sigma(\mathcal{H}_d).
\end{equation}
If the original channel $\Phi$ is degradable, then necessarily all its restrictions must be degradable. This is true for any channel: restricting a Stinespring isometry of $\Phi$ to $\mathcal{H}_d$ gives a dilation of $\Phi_d$ whose complementary map is $\tilde{\Phi}_d$, so the same degrading channel $\mathcal{D}$ satisfies $\tilde{\Phi}_d = \mathcal{D} \circ \Phi_d$. What is specific to convex combinations of pure-loss channels is that they cannot increase the photon number, so that $\Phi$ maps $\sigma(\mathcal{H}_d)$ into itself; the restriction $\Phi_2$ is therefore a genuine qubit-to-qubit channel, to which the Choi-rank criterion below applies. In particular, considering the qubit restriction ($d=2$), there must exist a map $\mathcal{D}_2$ satisfying:
\begin{equation}
	\tilde{\Phi}_2 = \mathcal{D}_2 \circ \Phi_2.
\end{equation}
This condition implies a strict bound on the rank of the corresponding Choi matrix $C_{\Phi_2}$: specifically, for a qubit channel to be degradable, it must hold that $\text{rank}\left( C_{\Phi_2} \right) \le 2$~\cite{Cubitt2008}. Conversely, finding $\text{rank}\left( C_{\Phi_2} \right) > 2$ is sufficient to prove that the full channel $\Phi$ is non-degradable.\\
\noindent We now apply this criterion to a fading channel $\Phi_{p(\lambda)} = \int_0^1 d\lambda\, p(\lambda)\, \mathcal{E}_{\lambda}$ with arbitrary $p(\lambda)$. We consider its action on the qubit basis states of $\mathcal{H}_2$:
\begin{align}
    \Phi_{p(\lambda)}(\ket{0}\bra{0}) &= \ket{0}\bra{0}, \\
    \Phi_{p(\lambda)}(\ket{1}\bra{1}) &= (1-\braket{\lambda})\ket{0}\bra{0} + \braket{\lambda} \ket{1}\bra{1}, \\
    \Phi_{p(\lambda)}(\ket{0}\bra{1}) &= \braket{\sqrt{\lambda}} \ket{0}\bra{1},
\end{align}
where $\langle f(\lambda) \rangle := \int_0^1 d\lambda\, p(\lambda)\, f(\lambda)$ denotes the average over the fading distribution; for a discrete distribution it reads $\langle f(\lambda) \rangle = \sum_n p_n f(\lambda_n)$.
The resulting unnormalized Choi matrix $C_{\Phi_2} = (\mathcal{I} \otimes \Phi)(\ket{\varphi^+}\bra{\varphi^+})$, with $\ket{\varphi^+} = \ket{00} + \ket{11}$ is:
\begin{equation}
    C_{\Phi_2} =
    \begin{pmatrix}
        1 & 0 & 0 & \braket{\sqrt{\lambda}} \\
        0 & 0 & 0 & 0 \\
        0 & 0 & 1 - \braket{\lambda} & 0 \\
        \braket{\sqrt{\lambda}} & 0 & 0 & \braket{\lambda}
    \end{pmatrix}.
\end{equation}
The eigenvalues of $C_{\Phi_2}$ are $\{0, 1-\braket{\lambda}\}$ together with the eigenvalues of the $2 \times 2$ submatrix 
\begin{equation}
    M = \begin{pmatrix} 1 & \braket{\sqrt{\lambda}} \\ \braket{\sqrt{\lambda}} & \braket{\lambda} \end{pmatrix}.
\end{equation}
The determinant of this submatrix is:
\begin{equation}
    \det(M) = \braket{\lambda} - \braket{\sqrt{\lambda}}^2.
\end{equation}
Since $f(x) = \sqrt{x}$ is strictly concave, Jensen's inequality implies $\braket{\sqrt{\lambda}} \le \sqrt{\braket{\lambda}}$, and consequently $\braket{\sqrt{\lambda}}^2 \le \braket{\lambda}$.
Crucially, equality in Jensen's inequality holds if and only if the distribution $p(\lambda)$ is a delta function, corresponding to a single Gaussian pure-loss channel. For any non-trivial fading distribution, the inequality is strict, implying $\det(M) > 0$ and thus $\text{rank}(M)=2$. Since for $\braket{\lambda} < 1$ the third eigenvalue $1-\braket{\lambda}$ is also strictly positive, we find:
\begin{equation}
    \text{rank}(C_{\Phi_2}) = 3 > 2.
\end{equation}
This violation proves that bosonic fading channels are generally non-degradable. Consequently, among all fading channels $\Phi_{p(\lambda)}$, a fading channel is never degradable unless $p(\lambda)$ is a Dirac delta. When $p(\lambda) = \delta(\lambda-\lambda_0)$ the channel reduces to the pure-loss channel $\mathcal{E}_{\lambda_0}$, which is degradable if and only if $\lambda_0 \ge 1/2$~\cite{Caruso_weak}. The standard proof of additivity for the coherent information does not apply here. While $Q_1$ remains a valid lower bound for the quantum capacity $Q$, the non-degradable nature of these channels leaves open the possibility of super-additive effects, suggesting that non-Gaussian strategies might benefit from multi-letter protocols or that $Q_1$ may strictly underestimate the ultimate transmission rate.

\begin{figure*}[t]
    \centering
        \includegraphics[height=8cm, keepaspectratio]{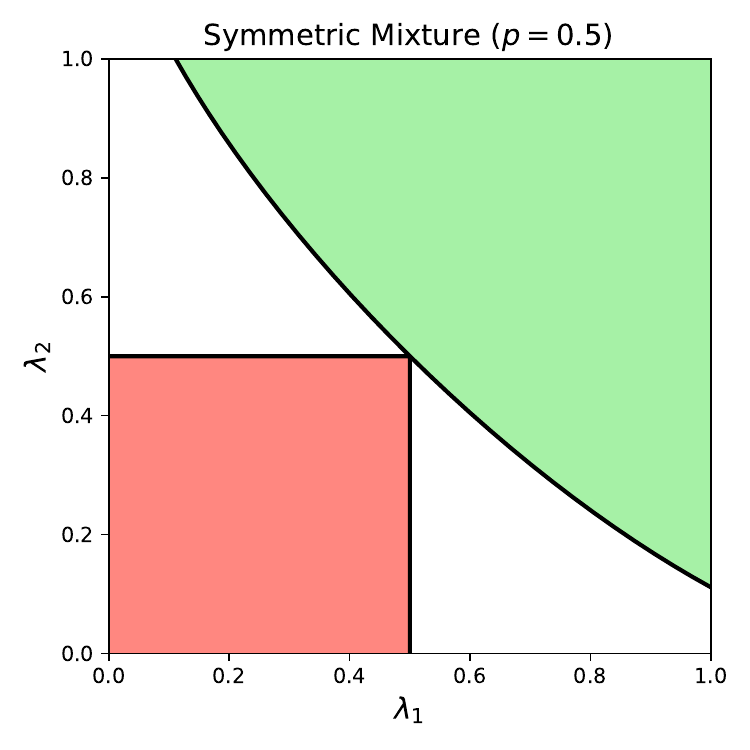}
    \hfill
    \includegraphics[height=8cm, keepaspectratio]{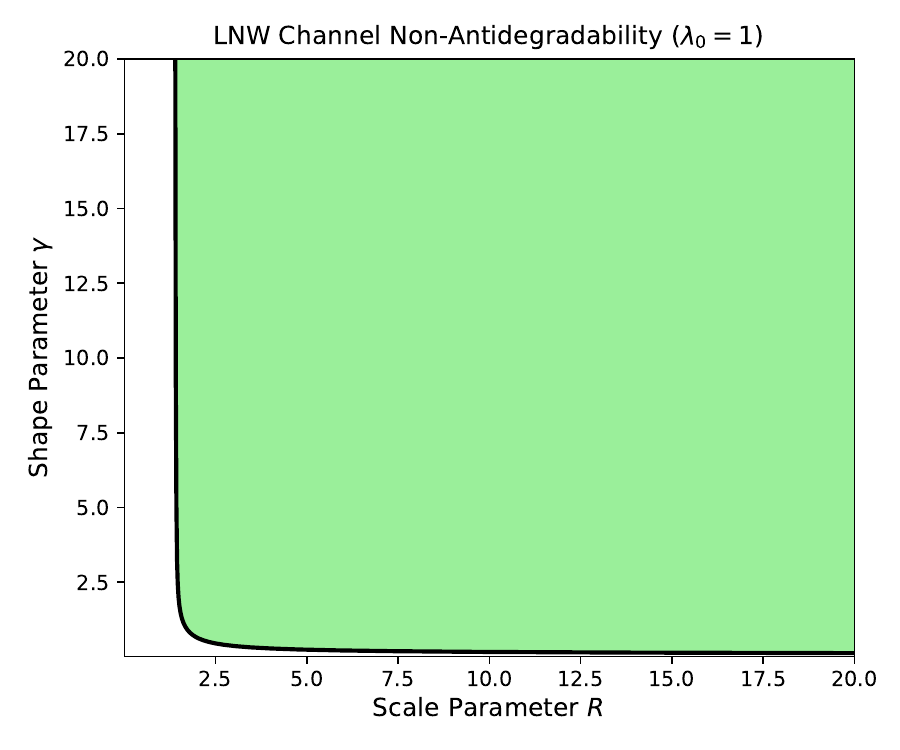}
    \caption{Regions of Antidegradability and Non-Antidegradability. Left: Parameter space $(\lambda_1, \lambda_2)$ for the symmetric binary fading channel ($p=0.5$). 
    The green area indicates the region where the channel is non-antidegradable. 
    The red square ($\lambda_1, \lambda_2 \leq 0.5$) highlights the region of antidegradability consistent with the theoretical property that the set of all antidegradable channels is convex.
    Right: Non-antidegradable region (green area) for the atmospheric fading channel modeled by the Log-Negative Weibull distribution, as a function of the scale ($R$) and shape ($\gamma$) parameters (with cutoff $\lambda_0 = 1$).}
    \label{fig:antideg_main}
\end{figure*}

\subsection{Antidegradability Analysis}\label{sec:antideg}

\noindent A channel $\Phi$ is defined as \emph{antidegradable} if its output can be simulated by applying a CPTP map $\mathcal{A}$ to its complementary output $\tilde{\Phi}$, i.e., $\Phi = \mathcal{A} \circ \tilde{\Phi}$~\cite{WildeBook}. A crucial consequence is that antidegradable channels have zero quantum capacity, $Q(\Phi) = 0$, as the environment retains at least as much information as the receiver~\cite{HolevoGiovannetti2012}.
Similarly as done in~\cite{Mele2024LossDephasing}, we establish a sufficient condition for non-antidegradability based on the data processing inequality for quantum fidelity. Specifically, if there exist two input states $\rho_1, \rho_2$ such that the fidelity of the complementary outputs is strictly greater than the fidelity of the main outputs, i.e.,
\begin{equation}\label{eq:fidelity_criterion}
    F(\tilde{\Phi}(\rho_1), \tilde{\Phi}(\rho_2)) > F(\Phi(\rho_1), \Phi(\rho_2)),
\end{equation}
then the channel $\Phi$ cannot be antidegradable (see Appendix~\ref{app:antidegradability_proof} for details).\\
\noindent We apply this criterion to a fading channel $\Phi_{p(\lambda)}$ with arbitrary $p(\lambda)$ by choosing the Fock states $\rho_1 = \ket{0}\bra{0}$ and $\rho_2 = \ket{1}\bra{1}$ as probes. As derived in Appendix~\ref{app:antidegradability_proof}, the fidelities for these inputs are given by:
\begin{align}
    F(\Phi(\rho_1), \Phi(\rho_2)) &= 1 - \braket{\lambda}, \\
    F(\tilde{\Phi}(\rho_1), \tilde{\Phi}(\rho_2)) &= \braket{\sqrt{\lambda}}^2.
\end{align}
Substituting these into Eq.~\eqref{eq:fidelity_criterion} yields a simple analytical condition for non-antidegradability based on the first two fractional moments of the fading distribution:
\begin{equation}\label{eq:non_antideg_condition}
    \braket{\sqrt{\lambda}}^2 + \braket{\lambda} > 1.
\end{equation}
If this inequality holds, the channel is guaranteed to be non-antidegradable, a necessary requirement for positive quantum capacity.
The same condition does more than exclude antidegradability: it guarantees a strictly positive quantum capacity, witnessed by weak thermal inputs. For $0 \le N \le 1$, let $\rho_N = (1-N)\ket{0}\bra{0} + N\ket{1}\bra{1}$, and let $\rho_{\text{th},N}$ be the thermal state with the same mean photon number $N$.
\begin{lemma}\label{lem:weak_inputs}
For every fading channel $\Phi_{p(\lambda)}$, as $N \to 0$,
\begin{align}
I_c(\Phi_{p(\lambda)}, \rho_N) &= \Delta F\, N \log_2\frac{1}{N} + O(N), \nonumber\\
I_c(\Phi_{p(\lambda)}, \rho_{\text{th},N}) &= \Delta F\, N \log_2\frac{1}{N} + O(N), \label{eq:weak_inputs}
\end{align}
where $\Delta F = \braket{\sqrt{\lambda}}^2 + \braket{\lambda} - 1$.
\end{lemma}
The proof is given in Appendix~\ref{app:weak_inputs}. The coefficient $\Delta F = F(\tilde{\Phi}(\rho_1), \tilde{\Phi}(\rho_2)) - F(\Phi(\rho_1), \Phi(\rho_2))$ is the gap between the two fidelities in Eq.~\eqref{eq:fidelity_criterion} for the probes $\rho_1 = \ket{0}\bra{0}$ and $\rho_2 = \ket{1}\bra{1}$: the same probes that witness non-antidegradability carry a positive coherent information when $\ket{1}$ is mixed in with a small weight. This is an instance of the mechanism identified in Ref.~\cite{Siddhu2021EntropicSingularities}: at $N = 0$ the entropies of $\Phi(\rho_N)$ and $\tilde{\Phi}(\rho_N)$ have logarithmic singularities, with rates $\braket{\lambda}$ and $1 - \braket{\sqrt{\lambda}}^2$, and a positive coherent information follows when the first is the stronger one. Hence, whenever Eq.~\eqref{eq:non_antideg_condition} holds, thermal states of sufficiently low mean photon number have $I_c > 0$, so that $Q \ge Q_1 > 0$ under any energy constraint. When $\braket{\sqrt{\lambda}}^2 + \braket{\lambda} < 1$, weak thermal inputs have instead a negative coherent information.
\subsubsection{Applications}
\noindent We evaluate the non-antidegradability condition~\eqref{eq:non_antideg_condition} for the two most relevant classes of fading channels: discrete binary mixtures and continuous atmospheric distributions.

\paragraph{Two-Component Fading}
Consider the convex combination of two lossy channels, $\Phi = p \mathcal{E}_{\lambda_1} + (1-p) \mathcal{E}_{\lambda_2}$. The condition~\eqref{eq:non_antideg_condition} takes the explicit form:
\begin{equation}
    p(p+1)\lambda_1 + (1-p)(2-p)\lambda_2 + 2p(1-p)\sqrt{\lambda_1 \lambda_2} > 1.
\end{equation}
This inequality identifies a region of the $(p,\lambda_1,\lambda_2)$ parameter space in which the quantum capacity is strictly positive (Lemma~\ref{lem:weak_inputs}).
In the limit of the \emph{erasure-like channel} ($\lambda_1=1, \lambda_2=0$), the condition simplifies to $p^2+p-1 > 0$. Our sufficient condition therefore certifies that the channel is non-antidegradable when the transmission probability exceeds the inverse golden ratio, $p > (\sqrt{5}-1)/2 = 1/\phi \approx 0.618$. The condition is only sufficient, and for this channel it is not tight: as shown in Sec.~\ref{sec:quantum-capacity}, every input orthogonal to the vacuum has $I_c = (2p-1)S(\rho)$, so the channel is non-antidegradable for every $p > 1/2$.\\
\noindent For the binary distribution case, we show in Fig.~\ref{fig:antideg_main} (left) the valid region for a symmetric mixture ($p=0.5$). The green area corresponds to non-antidegradable configurations, the red one collects the configurations for which our sufficient condition does not apply; among these, the channel is provably antidegradable in the square $\lambda_1, \lambda_2 \le 1/2$, by convexity of the set of antidegradable channels. A detailed analysis of other specific configurations is provided in Appendix~\ref{app:binary_antideg}.
\paragraph{Atmospheric Turbulence (LNW)}
For free-space optical links, the transmissivity fluctuations are described by the Log-Negative Weibull (LNW) distribution $p_{R,\gamma}(\lambda)$~\cite{Pirandola2021}.
We computed the moments $\braket{\lambda}$ and $\braket{\sqrt{\lambda}}$ over the physical parameter space $(R, \gamma)$ from an exact integral representation (App.~\ref{app:general_fading}). The result is shown in Fig.~\ref{fig:antideg_main} (right). The green region highlights the regime where the channel is proven to be non-antidegradable. By Lemma~\ref{lem:weak_inputs}, the quantum capacity is strictly positive throughout the green region, and weak thermal inputs already achieve a positive coherent information: unassisted quantum communication is possible over a wide range of realistic turbulence conditions, even beyond the weak-fluctuation regime.

\subsection{Distillability and PPT Violation}\label{section: distillability}

\noindent We call a quantum channel $\Phi$ distillable if its two-way assisted quantum capacity $Q_2(\Phi)$~\cite{PLOB} is strictly positive. A sufficient condition to prove distillability of a quantum channel is to find any input state $\rho_{RA}$ such that the channel's output state, $\rho_{RB} = (\mathcal{I}_R \otimes \Phi_A)(\rho_{RA})$, is distillable. As proved in~\cite{2-qubit-distillation}, there is a necessary and sufficient condition for distillability of a two-qubit state: a two qubit state is distillable if and only if it is not PPT~\cite{Peres1996PPT, Horodecki1996PPT}.

\noindent We apply this test using a generic pure entangled two-qubit state as the probe:
\begin{equation}
    \ket{\psi}_{RA} = c_0 \ket{00} + c_1 \ket{11},
\end{equation}
with real coefficients satisfying $c_0^2 + c_1^2 = 1$ and $c_0, c_1 \neq 0$ to ensure entanglement. The fading channel $\Phi$ acts on the system $A$ yielding the output density matrix $\rho_{RB}$. Importantly, the output system $B$ is also spanned by $\ket{0}$ and $\ket{1}$, as implied by the structure of the fading channel; thus $\rho_{RB}$ remains a two-qubit state.\\
\noindent We evaluate the partial transpose $\rho_{RB}^{T_R}$ with respect to the reference system $R$. If $\Phi$ is a fading channel, the matrix $\rho_{RB}^{T_R}$ takes the form:
\begin{equation}
    \rho_{RB}^{T_R} = 
    \begin{pmatrix}
        c_0^2 & 0 & 0 & 0 \\
        0 & c_1^2(1-\braket{\lambda}) & c_0 c_1 \braket{\sqrt{\lambda}} & 0 \\
        0 & c_0 c_1 \braket{\sqrt{\lambda}} & 0 & 0 \\
        0 & 0 & 0 & c_1^2 \braket{\lambda}
    \end{pmatrix}.
\end{equation}
The matrix is block-diagonal. Its positivity depends on the eigenvalues of the central $2 \times 2$ block acting on the subspace spanned by $\{\ket{01}, \ket{10}\}$:
\begin{equation}
    M' = \begin{pmatrix}
    c_1^2(1-\braket{\lambda}) & c_0 c_1 \braket{\sqrt{\lambda}} \\
    c_0 c_1 \braket{\sqrt{\lambda}} & 0
    \end{pmatrix}.
\end{equation}
The determinant of this submatrix is:
\begin{equation}
    \det(M') = - c_0^2 c_1^2 \braket{\sqrt{\lambda}}^2.
\end{equation}
For any entangled input ($c_0, c_1 \neq 0$) and any non-trivial fading channel (where $\braket{\sqrt{\lambda}} > 0$), this determinant is strictly negative.

\noindent Since a matrix with a negative determinant must possess at least one negative eigenvalue, we have $\rho_{RB}^{T_R} \ngeq 0$. This confirms that the output state is NPT (Non-Positive Partial Transpose). The generation of an NPT state proves that the fading channel $\Phi$ is distillable. Consequently, its two-way assisted quantum capacity is strictly positive:
\begin{equation}
    Q_2(\Phi) > 0.
\end{equation}
This result holds for any convex combination of lossy channels except the complete-loss channel and ensures that the channel is capable of entanglement distribution if assisted by classical communication.

\begin{figure*}[t]
    \centering
     \includegraphics[width=0.49\linewidth]{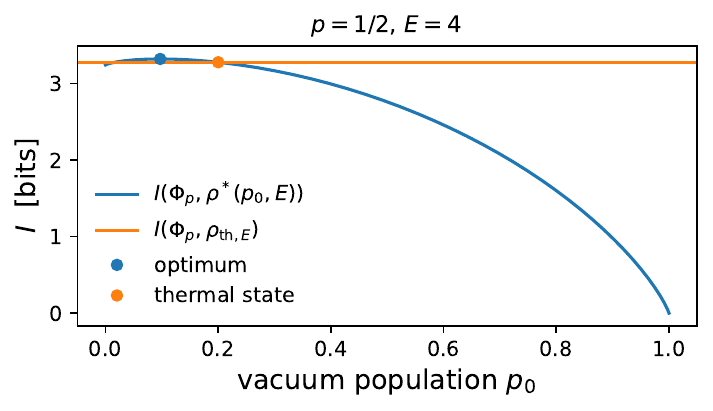}
    \hfill
    \includegraphics[width=0.49\linewidth]{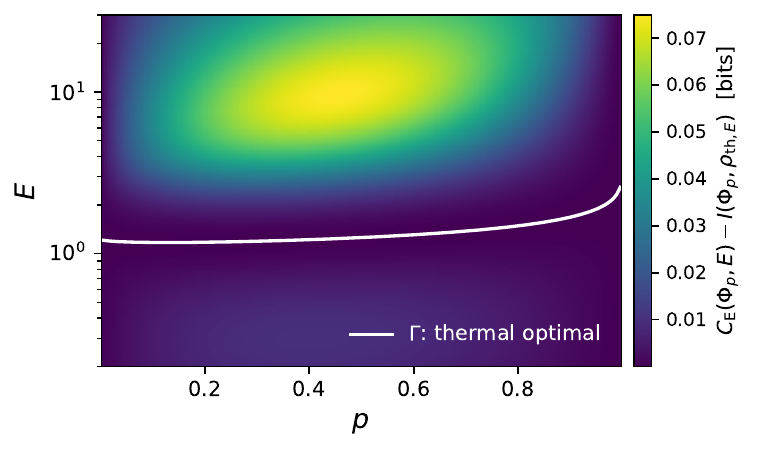}
    \caption{\textbf{Left:} Quantum mutual information of the erasure-like channel $I(\Phi_p, \rho^*(p_0, E))$
as a function of the vacuum population $p_0$, for the erasure-like channel 
$\Phi_p = p\mathcal{I} + (1-p)\mathcal{E}_0$ with fixed $p=0.5$ and $E=4$. 
The input state $\rho^*(p_0, E) = p_0\ket{0}\bra{0} + (1-p_0)\sigma_N$ is the
optimal ansatz of Eq.~\eqref{eq:optimal_ansatz}, parametrized by the vacuum 
population $p_0$. The orange dot marks the thermal state 
($p_0^{\mathrm{th}} = 1/(E+1) = 0.2$), while the blue dot marks the true 
optimum ($p_0^{\mathrm{opt}} \approx 0.1$), demonstrating that suppressing 
the vacuum component below the thermal value strictly increases the mutual 
information. \textbf{Right:} Absolute capacity gain 
$C_\text{E}(\Phi_p, E) - I(\Phi_p, \rho_{\mathrm{th},E})$ in the $(p, E)$ plane, 
quantifying the advantage of the optimal non-Gaussian state over the thermal 
benchmark across the full parameter space of the erasure-like channel. The energy axis is logarithmic. The gain is strictly positive everywhere except on the white curve $\Gamma$, on which the thermal state is stationary and, numerically, attains the capacity; all other Gaussian states remain strictly sub-optimal. The map is computed from the closed form of App.~\ref{app:erasure_derivation}.}
    \label{fig:erasure_optimization}
\end{figure*}

\subsection{Two-Way Capacities of Fading Channels}\label{sec:improved Q2}

\noindent When the transmissivity takes a finite set of values, the two-way capacities of a fading channel can be determined exactly. We write $c(\lambda) = -\log_2(1-\lambda)$ for the two-way quantum and secret-key capacity of the pure-loss channel $\mathcal{E}_\lambda$~\cite{PLOB} and $K(\cdot)$ for the secret-key capacity, not to be confused with the photon number $K$ of the multi-rail code below.

\begin{theorem}\label{thm:exact_two_way}
Let $\Phi = \sum_{n=1}^{d} p_n \mathcal{E}_{\lambda_n}$ be a fading channel with finitely many transmissivities $\lambda_n \in [0,1]$ and $p_n > 0$. Assume that the channel is memoryless and that the realization is revealed to neither party. Then, with no constraint on the input energy,
\begin{equation}\label{eq:exact_two_way}
    Q_2(\Phi) = K(\Phi) = \sum_{n=1}^{d} p_n\, c(\lambda_n),
\end{equation}
which is infinite if some $\lambda_n = 1$.
\end{theorem}

Here, as in Ref.~\cite{PLOB} (Supplementary Note 7), the capacities without an energy constraint are the limits $\bar{N}\to\infty$ of the capacities $Q_2(\Phi,\bar{N})$ and $K(\Phi,\bar{N})$ at mean input energy $\bar{N}$ per channel use. The two-way capacities of a finite mixture of pure-loss channels are thus the corresponding mixtures of the pure-loss capacities. The converse is known: the bound $K(\Phi) \le \sum_n p_n c(\lambda_n)$ follows from simulating the fading channel by its components~\cite{PLOB, PLB19} and expresses that a receiver who was told which component had acted could do no better than the average of the component capacities; Ref.~\cite{PLB19} also covers continuous ensembles. The same work derives the lower bound $\sum_n p_n c(\lambda_n) - H(\{p_n\})$ from the reverse coherent information of the average channel, where $H$ is the Shannon entropy, and therefore leaves a gap equal to the entropy of the fading distribution. Theorem~\ref{thm:exact_two_way} closes that gap: the parties can learn which component acted, and the receiver announces what he has learnt, so that both parties hold a copy of the label. Its entropy then appears in the two entropies of the reverse coherent information alike and cancels, whereas in Ref.~\cite{PLB19} the label is not available at all, and the entropy of the mixture exceeds the average entropy of the components by up to $H(\{p_n\})$.

The component can be identified because, on a displaced input, each component imprints its own, readable output displacement. Since
\begin{equation}\label{eq:loss_displacement}
    \mathcal{E}_\lambda\big[D(\alpha)\,\rho\,D^\dagger(\alpha)\big] = D(\sqrt{\lambda}\,\alpha)\, \mathcal{E}_\lambda(\rho)\, D^\dagger(\sqrt{\lambda}\,\alpha),
\end{equation}
displacing the transmitted mode by a large amplitude $\alpha$ turns the identity of the component into an output amplitude: the amplitudes $\sqrt{\lambda_n}\,\alpha$ separate linearly in $\alpha$, while the noise against which they are read does not depend on $\alpha$. The receiver taps an arbitrarily small fraction of the received light, homodynes it to identify the component, removes the corresponding known displacement from the mode he keeps, and distils entanglement from what remains. The displacement is a local unitary and carries no entanglement, so identifying the component costs only the vanishing loss of the tap. Appendix~\ref{subsec:two_way_discrete} gives the protocol and the proof. The covariance property that makes this work is the one that Ref.~\cite{PLB19} identifies as an obstruction, because the required correction depends on the component and the components are therefore not jointly teleportation covariant; on the side of achievability that same dependence is what carries the information.

The scope of Theorem~\ref{thm:exact_two_way} is limited in three ways. First, the result relies on free classical communication. In our protocol it flows only from the receiver to the sender, so that entanglement can be distributed at the rate of Eq.~\eqref{eq:exact_two_way} with backward communication alone; transmitting quantum information then uses the free forward communication for teleportation, and the unassisted quantum capacity is not determined by this argument. Second, Eq.~\eqref{eq:exact_two_way} gives the capacities in the limit of unbounded input energy, and our protocol approaches it at a high energy cost. The displacement it requires diverges both when the transmissivities get close and when the rate is pushed towards the capacity (Appendix~\ref{subsec:two_way_discrete}). At a fixed energy per channel use the theorem gives only an upper bound, and the multi-rail protocol summarized below gives an explicit rate, which, unlike Theorem~\ref{thm:exact_two_way}, applies to every $p(\lambda)$. Third, the result concerns a finite set of transmissivities. For a continuous $p(\lambda)$ our argument does not extend. The upper bound $K(\Phi) \le \int d\lambda\, p(\lambda)\, c(\lambda)$ of Ref.~\cite{PLB19} still holds, but whether it is attainable is open. The rest of this subsection summarizes a protocol that gives an explicit achievable rate in both of these regimes.

\noindent At a fixed mean photon number $\bar{N}$ per channel use, Theorem~\ref{thm:exact_two_way} gives nothing beyond $Q_2(\Phi,\bar{N}) \le Q_2(\Phi)$ for the energy-constrained capacity $Q_2(\Phi,\bar{N})$. An explicit lower bound on the latter, valid for every $p(\lambda)$, follows from a multi-rail protocol~\cite{Winnel,Mele2024LossDephasing,MeleNature}, which we derive in Appendix~\ref{app:multirail_bound}: Alice encodes half of a maximally entangled state in the subspace $\mathcal{H}_{K,M}$ of $K$ photons distributed over $M$ modes, and Bob keeps a block only if he detects all $K$ photons. Detecting $K$ photons excludes photon losses, but it does not leave the encoded state intact: for a realization $\lambda_1,\dots,\lambda_M$ of the transmissivities the no-loss branch multiplies each configuration $\ket{\vec{n}}$ by $\prod_{i=1}^M \lambda_i^{n_i/2}$, which is not constant on $\mathcal{H}_{K,M}$ unless all $\lambda_i$ coincide, and, the realization being unknown to Alice and Bob, the accepted state is mixed. It is maximally correlated, so that its distillable entanglement equals its coherent information~\cite{LeditzkyDattaSmith2018} and is achieved by the hashing protocol~\cite{DevetakWinter2005}. For a single photon on two rails ($K=1$, $M=2$) the accepted state is a maximally entangled two-qubit state that has undergone a dephasing, and
\begin{equation}\label{eq:Q2_mr_dualrail}
    Q_2(\Phi,\bar{N}) \ge \frac{\langle \lambda \rangle}{2} \left[ 1 - H_2\!\left( \frac{1+c}{2} \right) \right] \qquad \text{for } \bar{N} \ge \tfrac{1}{2},
\end{equation}
where $c = \langle \sqrt{\lambda} \rangle^2 / \langle \lambda \rangle$ and $H_2$ is the binary entropy: the fluctuations of the transmissivity act as an effective dephasing between the rails, of strength $1-c = \text{Var}(\sqrt{\lambda})/\langle \lambda \rangle$. Since $c \ge \langle \lambda \rangle$, the rate is strictly positive whenever $\langle \lambda \rangle > 0$, and a single photon on $M \ge \max\{2, 1/\bar{N}\}$ rails keeps it so at every $\bar{N} > 0$, in agreement with the distillability result of Sec.~\ref{section: distillability}.

\section{Breaking Gaussian Optimality: The Erasure-like Channel}
\label{sec:erasure_channel}

\noindent Having established the qualitative properties of bosonic fading channels, we now turn our attention to the quantitative evaluation of their communication capacities. A central question in continuous-variable quantum information is whether Gaussian states, which are optimal for the entanglement-assisted capacity of the static pure-loss channel, remain the best possible resources for information transmission in the presence of fluctuations. While the general case requires the numerical optimization developed in the next sections, we can analytically prove the suboptimality of the thermal states by examining a simplified but physically insightful model: the continuous-variable erasure-like channel. This channel represents the extreme limit of binary fading, where the signal is either transmitted without error with probability $p$, or completely lost and replaced by the vacuum state with probability $1-p$. By focusing on this model, we can rigorously demonstrate that there exist regimes where non-Gaussian Fock-diagonal states strictly outperform their Gaussian counterparts. 

\subsection{Channel Model and Physical Motivation}
\noindent We consider a binary fading channel where the transmissivity fluctuates between perfect transmission ($\lambda=1$) and complete loss ($\lambda=0$). The resulting channel map is a convex combination of the identity $\mathcal{I}$ and the vacuum-replacement channel $\mathcal{E}_0$:
\begin{equation}
    \Phi_p(\rho) = p \rho + (1-p) \text{Tr}[\rho] \ket{0}\bra{0},
\end{equation}
where $p \in [0,1]$ represents the probability of successful transmission. Physically, this models a link subject to severe intermittency (e.g., deep fading or on-off keying), where the signal is either received intact or completely lost to the environment.

\noindent Crucially, unlike the standard finite-dimensional erasure channel~\cite{Bennett1997Erasure}, the ``erasure flag" here is the vacuum state $\ket{0}$, which is not orthogonal to the input signal states (unless the input is strictly supported on $n \ge 1$). This lack of orthogonality induces ambiguity: receiving a vacuum state could mean the signal was erased, or that the input signal has overlap with the vacuum. Gaussian inputs (thermal states) have significant overlap with the vacuum, making them potentially suboptimal for distinguishing these events. For this reason we call $\Phi_p$ an \emph{erasure-like} channel throughout.

\begin{figure*}[t]
	\centering
	\includegraphics[width=0.98\linewidth]{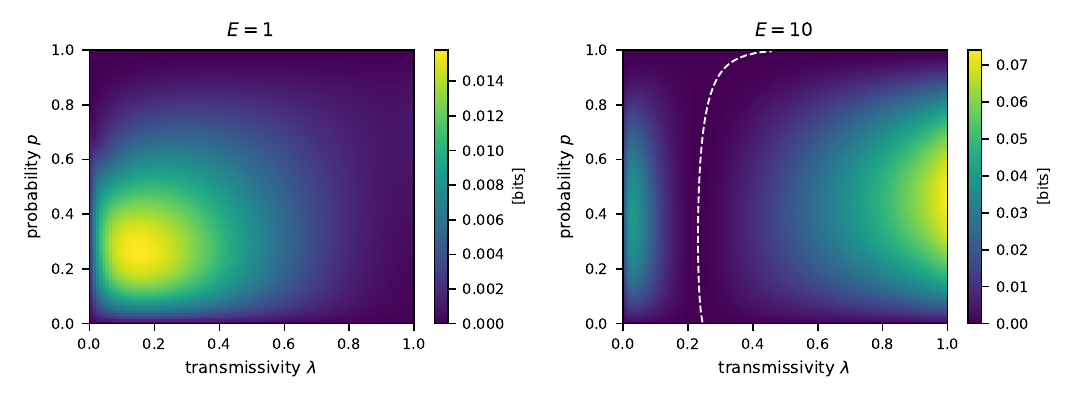}
    \caption{Achievable quantum mutual information gain over the thermal benchmark, $I(\Phi_{p,\lambda}^{(0)}, \rho^{(1)}) - I(\Phi_{p,\lambda}^{(0)}, \rho_{\mathrm{th}})$, 
for the erasure-lossy channel $\Phi_{p,\lambda}^{(0)} = p\mathcal{E}_{\lambda} + 
(1-p)\mathcal{E}_0$, where $\rho^{(1)} = p_0\ket{0}\bra{0} + (1-p_0)\sigma_N$ 
is the first-order variational state with vacuum population $p_0$ optimized 
over $[0,1]$. The gain is mapped across the $(p, \lambda)$ parameter space 
for low energy ($E=1$, left) and high energy ($E=10$, right). In the right panel the gain vanishes on the dashed curve, on which the thermal state is stationary within the first-order family; there, second-order states still outperform it (App.~\ref{app:plots_gallery}).}
	\label{fig:gain_erasure_lossy}
\end{figure*}

\subsection{Analytical Solution for $C_\text{E}$}

\noindent We seek to compute the entanglement-assisted classical capacity $C_\text{E}(\Phi_p, E)$ of the erasure-like channel under a mean energy constraint $E$. This fundamental limit is given by the maximum of the quantum mutual information over all valid input states:
\begin{equation}
    C_\text{E}(\Phi_p, E) = \max_{\rho:\text{Tr}[a^\dagger a \rho] \le E} I(\Phi_p, \rho),
\end{equation}
where $I(\Phi_p, \rho)$ is defined as
\begin{equation}
    I(\Phi_p, \rho) = S(\rho) + S(\Phi_p(\rho)) - S(\tilde{\Phi}_p(\rho)),
\end{equation}
with $S(\cdot)$ denoting the von Neumann entropy and $\tilde{\Phi}_p$ being the complementary channel of $\Phi_p$.

\noindent While the complete mathematical derivation is deferred to Appendix~\ref{app:erasure_derivation}, we outline here the core analytical argument that allows us to exactly solve this optimization. The first crucial simplification arises from the phase-covariance of the fading map. As proven in Appendix~\ref{appendix: diagonal optimizer}, this symmetry guarantees that the optimal input state $\rho$ is diagonal in the Fock basis without loss of generality. To determine the exact optimal population distribution $\{p_n\}$, we perform the maximization in two distinct steps. We begin by isolating the population of the vacuum state, $p_0$, and parametrizing the generic Fock-diagonal input state as:
\begin{equation}
    \rho = p_0 \ket{0}\bra{0} + (1-p_0) \sigma,
\end{equation}
where $\sigma$ is a normalized density operator supported entirely on the subspace orthogonal to the vacuum (i.e., $\text{span}\{\ket{n}\}_{n \ge 1}$). The global energy bound $\text{Tr}[a^\dagger a \rho] \le E$ consequently imposes a rescaled energy constraint on $\sigma$, requiring its mean photon number $N_{\sigma} = \text{Tr}[a^\dagger a \sigma] $ to satisfy $N_{\sigma} \le E/(1-p_0)$.\\
\noindent By inspecting the structure of the quantum mutual information for this channel (detailed in Appendix~\ref{app:erasure_derivation}), we observe a remarkable feature: for a fixed vacuum probability $p_0$, $I(\Phi_p, \rho)$ depends on the state $\sigma$ exclusively through its von Neumann entropy $S(\sigma)$, and is a strictly monotonically increasing function of it. Therefore, maximizing the mutual information is mathematically equivalent to maximizing the entropy of $\sigma$ subject to its specific energy and support constraints.\\
\noindent From standard statistical mechanics, the unique state that maximizes entropy for a given mean energy is the thermal state. In our case, because $\sigma$ is strictly confined to the positive photon-number subspace, it must take the form of a thermal state shifted away from the vacuum. This implies that the capacity-achieving state $\rho^*$ consists of a discrete vacuum component and a shifted thermal tail:
\begin{equation}\label{eq:optimal_ansatz}
    \rho^*(p_0,E) = p_0 \ket{0}\bra{0} + (1-p_0) \sum_{n=1}^\infty \frac{1}{N_\sigma} \left( \frac{N_\sigma-1}{N_\sigma} \right)^{n-1} \ket{n}\bra{n},
\end{equation}
where $N_\sigma = E/(1-p_0)$ dictates the effective temperature of the tail. This structural result drastically simplifies our task. It reduces an infinite-dimensional functional optimization to a straightforward scalar maximization over the single parameter $p_0 \in [0,1]$:
\begin{equation}
    C_\text{E}(\Phi_p, E) = \max_{p_0 \in [0,1]} I(\Phi_p, \rho^*(p_0,E)).
\end{equation}
Since the objective function evaluated on $\rho^*(p_0,E)$ admits a closed analytical expression, the exact capacity can be readily computed, providing a direct avenue to explicitly demonstrate the sub-optimality of Gaussian encoding.

\subsection{Sub-optimality of Thermal Encodings}

\noindent By evaluating the analytical expression for the capacity derived above, we show that thermal encoding is sub-optimal for the bosonic erasure-like channel on the whole $(p,E)$ plane except on a curve. As illustrated in Fig.~\ref{fig:erasure_optimization} (left) for a reference channel with transmission probability $p=0.5$ and mean energy $E=4$, the quantum mutual information $I(\Phi_p, \rho^*(p_0,E))$, plotted as a function of $p_0$, is maximized at a vacuum population $p_0^{\text{opt}} \approx 0.1$. By contrast, a Gaussian thermal state with the same average energy dictates a strictly larger vacuum probability of $p_0^{\text{th}} = 1/(E+1) = 0.2$. From a physical perspective, this spectral redistribution reduces the overlap between the input signal and the vacuum state introduced by the erasure events, thereby enhancing the distinguishability of the transmitted information at the receiver.
\noindent The quantitative advantage of this non-Gaussian strategy is mapped in Fig.~\ref{fig:erasure_optimization} (right), which displays the absolute capacity gain $\Delta C_\text{E} = C_\text{E}(\Phi_p,E) - I(\Phi_p,\rho_{\text{th},E})$ across the $(p, E)$ parameter space. The plot reveals that the most substantial gains emerge at intermediate energies and transmission probabilities.
\noindent The gain is strictly positive over almost all of the plane, but not everywhere: it
vanishes on the curve $\Gamma$ drawn in the same panel. Since the capacity is attained within the
family of Eq.~\eqref{eq:optimal_ansatz}, and the resulting scalar objective
(App.~\ref{app:erasure_derivation}) is a differentiable function of $p_0$ with an interior maximizer, the thermal state can be optimal
only where it is a stationary point, that is where $\partial_{p_0} I(\Phi_p, \rho^*(p_0,E))$
vanishes at the thermal value $p_0 = 1/(E+1)$. This condition is a single equation in $(p,E)$, and defines a set $\Gamma$ of zero measure
in the $(p,E)$ plane. Numerically, $\Gamma$ is a single
curve $E = E^*(p)$ extending over the whole range of $p$; it lies at low energies, with
$E^*(p) \simeq 1.2$ for $p \lesssim 1/2$, and grows slowly as $p \to 1$, where the channel
approaches the identity channel, for which thermal states are optimal at every energy. Off $\Gamma$ the thermal state is not stationary and is
therefore strictly sub-optimal. On
$\Gamma$ the thermal state is stationary and, numerically, optimal along the whole curve.
\\
As shown in Appendix~\ref{appendix: diagonal optimizer}, the erasure-like channel neither destroys coherences nor maps two distinct Fock-diagonal states to the same output, because $\Phi_p(\rho) = p\,\rho + (1-p)\ket{0}\bra{0}$ with $p > 0$; the maximizer of $I(\Phi_p, \cdot)$ at fixed energy is therefore unique, and any non-Fock-diagonal state is strictly sub-optimal. Combined with the strict sub-optimality of 
the thermal state established above, this rules out the entire Gaussian 
class as a co-optimizer of $C_\text{E}(\Phi_p, E)$ off $\Gamma$. On $\Gamma$, where the thermal state is numerically optimal, the same uniqueness makes it the only optimal state, and every other Gaussian state is strictly sub-optimal there as well. These results motivate the study of more general fading models, for which no closed form is available and which we address numerically in the following sections.

\section{Iterative Optimization algorithm for General Fading}
\label{sec:optimization_method}

\begin{figure*}[t]
	\centering
   \includegraphics[width=0.98\linewidth]{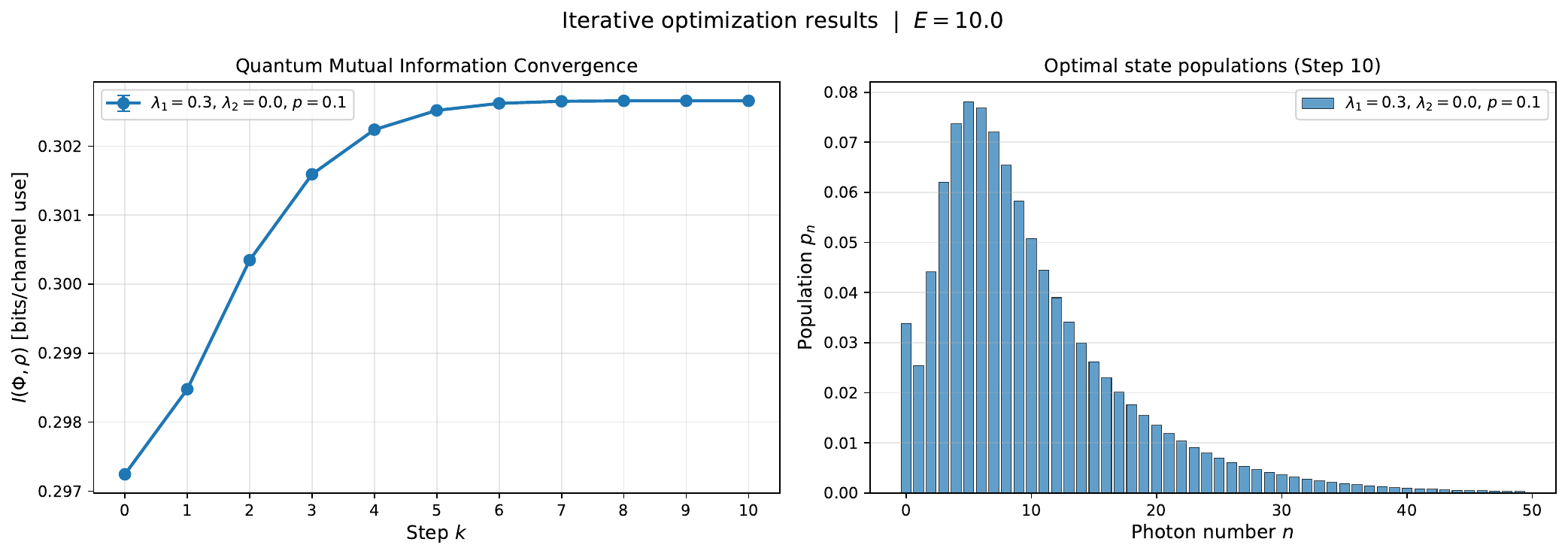}
    
    \caption{Convergence of the iterative optimization algorithm for the erasure-lossy channel $\Phi^{(0)}_{p,\lambda} = p\mathcal{E}_\lambda + (1-p)\mathcal{E}_0$ of Sec.~\ref{sec:binary-lossy}, with $p = 0.1$ and $\lambda = 0.3$, at $E = 10$. As the order $k$ increases, the algorithm expands the non-Gaussian discrete head of the input state, monotonically increasing the achievable information rate until saturation is reached.}
    \label{fig:convergence}
\end{figure*}

\noindent While the analytical solution for the erasure-like channel (Sec.~\ref{sec:erasure_channel}) proves that the optimal input is non-Gaussian everywhere except on the curve $\Gamma$, general fading channels with arbitrary distributions $\{p_n, \lambda_n\}$ do not admit a simple closed-form capacity expression. However, the physical insight gained from the erasure-like channel, i.e.~that the optimal state consists of a tunable discrete head and a thermal tail, suggests a powerful ansatz for numerical optimization. 

\noindent We generalize this structure by introducing an iterative variational algorithm. The core idea is to expand the ``discrete head" of the distribution one Fock state at a time. At each step $k$, we optimize the populations of the first $k$ Fock states individually, while fixing the remaining infinite tail to be a thermal distribution (shifted to start at $n=k$), which is the entropy-maximizing form for the high-energy subspace.

\subsection{The Variational Ansatz}
\noindent We construct a hierarchy of Fock-diagonal states $\rho_E^{(k)}$ parametrized by $k$ independent variables $\{w_0, w_1, \dots, w_{k-1}\}$, representing the weights of the first $k$ Fock states. The ansatz at step $k$ is defined as:
\begin{equation}\label{eq:var ansatz}
    \rho_E^{(k)} = \sum_{n=0}^{k-1} w_n \ket{n}\bra{n} + \left(1 - \sum_{n=0}^{k-1} w_n \right) \sigma_N^{(k)},
\end{equation}
where $\sigma_N^{(k)}$ is the shifted thermal state acting on the subspace $n \ge k$:
\begin{equation}
    \sigma_N^{(k)} = \sum_{n=k}^\infty \frac{1}{N-k+1} \left( \frac{N-k}{N-k+1} \right)^{n-k} \ket{n}\bra{n}.
\end{equation}
The parameter $N$ (the effective temperature of the tail) is not a free variable but is fixed by the global energy constraint $\text{Tr}[a^\dagger a \rho_E^{(k)}] = E$:
\begin{equation}
    N = \frac{E - \sum_{n=0}^{k-1} n w_n}{1 - \sum_{n=0}^{k-1} w_n}.
\end{equation}
To ensure the weights represent valid independent probabilities ($w_n \ge 0, \sum_n w_n \le 1$), we parametrize them using a sequence of independent variables $\pi_i \in [0,1]$ such that $w_n = \pi_n \prod_{j=0}^{n-1}(1-\pi_j)$. Since $\sigma_N^{(k)}$ is supported on photon numbers $n \ge k$, its mean photon number satisfies $N \ge k$, and the weights must therefore obey $\sum_{n=0}^{k-1} (k-n)\, w_n \ge k - E$. This condition is automatically satisfied when $k \le E$; otherwise it restricts the admissible set, to which the numerical optimization is confined. Physically, the variable $\pi_n$ can be intuitively understood as the conditional probability of the state containing exactly $n$ photons, given that it contains at least $n$ photons. This sequential ``survival" parametrization elegantly maps the strictly constrained probability simplex into a simple unit hypercube $[0,1]^k$. Consequently, the simplex constraints are automatically satisfied, and the optimizer only has to respect the feasibility condition above, significantly improving both the numerical stability and the convergence speed of the algorithm.

\subsection{Hierarchy of Solutions}
\noindent This construction creates a nested hierarchy of solutions physically motivated by our analytical findings:

\begin{itemize}
    \item \textbf{Step $k=0$ (Thermal Benchmark):} The ansatz reduces to a pure thermal state $\rho_E^{(0)}=\rho_{\text{th},E}$. This recovers the rate of the thermal state.
    \item \textbf{Step $k=1$ (Erasure-like Solution):} The state is $\rho_E^{(1)} = w_0 \ket{0}\bra{0} + (1-w_0)\sigma_N$. This is exactly the optimal form found analytically for the erasure-like channel in Sec.~\ref{sec:erasure_channel}. It allows the optimization to ``carve out" the vacuum component to distinguish it from loss-induced vacuum.
    \item \textbf{Step $k > 1$ (General Fading):} For more complex fading distributions (e.g., mixtures of two or more lossy channels, or continuous distributions like Weibull), independent control over higher photon number populations ($w_0, w_1, \dots$) allows the state to adapt to the specific interference patterns of the channel components. The variational state at the step $k$ is the one defined in Eq.~\eqref{eq:var ansatz}.
\end{itemize}
Since the family of states at step $k$ strictly includes the family at step $k-1$ (which can be explicitly verified by setting the parameter $w_{k-1}$ to its corresponding thermal value), the maximum mutual information is monotonically non-decreasing with $k$:
\begin{equation}
    I_{\text{max}}^{(k)} \ge I_{\text{max}}^{(k-1)}.
\end{equation}
\noindent Monotonicity alone does not guarantee that the hierarchy reaches $C_\text{E}$, so we
now show that it does. Let $\rho$ be any Fock-diagonal state with $\text{Tr}[a^\dagger a\,\rho] = E$
(states of lower energy need not be considered: mixing in a weight $\delta$ of a Fock state
$\ket{M}$ brings the energy to $E$ and, by concavity, lowers $I$ by at most
$\delta\, I(\Phi,\rho)$, with $\delta \to 0$ as $M \to \infty$),
and let $\rho^{(k)}$ be obtained from it by keeping the populations of the first $k$ Fock states
and replacing the rest by the shifted thermal state carrying the same weight $w_{\ge k}$ and the
same mean photon number $N_k$. Since the replaced part of $\rho$ is supported on $n \ge k$, its
conditional mean obeys $N_k \ge k$: the feasibility condition of Sec.~\ref{sec:optimization_method}
holds automatically, so $\rho^{(k)}$ belongs to the $k$-th family and has energy $E$. The two
states differ only in that tail, whose weight $w_{\ge k}$ vanishes as $k \to \infty$, and the
entropy bounds for mixtures,
$\sum_i q_i S(\sigma_i) \le S(\sum_i q_i\sigma_i) \le \sum_i q_i S(\sigma_i) + H(\{q_i\})$,
bound the resulting change of the input, output and exchange entropies by terms of the
form $H(w_{\ge k})$ and $w_{\ge k}\, g(\cdot)$, where
$g(x) = (x+1)\log_2(x+1) - x\log_2 x$ is the entropy of a thermal state of mean photon
number $x$, which vanish because the tail energy
$w_{\ge k} N_k$ is at most $E$. Hence $I(\Phi, \rho^{(k)}) \to I(\Phi, \rho)$ and
$\sup_k I_{\text{max}}^{(k)} \ge I(\Phi,\rho)$. Since the supremum of $I(\Phi,\cdot)$ can be restricted
to Fock-diagonal states (Appendix~\ref{appendix: diagonal optimizer}), the hierarchy converges to
$C_\text{E}$. The statement concerns the exact maxima $I_{\text{max}}^{(k)}$: the values we report
are obtained with a local optimizer and are therefore lower bounds on them.
We heavily exploit this nested topological structure to enhance the efficiency of our algorithm through a ``warm-start'' technique. Specifically, the numerical optimization at step $k$ is initialized using the optimal parameters $\pi_i$ found at step $k-1$, while the newly introduced $k$-th parameter is explicitly set to match the corresponding thermal distribution tail. This ensures that the optimizer always begins its search from the state found at the previous step, guaranteeing monotonic improvement at each iteration and drastically reducing the computational overhead required to converge in the expanded parameter space.
In the examples we computed, most of the gain comes from the lowest orders, but the number of orders needed depends on the channel: at one of the points we tested, the gain is still growing at $k = 10$ (App.~\ref{app:hierarchy_numerics}).

\subsection{Numerical Implementation}
\noindent The mutual information $I(\Phi, \rho_E^{(k)})$ is evaluated numerically by diagonalizing the output and complementary states (see Appendix~\ref{app:entropic_functionals} for details). The optimization over the $k$-dimensional space $\{\pi_0, \dots, \pi_{k-1}\}$ is performed with a local sequential quadratic programming (SLSQP) routine, warm-started from the optimum found at step $k-1$.
Since the Fock space is infinite, we introduce a sufficiently large cutoff dimension $\bar{n}$ for the numerical computation of entropies. We derive a rigorous bound on the truncation error in Appendix~\ref{app:truncation_error_bound}, which, in the present implementation, is evaluated for every reported value: all the values obtained with the algorithm are certified in this sense, that is, given with two-sided bounds on the truncation error.

\begin{figure*}[t]
	\centering
	\includegraphics[width=0.98\linewidth]{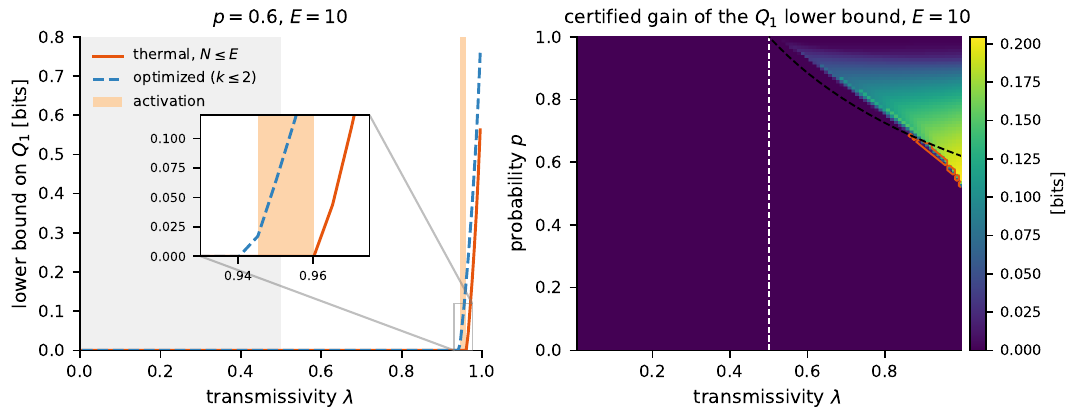}
    \caption{Activation of the quantum capacity for the binary fading family 
$\Phi_{p,\lambda}^{(0)} = p\mathcal{E}_{\lambda} + (1-p)\mathcal{E}_0$, 
at energy $E = 10$, where $I_c^{(k)} \equiv I_c(\Phi_{p,\lambda}^{(0)}, 
\rho^{(k)})$ denotes the coherent information evaluated on the $k$-th 
order variational state. The thermal benchmark is the best thermal state with mean photon number $N \le E$ (Sec.~\ref{sec:quantum-capacity}). \textbf{Left:} the certified lower bound on $Q_1$ as a function of the transmissivity $\lambda$, at fixed $p = 0.6$, given by the best thermal state (orange) and by the optimized state ($k \le 2$, blue). For $\lambda \le 1/2$ (grey) the channel is antidegradable and $I_c \le 0$ for every input. The optimized state becomes positive at slightly lower values of $\lambda$ than the thermal states, and activation is found in the window $\lambda \in [0.945, 0.960]$ (shaded, enlarged in the inset), beyond which the thermal benchmark is positive as well. 
\textbf{Right:} certified gain of the lower bound on $Q_1$ over the thermal benchmark, across the $(p, \lambda)$ plane; the orange contour encloses the region where activation is found with $k \le 2$, along the curve where the thermal coherent information vanishes; higher orders enlarge it (Sec.~\ref{sec:quantum-capacity}). Above the black dashed line, $p\lambda(1+p) = 1$, weak thermal inputs already have a positive coherent information (Lemma~\ref{lem:weak_inputs}), so that no activation is possible there.}
    \label{fig:activation_Q1}
\end{figure*}

\section{Numerical results for fading channels}\label{sec:numerical}

\noindent We now apply our iterative optimization algorithm to physically relevant fading models. We consider both discrete binary mixtures, which serve as fundamental building blocks for understanding fading dynamics, and continuous distributions modeling atmospheric turbulence.

\subsection{Mixtures of two lossy channels}\label{sec:binary-lossy}
\noindent We first consider the family of channels defined by the convex combination of two pure-loss channels:
\begin{equation}
\Phi_{p,\lambda_1, \lambda_2}(\rho)=p\,\mathcal{E}_{\lambda_1}(\rho)+(1-p)\,\mathcal{E}_{\lambda_2}(\rho).
\end{equation}
We focus on two sub-families that represent extreme but physically motivated fading scenarios:

\paragraph{Lossy vs. Noisy Channel ($\lambda_2=0$):}
The channel $\Phi_{p,\lambda}^{(0)} = p\mathcal{E}_\lambda + (1-p)\mathcal{E}_0$ models a link where the signal is transmitted with transmissivity $\lambda$ with probability $p$, or completely lost (replaced by vacuum) with probability $1-p$. This generalizes the erasure-like channel to the case of imperfect transmission. 
In Fig.~\ref{fig:gain_erasure_lossy}, we plot the achievable gain in mutual information over the thermal benchmark, $G = I(\Phi_{p,\lambda}^{(0)},\rho^{(1)}) - I(\Phi_{p,\lambda}^{(0)},\rho_{\text{th}})$ in the $(p, \lambda)$ plane for low ($E=1$) and high ($E=10$) energies. The gain is largest at intermediate values of $p$ and grows with the energy, up to $0.074$ bits at $E = 10$. At $E = 2$, $5$ and $10$ (Fig.~\ref{fig:gain_grid_2x2}) it vanishes on a curve of the $(p,\lambda)$ plane on which the thermal state is stationary within the first-order family, the analogue of the curve $\Gamma$ of Sec.~\ref{sec:erasure_channel}. Unlike on $\Gamma$, however, the thermal state is not optimal there: second-order states still outperform it (App.~\ref{app:plots_gallery}).

\paragraph{Identity vs. Lossy Channel ($\lambda_1=1$):}
The channel $\Phi_{p,\lambda}^{(1)} = p\mathcal{I} + (1-p)\mathcal{E}_\lambda$ represents a scenario where the signal is either perfectly transmitted or attenuated. This models intermittent noise or ``clear air" turbulence. Also in this case, our algorithm finds non-Gaussian states that outperform the thermal benchmark, particularly when the contrast between the two fading components is high. Plots for this model can be found in App.~\ref{app:identity_lossy}.

\subsection{General fading ensembles}\label{sec:general}
\noindent We extend our analysis to realistic atmospheric fading models. We consider the discretized Log-Negative Weibull distribution, which describes the transmissivity statistics of a beam subject to wander-induced fading~\cite{Vasylyev2012PRL_LogNegativeWeibull, Pirandola2021}. By applying our iterative algorithm to the discretized ensemble $\{p_i, \lambda_i\}$, which is the Gauss quadrature rule of the distribution described in App.~\ref{app:general_fading}, we compute the single-shot coherent information of Sec.~\ref{sec:quantum-capacity} over a window of turbulence parameters, at $E = 1$ and $E = 10$ (App.~\ref{app:general_fading}). For this model, states of order $k \le 2$ improve only slightly on the thermal benchmark, and with them activation is found only at $E = 10$, at a few points on the edge of the parameter window.

\subsection{Activation of quantum capacity}\label{sec:quantum-capacity}
\noindent While the entanglement-assisted capacity is strictly positive for any non-trivial quantum channel \cite{Bennett2002IEEEEACC}, the quantum capacity vanishes for a much broader class of channels. A key question is whether non-Gaussian strategies can ``activate" the quantum capacity in regimes where thermal inputs yield no positive coherent information. We address this by maximizing the single-shot coherent information $Q_1(\Phi) = \max_{\rho} I_c(\Phi, \rho)$, where $I_c(\Phi,\rho) = S(\Phi(\rho)) - S(\tilde{\Phi}(\rho))$. As for $C_\text{E}$, the maximization is under the energy constraint $\text{Tr}[a^\dagger a\rho] \le E$. Since, by Lemma~\ref{lem:weak_inputs}, thermal states can have a positive coherent information at small $N$ also where it is negative at $N = E$, the thermal benchmark is the best thermal state compatible with the constraint, $I_c^{\text{th}} = \max_{N \le E} I_c(\Phi, \rho_{\text{th},N})$, which we maximize numerically over a grid of values of $N$ (App.~\ref{app:ic_numerics}). We speak of activation where $I_c^{\text{th}} \le 0$ while an optimized non-Gaussian input yields a strictly positive coherent information; by Lemma~\ref{lem:weak_inputs}, this can only happen where $\braket{\sqrt{\lambda}}^2 + \braket{\lambda} \le 1$. By definition, $Q_1(\Phi)$ constitutes a fundamental lower bound to the quantum capacity $Q(\Phi)$ for any channel.\\
\noindent We apply our iterative optimization to $I_c(\Phi, \rho)$. The results are shown in Fig.~\ref{fig:activation_Q1}. We identify parameter regions where the coherent information of every thermal state compatible with the energy constraint drops to zero or becomes negative, so that the thermal-state lower bound on $Q$ is vacuous, while our optimized non-Gaussian state yields a strictly positive coherent information. The maps of Fig.~\ref{fig:activation_Q1} use states of order $k \le 2$. With these states, at $p = 0.6$ and $E = 10$ activation is found in a narrow window of $\lambda$ next to the transmissivity at which the thermal coherent information vanishes; in the $(p,\lambda)$ plane these windows form a band along that curve, which ends where $p\lambda(1+p) > 1$: there weak thermal inputs already give a positive coherent information. Higher orders enlarge the activation region considerably: at the same $p$ and $E$, states of order $k \le 8$ have a positive coherent information down to $\lambda \approx 0.8$ (App.~\ref{app:higher_orders}). The maps are therefore a conservative estimate of the region where activation occurs.
For the entanglement-assisted capacity, the phase-covariance of the channel guarantees, via the concavity of the mutual information, that the Fock-diagonal subspace contains the global optimum (see Appendix~\ref{appendix: diagonal optimizer}). This argument does \emph{not} transfer to the coherent information, which is not concave in the input: phase twirling can increase it as well as decrease it (for the identity channel it raises $I_c$ from $0$ to $1$ on the input $(\ket{0}+\ket{1})/\sqrt{2}$), so our Fock-diagonal optimization yields a lower bound on $Q_1$ within this subspace, and the absolute maximum of $I_c$ might in principle be attained by a non-diagonal state. This restriction does not weaken our conclusion: exhibiting any single state with $I_c > 0$ in a regime where the thermal benchmark is non-positive is, by itself, sufficient to certify a strictly positive $Q_1$, and hence channel activation.\\
\noindent This result establishes \textbf{channel activation}: the fading
channel, for which every thermal state compatible with the energy constraint has a non-positive coherent information in 
this regime, is in fact capable of reliable quantum transmission 
($Q \ge Q_1 > 0$) when accessed with a suitable non-Gaussian encoding.
The positivity of $Q_1$ at these points is certified; the non-positivity of the thermal benchmark is verified only on the sampled values of $N$ (App.~\ref{app:ic_numerics}). The comparison is with the thermal benchmark, that is with all thermal states compatible with the energy constraint, which are the Fock-diagonal Gaussian states. 
Whether some other (non-thermal) Gaussian state could also achieve a positive 
coherent information in these regimes remains open, precisely because the 
non-concavity of $I_c$ prevents the phase-twirling reduction that settles the 
analogous question for the entanglement-assisted capacity.

\noindent \emph{Which input states realize the activation.} The states returned by the
optimization are not generic: by construction they belong to the variational family of
Eq.~\eqref{eq:var ansatz}, and in the activation region they take a simple form, which we now
describe. We do so at the representative point $p = 0.6$, $\lambda = 0.95$, $E = 10$ of
Fig.~\ref{fig:activation_Q1}, where the thermal state of mean photon number $E$ has $I_c \approx -0.13$ bits and thermal states have a negative coherent information at all the sampled values of $0 < N \le E$. Already at
first order the best state we find lies exactly on the boundary $w_0 = 0$ of the admissible set: it
is the vacuum-free thermal state, that is a geometric distribution supported on $n \ge 1$ whose
temperature is fixed by the energy constraint, and it already gives a positive coherent information. Up to $k = 5$, higher orders keep $w_0 = w_1 = 0$ and redistribute the weight over the
lowest Fock levels, reaching $I_c \approx 0.11$ bits; the populations of these states and their
coherent information at each order are shown in Fig.~\ref{fig:a5_optimal_states}. Beyond $k = 5$ the
structure changes: the optimized states concentrate on a few Fock levels, which carry only part of the available energy, the rest being in a tail of small weight, and $I_c$ grows to about $0.22$ bits at $k = 8$ (App.~\ref{app:higher_orders}). These values are evaluated with two-sided bounds on the
truncation error (App.~\ref{app:truncation_error_bound}), so that the positive ones are certified
lower bounds on $Q_1$.

\noindent The structure of these states has a simple origin. The complete-loss branch delivers the
vacuum to the receiver, so a vacuum population in the input cannot be distinguished from a loss
event; suppressing it separates the two branches. In the erasure-like limit $\lambda = 1$, where
the channel reduces to $\Phi_p = p\,\mathcal{I} + (1-p)\,\mathcal{E}_0$ of
Sec.~\ref{sec:erasure_channel}, this is exact: for every input with $\braket{0|\rho|0} = 0$, and
therefore with $\rho\ket{0} = 0$, the coherences between the two branches vanish, so that the
channel and its complement produce the orthogonal decompositions
$p\,\rho \oplus (1-p)\ket{0}\bra{0}$ and $(1-p)\,\rho \oplus p\ket{0}\bra{0}$, whence
\begin{equation}\label{eq:vacuum_free_ic}
    I_c(\Phi_p, \rho) = (2p-1)\,S(\rho),
\end{equation}
which is the coherent information of an erasure channel~\cite{Bennett1997Erasure} with erasure probability
$1-p$ and is strictly positive for $p > 1/2$. The optimization therefore
drives the input away from the vacuum, and the non-Gaussianity responsible for the activation stays
confined to the lowest Fock levels.

\section{Conclusions}

\noindent In this work, we have provided a comprehensive analysis of the quantum Shannon-theoretic properties of bosonic fading channels, modeled as convex combinations of pure-loss maps. Our results advance the understanding of this practically relevant class of channels along two complementary directions. On the structural side, we have rigorously established that any non-trivial fading channel is non-degradable, as witnessed by the rank of its Choi matrix exceeding the degradability threshold. We have further derived an explicit analytical condition for non-antidegradability, $\langle\sqrt{\lambda}\rangle^2 + \langle\lambda\rangle > 1$, expressed in terms of the first two fractional moments of the fading distribution, and evaluated it over the full parameter space of physically relevant models, including the Log-Negative Weibull distribution. The same condition also guarantees a strictly positive quantum capacity, witnessed by thermal inputs of low mean photon number. Most notably, we have proven that every non-trivial fading channel is distillable: by constructing an NPT output state from a two-qubit entangled probe, we have shown that $Q_2(\Phi) > 0$ for any fading distribution that is not concentrated at zero transmissivity. This result guarantees that entanglement distribution and quantum key distribution are feasible at a strictly positive rate, regardless of the turbulence strength. For a finite set of transmissivities we have determined the two-way quantum and secret-key capacities exactly, both being $\sum_n p_n \log_2[1/(1-\lambda_n)]$ when the input energy is unconstrained. Whether the same value is attained for a continuous distribution is open. At a fixed input energy we have derived an explicit rate via multi-rail encoding, on which the fluctuations of the transmissivity act as an effective dephasing between the rails; it is strictly positive whenever $\langle\lambda\rangle > 0$. On the operational side, we have demonstrated that thermal-state encodings, which are optimal for static bosonic Gaussian channels~\cite{Holevo2001PRAGaussian, giovannetti_solution_gaussian_optimality_conjecture_2015}, are in general not optimal for fading channels. For the paradigmatic binary erasure-like model, we have derived the exact capacity-achieving state in closed form, proving analytically that there are channels in this family for which the entire Gaussian class is strictly sub-optimal for the entanglement-assisted classical capacity. We then studied two further binary families numerically through an iterative variational algorithm that optimizes over Fock-diagonal input states, demonstrating strictly positive gains in the achievable mutual information over thermal inputs at every parameter point we computed. Our analysis of the single-shot coherent information also reveals channel activation: non-Gaussian Fock-diagonal states achieve a strictly positive coherent information where every thermal input compatible with the energy constraint gives zero or negative values, which proves that $Q(\Phi) \ge Q_1(\Phi) > 0$ there.\\
\noindent Our findings identify non-Gaussian state preparation as a resource for communication over fading channels. The optimized states are Fock-diagonal and, at low energy, confined to the lowest Fock levels. Our analysis also delineates several questions that remain open. First, the activation we establish is certified through the single-shot quantity $Q_1$; since fading channels are non-degradable, the coherent information may be super-additive, and determining the true quantum capacity $Q$, or even establishing tighter multi-letter lower bounds, remains an outstanding problem. Second, our optimization of the coherent information is carried out over Fock-diagonal states, which the phase-twirling argument does not certify as optimal for this non-concave functional; whether non-diagonal inputs, or non-thermal Gaussian states, can further enlarge or alter the activation region is not settled by our analysis. Third, while we prove that for some fading channels the entire Gaussian class is strictly sub-optimal for the entanglement-assisted capacity, a closed-form expression for $C_\text{E}$ of a general fading channel is still lacking, and our variational algorithm provides increasingly tight lower bounds rather than the exact value. Finally, extending the present single-mode treatment to the multi-mode regime, where correlated fading across modes and joint non-Gaussian encodings come into play, is a natural and physically motivated direction, as is the investigation of the super-additive effects made possible by non-degradability. A further natural extension is fading with a fluctuating phase, where no phase reference is available; quantum and entanglement-assisted communication in that setting have been studied in~\cite{Zhuang2021PhaseReference}. Addressing these questions would complete the quantum Shannon theory of fading channels initiated here.
\section*{Acknowledgements}
\noindent
GDP has been supported by the UNA EUROPA SeedFunding project QUANTUMUnaE (CUP J37G25000380006).
GDP is a member of the ``Gruppo Nazionale per la Fisica Matematica (GNFM)'' of the ``Istituto Nazionale di Alta Matematica ``Francesco Severi'' (INdAM)''. F.A.M.\ acknowledges financial support from the European Union (ERC StG ETQO, Grant Agreement no.\ 101165230).

\bibliographystyle{unsrt}
\bibliography{biblio}

\clearpage 
\onecolumngrid 

\appendix
\section{Proof of Non-Antidegradability Condition}\label{app:antidegradability_proof}

\noindent In this appendix, we provide the detailed derivation of the fidelity expressions used for the non-antidegradability criterion discussed in Sec.~\ref{sec:antideg}. Our goal is to derive an explicit sufficient condition for the channel $\Phi = \sum_n p_n \mathcal{E}_{\lambda_n}$ to be non-antidegradable, based on the inequality $F(\tilde{\Phi}(\rho_1), \tilde{\Phi}(\rho_2)) > F(\Phi(\rho_1), \Phi(\rho_2))$, where $\tilde{\Phi}$ is the complementary channel.
We choose as probe states the first two Fock states, $\rho_1 = \ket{0}\bra{0}$ and $\rho_2 = \ket{1}\bra{1}$.

\subsection{Main Output Fidelity}
\noindent The action of the fading channel $\Phi$ on the probe states yields:
\begin{align}
    \Phi(\rho_1) &= \sum_n p_n \mathcal{E}_{\lambda_n}(\ket{0}\bra{0}) = \sum_n p_n \ket{0}\bra{0} = \ket{0}\bra{0}, \\
    \Phi(\rho_2) &= \sum_n p_n \mathcal{E}_{\lambda_n}(\ket{1}\bra{1}) = \sum_n p_n \left[ (1-\lambda_n)\ket{0}\bra{0} + \lambda_n \ket{1}\bra{1} \right] \nonumber \\
    &= (1-\braket{\lambda})\ket{0}\bra{0} + \braket{\lambda} \ket{1}\bra{1}.
\end{align}
Since both output states are diagonal in the Fock basis, their fidelity $F(\rho,\sigma) = (\text{Tr}\sqrt{\sqrt{\rho}\sigma\sqrt{\rho}})^2 = \left( \sum_k \sqrt{\rho_{kk}\sigma_{kk}} \right)^2$ simplifies to the squared sum of the square roots of the product of their populations. The only non-zero overlap occurs in the vacuum component:
\begin{equation}
    F(\Phi(\rho_1), \Phi(\rho_2)) = 1 - \braket{\lambda}.
\end{equation}

\subsection{Complementary Output Fidelity}
\noindent To compute the fidelity of the complementary channel, we use the Stinespring dilation derived in Sec.~\ref{sec:complementary}. The complementary map is given by:
\begin{equation}
	\tilde{\Phi}(\rho_S) = \text{Tr}_{S}\left[ V (\rho_S \otimes \ket{0}_E\bra{0}\otimes \ket{\psi_p}_A\bra{\psi_p}) V^\dagger \right],
\end{equation}
where the isometry $V$ acting on $S \otimes E \otimes A$ is defined as $V = \sum_{n} U_{\lambda_n}^{(SE)} \otimes \ket{n}_A \bra{n}$ and the ancilla state is $\ket{\psi_p}_A = \sum_n \sqrt{p_n}\ket{n}_A$.
Let us explicitly compute the action of the global unitary $V$ on the input states $\ket{0}_S$ and $\ket{1}_S$ (tensored with the environment and ancilla).
For the input $\rho_1 = \ket{0}\bra{0}$:
\begin{align}
    V \ket{0}_S \ket{0}_E \ket{\psi_p}_A &= \sum_n \left( U_{\lambda_n}^{(SE)} \otimes \ket{n}_A\bra{n} \right) \ket{0}_S \ket{0}_E \sum_k \sqrt{p_k}\ket{k}_A \nonumber \\
    &= \sum_n \sqrt{p_n} \left( U_{\lambda_n}^{(SE)}\ket{0}_S\ket{0}_E \right) \otimes \ket{n}_A \nonumber \\
    &= \sum_n \sqrt{p_n} \ket{0}_S\ket{0}_E \otimes \ket{n}_A = \ket{0}_S \ket{0}_E \otimes \ket{\psi_p}_A.
\end{align}
Tracing out the system $S$ (which is in the state $\ket{0}_S$), we obtain the first complementary output:
\begin{equation}
    \tilde{\Phi}(\rho_1) = \ket{0}\bra{0}_E \otimes \ket{\psi_p}\bra{\psi_p}_A.
\end{equation}
For the input $\rho_2 = \ket{1}\bra{1}$:
\begin{align}
    V \ket{1}_S \ket{0}_E \ket{\psi_p}_A &= \sum_n \sqrt{p_n} \left( U_{\lambda_n}^{(SE)}\ket{1}_S\ket{0}_E \right) \otimes \ket{n}_A \nonumber \\
    &= \sum_n \sqrt{p_n} \left( \sqrt{\lambda_n}\ket{1}_S\ket{0}_E + \sqrt{1-\lambda_n}\ket{0}_S\ket{1}_E \right) \otimes \ket{n}_A \nonumber \\
    &= \ket{1}_S \ket{0}_E \otimes \left( \sum_n \sqrt{p_n \lambda_n}\ket{n}_A \right) + \ket{0}_S \ket{1}_E \otimes \left( \sum_n \sqrt{p_n (1-\lambda_n)}\ket{n}_A \right).
\end{align}
Let us define the unnormalized vectors $\ket{\psi_{\lambda p}}_A = \sum_n \sqrt{p_n \lambda_n}\ket{n}_A$ and $\ket{\psi_{(1-\lambda)p}}_A = \sum_n \sqrt{p_n (1-\lambda_n)}\ket{n}_A$.
Tracing out the system $S$ (summing over the orthogonal states $\ket{0}_S$ and $\ket{1}_S$) yields:
\begin{equation}
    \tilde{\Phi}(\rho_2) = \ket{0}\bra{0}_E \otimes \ket{\psi_{\lambda p}}\bra{\psi_{\lambda p}}_A + \ket{1}\bra{1}_E \otimes \ket{\psi_{(1-\lambda)p}}\bra{\psi_{(1-\lambda)p}}_A.
\end{equation}
Finally, we compute the fidelity $F(\tilde{\Phi}(\rho_1), \tilde{\Phi}(\rho_2))$. Since the environment states $\ket{0}_E$ and $\ket{1}_E$ are orthogonal, the fidelity splits into the sum of overlaps in each environment sector. However, $\tilde{\Phi}(\rho_1)$ has no support on $\ket{1}_E$, so only the $\ket{0}_E$ sector contributes:
\begin{equation}
    F(\tilde{\Phi}(\rho_1), \tilde{\Phi}(\rho_2)) = \left| \braket{\psi_p | \psi_{\lambda p}} \right|^2.
\end{equation}
Calculating the inner product explicitly:
\begin{equation}
    \braket{\psi_p | \psi_{\lambda p}} = \sum_n (\sqrt{p_n}) (\sqrt{p_n \lambda_n}) = \sum_n p_n \sqrt{\lambda_n} = \braket{\sqrt{\lambda}}.
\end{equation}
The derivation holds for an arbitrary $p(\lambda)$. With the continuous control register of Sec.~\ref{sec:complementary}, $\ket{\psi_p}_A$, $\ket{\psi_{\lambda p}}_A$ and $\ket{\psi_{(1-\lambda)p}}_A$ are the functions $1$, $\sqrt{\lambda}$ and $\sqrt{1-\lambda}$, every step above goes through with sums replaced by integrals, and the inner product becomes $\int_0^1 d\lambda\, p(\lambda)\sqrt{\lambda} = \braket{\sqrt{\lambda}}$.
Thus, the complementary fidelity is simply the square of the fractional moment:
\begin{equation}
    F(\tilde{\Phi}(\rho_1), \tilde{\Phi}(\rho_2)) = \braket{\sqrt{\lambda}}^2.
\end{equation}

\subsection{Detailed Analysis of Binary Fading Mixtures}\label{app:binary_antideg}

\noindent In the following, we analyze the non-antidegradability condition for the fundamental class of binary fading channels. We consider a general convex combination of two pure-loss channels with transmissivities $\lambda_1$ and $\lambda_2$ mixed with probability $p$:
\begin{equation}\label{app: binary fading definition}
    \Phi_{p, \lambda_1, \lambda_2} = p \mathcal{E}_{\lambda_1} + (1-p) \mathcal{E}_{\lambda_2}.
\end{equation}
Applying the general sufficient condition derived in Eq.~\eqref{eq:non_antideg_condition} ($\braket{\sqrt{\lambda}}^2 + \braket{\lambda} > 1$), we obtain the explicit inequality governing the binary case:
\begin{equation}\label{eq:binary_general_condition}
    \left( p\sqrt{\lambda_1} + (1-p)\sqrt{\lambda_2} \right)^2 + p\lambda_1 + (1-p)\lambda_2 > 1.
\end{equation}
This inequality defines a volume in the parameter space $(p, \lambda_1, \lambda_2)$ where the channel is guaranteed to be non-antidegradable. To gain physical insight, we examine 2D sections of this volume corresponding to specific channel configurations.

\begin{figure*}[t]
    \centering
    \includegraphics[width=0.32\linewidth]{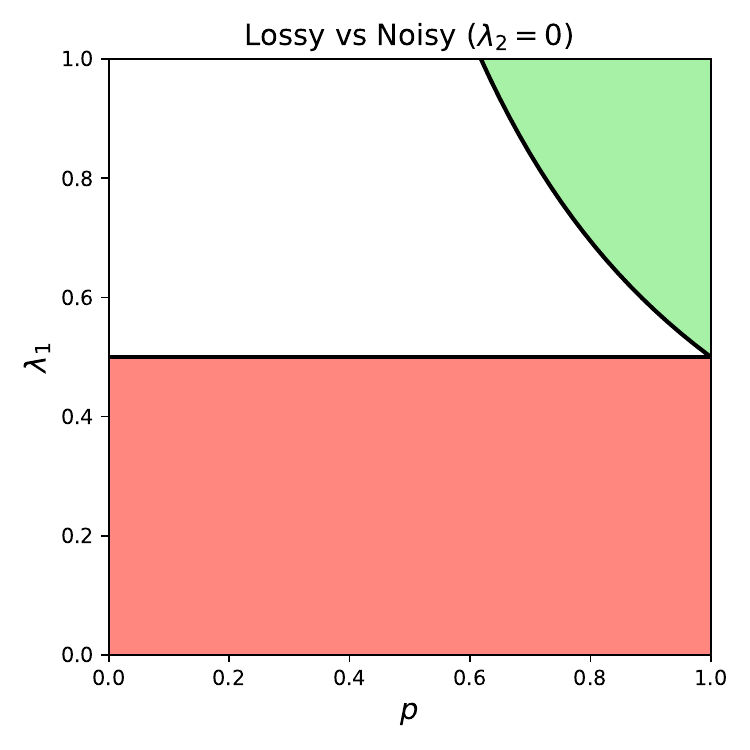}
    \hfill
    \includegraphics[width=0.32\linewidth]{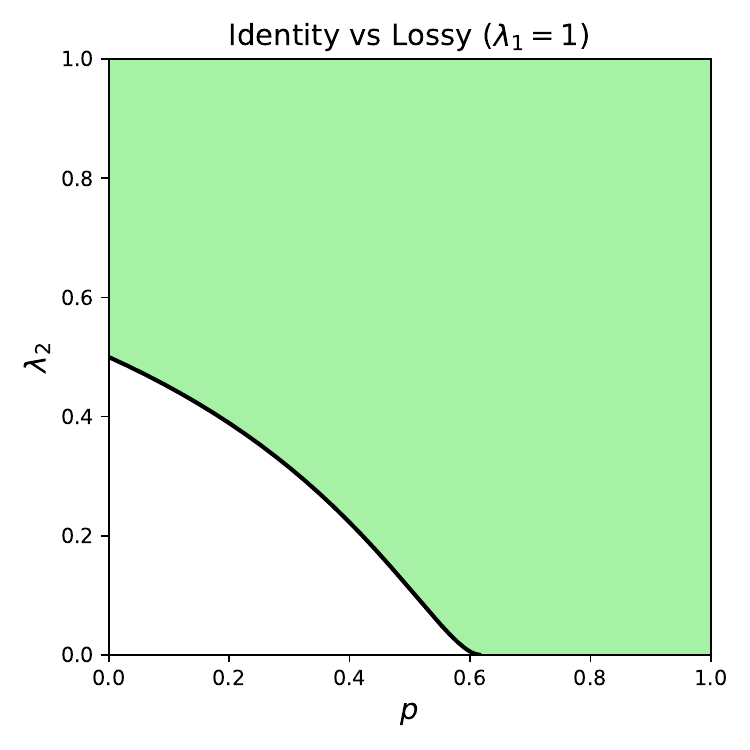}
    \hfill
    \includegraphics[width=0.32\linewidth]{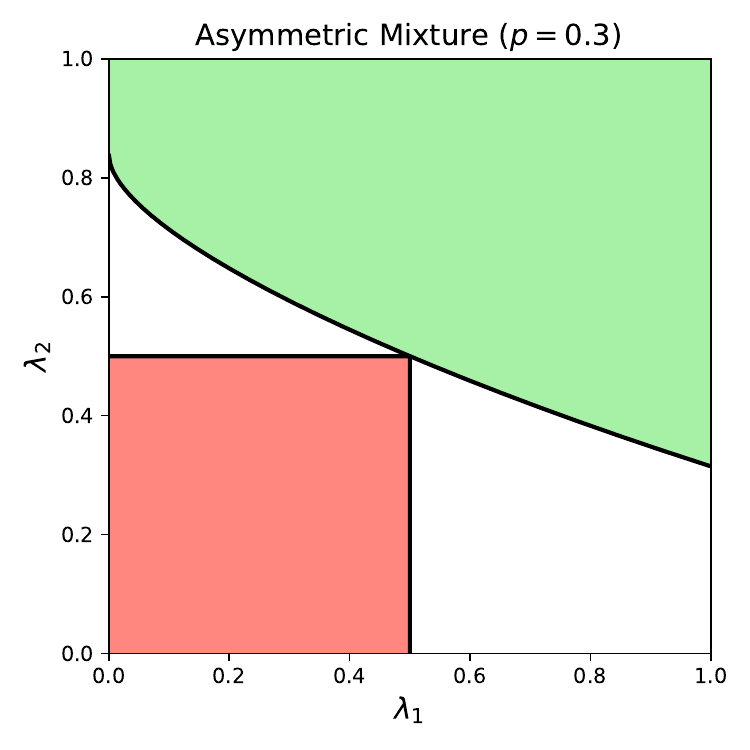}
    
    \caption{Regions of non-antidegradability (green) and of configurations for which our sufficient condition does not apply (red) for specific binary fading configurations defined in Eq.~\eqref{app: binary fading definition}. \textbf{(a)} Lossy vs. Noisy ($\lambda_2=0$) in the $(p, \lambda_1)$ plane. \textbf{(b)} Identity vs. Lossy ($\lambda_1=1$) in the $(p, \lambda_2)$ plane. \textbf{(c)} General binary mixture with fixed asymmetric weight $p=0.3$ in the $(\lambda_1, \lambda_2)$ plane. The boundary curves represent the analytic threshold where Eq.~\eqref{eq:binary_general_condition} holds with equality.}
    \label{fig:app_antideg_cases}
\end{figure*}

\subsubsection{Specific Case Studies}

\begin{itemize}
    \item \textbf{Binary Erasure-like Channel ($\lambda_1=1, \lambda_2=0$):} 
    This corresponds to the mixture $\Phi_p = p \mathcal{I} + (1-p) \mathcal{E}_0$. Substituting into Eq.~\eqref{eq:binary_general_condition}, we recover the condition $p^2 + p - 1 > 0$. Our sufficient condition therefore certifies non-antidegradability above the inverse golden ratio:
    \begin{equation}
        p > \frac{\sqrt{5}-1}{2} = \frac{1}{\phi} \approx 0.618.
    \end{equation}
    This bound is not tight: by Eq.~\eqref{eq:vacuum_free_ic}, an input orthogonal to the vacuum has $I_c > 0$, and hence the channel is non-antidegradable, for every $p > 1/2$.
    
    \item \textbf{Lossy vs. Noisy Channel ($\lambda_2=0$):} 
    Consider the mixture of a generic lossy channel $\mathcal{E}_{\lambda_1}$ and the complete-loss channel $\mathcal{E}_0$ ($\Phi = p \mathcal{E}_{\lambda_1} + (1-p) \mathcal{E}_0$). The condition simplifies to $p(p+1)\lambda_1 > 1$, or:
    \begin{equation}
        \lambda_1 > \frac{1}{p(p+1)}.
    \end{equation}
    Above this transmissivity weak thermal inputs already achieve a positive coherent information (Lemma~\ref{lem:weak_inputs}). Below it the condition is silent, and non-Gaussian inputs can still achieve $Q_1 > 0$: for $p = 0.6$ the threshold exceeds $1$, yet $Q_1 > 0$ at $\lambda_1 = 0.95$ (Sec.~\ref{sec:quantum-capacity}). The valid region is shown in Fig.~\ref{fig:app_antideg_cases}(a).

    \item \textbf{Identity vs. Lossy Channel ($\lambda_1=1$):} 
    Consider the mixture of the identity channel and a noisy background $\mathcal{E}_{\lambda_2}$ ($\Phi = p \mathcal{I} + (1-p) \mathcal{E}_{\lambda_2}$). The condition identifies the maximum background loss $\lambda_2$ tolerable before the channel potentially becomes antidegradable. The phase diagram is shown in Fig.~\ref{fig:app_antideg_cases}(b).

    \item \textbf{Asymmetric General Mixture ($p=0.3$):} 
    We investigate the case of a biased mixture where the first component is less probable ($p=0.3$). As shown in Fig.~\ref{fig:app_antideg_cases}(c), the non-antidegradable region in the $(\lambda_1, \lambda_2)$ plane is significantly skewed compared to the symmetric $p=0.5$ case (discussed in the main text), highlighting the strong dependence of the quantum capacity region on the mixing weights.
\end{itemize}

\subsection{Coherent information of weak inputs}\label{app:weak_inputs}
\noindent We prove Lemma~\ref{lem:weak_inputs}. We write $\eta(x) = -x\log_2 x$, and $0 \le N \le 1$ whenever $\rho_N$ appears.

\emph{Spectra for $\rho_N$.} Since $\rho_N = (1-N)\ket{0}\bra{0} + N\ket{1}\bra{1} = (1-N)\rho_1 + N\rho_2$, the outputs computed in the previous subsections give, by linearity,
\begin{align}
\Phi(\rho_N) &= (1 - N\braket{\lambda})\ket{0}\bra{0} + N\braket{\lambda}\ket{1}\bra{1}, \\
\tilde{\Phi}(\rho_N) &= \ket{0}\bra{0}_E \otimes \big[(1-N)\ket{\psi_p}\bra{\psi_p} + N\ket{\psi_{\lambda p}}\bra{\psi_{\lambda p}}\big]_A + N\ket{1}\bra{1}_E \otimes \ket{\psi_{(1-\lambda)p}}\bra{\psi_{(1-\lambda)p}}_A .
\end{align}
The output has eigenvalues $1 - N\braket{\lambda}$ and $N\braket{\lambda}$. The complementary output is block diagonal in $E$. The block of $\ket{1}_E$ has rank one, with eigenvalue $N\braket{\psi_{(1-\lambda)p}|\psi_{(1-\lambda)p}} = N(1-\braket{\lambda})$. The block of $\ket{0}_E$ is supported on the span of $\ket{\psi_p}$ and $\ket{\psi_{\lambda p}}$, whose inner products are $\braket{\psi_p|\psi_p} = 1$, $\braket{\psi_{\lambda p}|\psi_{\lambda p}} = \braket{\lambda}$ and $\braket{\psi_p|\psi_{\lambda p}} = \braket{\sqrt{\lambda}}$, so that its two eigenvalues $\mu_\pm$ satisfy
\begin{equation}
\mu_+ + \mu_- = 1 - N(1-\braket{\lambda}), \qquad \mu_+\mu_- = N(1-N)\big(\braket{\lambda} - \braket{\sqrt{\lambda}}^2\big).
\end{equation}
Only inner products of vectors of the register enter, so these spectra hold for every $p(\lambda)$, discrete or continuous, and both states have rank at most three. Hence
\begin{equation}\label{eq:weak_ic_exact}
I_c(\Phi, \rho_N) = \eta\big(N\braket{\lambda}\big) + \eta\big(1 - N\braket{\lambda}\big) - \eta\big(N(1-\braket{\lambda})\big) - \eta(\mu_+) - \eta(\mu_-).
\end{equation}

\emph{Small $N$.} By the Cauchy--Schwarz inequality, $\braket{\sqrt{\lambda}}^2 \le \braket{\lambda} \le 1$. As $N \to 0$, $\mu_+ = 1 + O(N)$, and therefore $\mu_- = N(\braket{\lambda} - \braket{\sqrt{\lambda}}^2) + O(N^2)$. Each state thus has one eigenvalue $1 + O(N)$, contributing $O(N)$ to its entropy, and small eigenvalues $cN + O(N^2)$ with $c \ge 0$, each contributing $cN\log_2(1/N) + O(N)$. The small eigenvalues add up to $N\braket{\lambda}$ for the output, and to $N(1-\braket{\lambda}) + N(\braket{\lambda} - \braket{\sqrt{\lambda}}^2) = N(1 - \braket{\sqrt{\lambda}}^2)$ for the complementary output, so that
\begin{equation}
I_c(\Phi, \rho_N) = \big[\braket{\lambda} - (1 - \braket{\sqrt{\lambda}}^2)\big]\, N\log_2\frac{1}{N} + O(N) = \Delta F\, N\log_2\frac{1}{N} + O(N).
\end{equation}

\emph{Thermal states.} Write $\rho_{\text{th},N} = (1-\epsilon)\,\rho_{N'} + \epsilon\,\tau$, where $\epsilon = N^2/(1+N)^2$ is the probability of two or more photons, $N' = N/(1+2N)$, and $\tau$ is the normalized part of $\rho_{\text{th},N}$ on $n \ge 2$: a thermal distribution shifted by two photons, with mean photon number $N+2$ and entropy $g(N)$, where $g(x) = (x+1)\log_2(x+1) - x\log_2 x$. For any two states, $(1-\epsilon)S(\omega_1) + \epsilon S(\omega_2) \le S\big((1-\epsilon)\omega_1 + \epsilon\omega_2\big) \le (1-\epsilon)S(\omega_1) + \epsilon S(\omega_2) + H_2(\epsilon)$, with $H_2$ the binary entropy. Applied to the output and to the complementary output, which are linear in the input, these bounds give
\begin{equation}
\big| I_c(\Phi,\rho_{\text{th},N}) - (1-\epsilon)\, I_c(\Phi,\rho_{N'}) - \epsilon\, I_c(\Phi,\tau) \big| \le H_2(\epsilon).
\end{equation}
Moreover, $I_c(\Phi,\tau) \le S(\Phi(\tau)) \le g(N+2)$, since the output has mean photon number at most $N+2$, and $I_c(\Phi,\tau) \ge -S(\tau) = -g(N)$, by subadditivity applied to a purification of $\tau$. Since $\epsilon = O(N^2)$ and $N' = N + O(N^2)$, we obtain $I_c(\Phi,\rho_{\text{th},N}) = I_c(\Phi,\rho_{N'}) + O(N^2\log_2(1/N)) = \Delta F\,N\log_2(1/N) + O(N)$, which completes the proof.

\emph{Signs.} When $\Delta F > 0$ the first term dominates, so $I_c(\Phi,\rho_{\text{th},N}) > 0$ for every sufficiently small $N$; since $\rho_{\text{th},N}$ has mean photon number $N$, this holds under any energy constraint. When $\Delta F < 0$, $I_c(\Phi,\rho_{\text{th},N}) < 0$ for every sufficiently small $N$. For $\Delta F = 0$ the leading term vanishes and the lemma does not decide the sign. The approach to the leading term is slow, with a relative correction of order $1/\log_2(1/N)$.

\section{Two-Way Capacities: Multi-Rail Protocol and Proof of Theorem 1}\label{app:multirail_bound}

\noindent We detail the derivation of the lower bound on the energy-constrained capacity $Q_2(\Phi,\bar{N})$ based on the multi-rail multi-photon protocol discussed in the main text in Sec.~\ref{sec:improved Q2}, and we prove Theorem~\ref{thm:exact_two_way} in the last subsection. This approach adapts the strategy from Ref.~\cite{Mele2024LossDephasing} to the case of fading channels. The essential difference is that the $M$ modes of a block experience independent, unknown transmissivities, whereas in Ref.~\cite{Mele2024LossDephasing} all modes undergo the same loss, so that the post-selection can be described as the probabilistic simulation of a dephasing channel. For independent transmissivities and $K \ge 2$ the acceptance probability depends on the input state, so no such simulation is available, and we bound the rate by distilling the accepted state directly.

\subsection{Protocol Description}
Alice and Bob use $M$ uses of the channel $\Phi$ to share a target entangled state $\ket{\Psi}_{AB}$.
\begin{itemize}
\item \textbf{Encoding:} Alice prepares a maximally entangled state $\ket{\Psi}_{AB} = \frac{1}{\sqrt{D}} \sum_{k=1}^D \ket{k}_A \ket{k}_B$, where the system $B$ (to be sent through the channel) is encoded in the subspace $\mathcal{H}_{K,M}$ of $K$ photons distributed across $M$ bosonic modes. The dimension of this subspace is $D = D_{K,M} = \binom{M+K-1}{K}$. The basis states $\ket{k}_B$ correspond to the Fock states $\ket{\vec{n}} = \ket{n_1, \dots, n_M}$ with $\sum_{i=1}^M n_i = K$.
\item \textbf{Transmission:} Alice sends the $M$ modes of system $B$ through $M$ independent realizations of the fading channel $\Phi$. The channel acts as $\Phi^{\otimes M}(\rho) = \int d\vec{\lambda} \, p(\vec{\lambda}) \, \mathcal{E}_{\vec{\lambda}}(\rho)$, where $\mathcal{E}_{\vec{\lambda}} = \bigotimes_{i=1}^M \mathcal{E}_{\lambda_i}$. The $M$ transmissivities are independent, $p(\vec{\lambda}) = \prod_{i=1}^M p(\lambda_i)$, and the realization $\vec{\lambda}$ is unknown to Alice and Bob.
\item \textbf{Post-selection:} Bob measures the total photon number $\hat{N}_{tot} = \sum_{i=1}^M \hat{n}_i$ on the received modes, i.e., he performs the projective measurement $\{\mathcal{P}_K, \mathbbm{1} - \mathcal{P}_K\}$, where $\mathcal{P}_K$ projects onto $\mathcal{H}_{K,M}$. The measurement does not resolve the photon numbers of the individual modes and preserves the coherences within $\mathcal{H}_{K,M}$.
\begin{itemize}
    \item If the outcome is $K$, the block is accepted. Since pure-loss channels never increase the photon number, this outcome excludes any photon loss. It does not imply that the encoded state is preserved: for a fixed realization $\vec{\lambda}$ the no-loss branch changes the relative amplitudes of the code states unless all $\lambda_i$ coincide, and, since $\vec{\lambda}$ is unknown, for every $p(\lambda)$ that is not a point mass the accepted state is mixed (see below).
    \item If the outcome is smaller than $K$, the block is discarded.
\end{itemize}
\end{itemize}
\noindent\textbf{Entanglement distillation:} Bob announces which blocks were accepted; this announcement, the subsequent distillation protocol, and the final teleportation step below are all free of cost under two-way classical communication, by the definition of $Q_2$. Alice and Bob apply an entanglement distillation protocol collectively to the accepted states (one state per accepted block). Since the channel is memoryless, Alice's input is a product state over blocks and Bob's measurement acts separately on each block; conditioned on the announced outcomes, the accepted blocks are independent and identically distributed copies of $\omega_{AB}$ (Eq.~\eqref{eq:mr_state} below), and, by the law of large numbers, their number concentrates around $P_{\text{succ}}$ times the total number of blocks as the latter grows. The distilled ebits are converted into transmitted qubits by teleportation, so the resulting rate per channel use lower bounds $Q_2(\Phi)$~\cite{PLOB}, and also $Q_2(\Phi,\bar{N})$ for every $\bar{N}$ not smaller than the mean photon number per use of the code.
\subsection{Success Probability and Rate}
We compute $P_{\text{succ}}$. For a single use of a pure-loss channel $\mathcal{E}_\lambda$ acting on a Fock state $\ket{n}$, the probability of transmitting all $n$ photons is $\lambda^n$.
For the fading channel $\Phi$, which is a convex combination of lossy channels $\mathcal{E}_\lambda$, the probability of transmitting a Fock state $\ket{\vec{n}}$ without loss is the average over the fading distribution:
\begin{equation}
    P(\vec{n} \to \vec{n}) = \int \left( \prod_{i=1}^M p(\lambda_i) \lambda_i^{n_i} \right) d\vec{\lambda} = \prod_{i=1}^M \int p(\lambda_i) \lambda_i^{n_i} d\lambda_i = \prod_{i=1}^M \langle \lambda^{n_i} \rangle.
\end{equation}
Note that here we assumed independent fading for each channel use (fast fading). This is the setting in which $\Phi^{\otimes M}$ describes $M$ uses of the memoryless channel $\Phi$. If instead the transmissivity is constant over the $M$ uses (slow fading), the block is described by the different channel $\int p(\lambda)\, \mathcal{E}_\lambda^{\otimes M} d\lambda$, and the average would be $\langle \lambda^K \rangle$.\\

\noindent The post-selection operation performed by Bob is described by the projector onto the subspace of total photon number $K$. We denote this projector as $\mathcal{P}_K$, defined explicitly as:
\begin{equation}
    \mathcal{P}_K = \sum_{\substack{n_1, \dots, n_M \\ \sum n_i = K}} \ket{n_1, \dots, n_M} \bra{n_1, \dots, n_M}.
\end{equation}
This projector selects only those events where the total photon number is conserved, filtering out instances where photon loss occurred. The state sent by Alice is the maximally mixed state on $\mathcal{H}_{K,M}$, i.e., $\rho_{in} = \frac{1}{D} \sum_{\vec{n} \in \mathcal{H}_{K,M}} \ket{\vec{n}}\bra{\vec{n}}$. The total success probability is the average of the transmission probabilities of the basis states:
\begin{equation}
    P_{\text{succ}} = \text{Tr}\left[ \mathcal{P}_K \Phi^{\otimes M}(\rho_{in}) \right] = \frac{1}{D_{K,M}} \sum_{\vec{n} \in \mathcal{H}_{K,M}} \prod_{i=1}^M \langle \lambda^{n_i} \rangle.
\end{equation}
This formula uses the specific moments of the fading distribution.

\noindent\textbf{Accepted state.} The projection $\mathcal{P}_K$ removes every Kraus branch in which at least one photon is lost. For a realization $\vec{\lambda}$, the remaining no-loss branch acts on the code space as
\begin{equation}\label{eq:mr_noloss}
    A_0(\vec{\lambda}) = \bigotimes_{i=1}^M \lambda_i^{\hat{n}_i/2}, \qquad A_0(\vec{\lambda}) \ket{\vec{n}} = \Big( \prod_{i=1}^M \lambda_i^{n_i/2} \Big) \ket{\vec{n}},
\end{equation}
which is proportional to the identity on $\mathcal{H}_{K,M}$ only if all $\lambda_i$ coincide. Since the realization is unknown to Alice and Bob, the unnormalized branches must be averaged over the independent transmissivities, and the result normalized only at the end. With the moments $m_s = \langle \lambda^s \rangle$ ($s \ge 0$, $m_0 = 1$), the unnormalized accepted state is $\frac{1}{D_{K,M}} \sum_{\vec{n},\vec{n}'} G_{\vec{n},\vec{n}'} \ket{\vec{n}}\bra{\vec{n}'}_A \otimes \ket{\vec{n}}\bra{\vec{n}'}_B$, where the sums run over the configurations of $\mathcal{H}_{K,M}$ and
\begin{equation}\label{eq:mr_gram}
    G_{\vec{n},\vec{n}'} = \prod_{i=1}^M m_{(n_i + n_i')/2}.
\end{equation}
The matrix $G$ is positive semidefinite, being the average of the rank-one matrices $v_{\vec{\lambda}} v_{\vec{\lambda}}^{\dagger}$ with $(v_{\vec{\lambda}})_{\vec{n}} = \prod_i \lambda_i^{n_i/2}$. Its trace gives back the success probability,
\begin{equation}\label{eq:mr_psucc}
    P_{\text{succ}} = \frac{\text{Tr}\, G}{D_{K,M}},
\end{equation}
and, for $P_{\text{succ}} > 0$, the normalized accepted state is
\begin{equation}\label{eq:mr_state}
    \omega_{AB} = \sum_{\vec{n},\vec{n}'} \tau_{\vec{n},\vec{n}'} \ket{\vec{n}}\bra{\vec{n}'}_A \otimes \ket{\vec{n}}\bra{\vec{n}'}_B, \qquad \tau = \frac{G}{\text{Tr}\, G}.
\end{equation}
This state is maximally correlated. The isometry $\ket{\vec{n}} \mapsto \ket{\vec{n}}_A \ket{\vec{n}}_B$ shows that $S(\omega_{AB}) = S(\tau)$, while the reduced state $\omega_B$ is diagonal with eigenvalues $q_{\vec{n}} = \tau_{\vec{n},\vec{n}} = \prod_i m_{n_i} / \text{Tr}\, G$, so that $S(\omega_B) = H(q)$, the Shannon entropy of $q$.

\noindent\textbf{Rate.} By the hashing inequality~\cite{DevetakWinter2005}, entanglement can be distilled from independent copies of $\omega_{AB}$ at the rate $S(\omega_B) - S(\omega_{AB}) = H(q) - S(\tau)$ ebits per copy. This quantity is non-negative, because dephasing $\tau$ in the basis $\{\ket{\vec{n}}\}$ yields the diagonal state with eigenvalues $q$, and dephasing cannot decrease the entropy. For maximally correlated states it coincides with the distillable entanglement under two-way LOCC~\cite{LeditzkyDattaSmith2018}, so no higher rate can be extracted from the accepted states. Each block uses the channel $M$ times and is accepted with probability $P_{\text{succ}}$; hence
\begin{equation}\label{eq:mr_rate}
    R_{K,M} = \frac{P_{\text{succ}}}{M} \left[ H(q) - S(\tau) \right]
\end{equation}
is an achievable rate. The code has mean photon number $K/M$ per mode, so $R_{K,M}$ lower bounds the energy-constrained capacity $Q_2(\Phi,\bar{N})$ for every $\bar{N} \ge K/M$ (cf.~Ref.~\cite{Mele2024LossDephasing}, Eq.~(D18)), and a fortiori $Q_2(\Phi)$, which for a finite support is however given exactly by Theorem~\ref{thm:exact_two_way}. Optimizing over the codes that meet the energy constraint gives
\begin{equation}\label{eq:Q2_mr_general}
    Q_2(\Phi,\bar{N}) \ge \sup_{\substack{K \ge 1,\, M \ge 2 \\ K/M \le \bar{N}}} \frac{P_{\text{succ}}}{M} \left[ H(q) - S(\tau) \right].
\end{equation}
The optimality mentioned above concerns only the distillation of the accepted states: it does not exclude higher rates from other inputs, or from protocols that also exploit the blocks in which some photons are lost.

\noindent\textbf{Fluctuating versus constant transmissivity.} The bracket in Eq.~\eqref{eq:mr_rate} satisfies $H(q) - S(\tau) \le \log_2 D_{K,M}$, with equality if and only if $\tau$ is pure with a uniform diagonal. Since the entries of $G$ are non-negative, this happens if and only if all entries of $G$ are equal, i.e., if and only if $\omega_{AB}$ is maximally entangled. If $p(\lambda) = \delta(\lambda - \lambda_0)$, then all entries of $G$ equal $\lambda_0^K$ and $R_{K,M} = (\lambda_0^K/M) \log_2 D_{K,M}$, the multi-rail rate of Ref.~\cite{Mele2024LossDephasing} for a pure-loss channel. Conversely, let $M \ge 2$ and $\langle \lambda \rangle > 0$. For $K = 1$, the diagonal entries of $G$ equal $\langle \lambda \rangle$ and the off-diagonal ones equal $\langle \sqrt{\lambda} \rangle^2$, which coincide only if $\text{Var}(\sqrt{\lambda}) = 0$. For $K \ge 2$, the diagonal entries at the configurations $(K,0,\dots,0)$ and $(K-1,1,0,\dots,0)$ are $\langle \lambda^K \rangle$ and $\langle \lambda^{K-1} \rangle \langle \lambda \rangle$, and
\begin{equation}
    \langle \lambda^K \rangle - \langle \lambda^{K-1} \rangle \langle \lambda \rangle = \frac{1}{2} \big\langle (\lambda^{K-1} - \lambda'^{K-1})(\lambda - \lambda') \big\rangle > 0
\end{equation}
unless $\lambda$ is almost surely constant, where $\lambda'$ is an independent copy of $\lambda$ and the integrand is non-negative because $x \mapsto x^{K-1}$ is increasing on $[0,1]$. Hence, for every fluctuating transmissivity the accepted state is not maximally entangled, and $R_{K,M} < (P_{\text{succ}}/M) \log_2 D_{K,M}$. The rate $(\langle \lambda^K \rangle / M) \log_2 D_{K,M}$ is obtained instead if all $M$ modes of a block share the same transmissivity (slow fading); that block channel, $\int p(\lambda)\, \mathcal{E}_\lambda^{\otimes M} d\lambda$, differs from $\Phi^{\otimes M}$ and is not considered here.

\noindent\textbf{Positivity.} If $\langle \lambda \rangle > 0$, every moment $m_s$ with $s > 0$ is strictly positive, so all entries of $G$ are strictly positive. Since $H(q) - S(\tau) = D(\tau \| \Delta(\tau))$, where $\Delta$ denotes dephasing in the basis $\{\ket{\vec{n}}\}$ and $D(\cdot\|\cdot)$ is the quantum relative entropy, the bracket in Eq.~\eqref{eq:mr_rate} vanishes only if $\tau = \Delta(\tau)$; since $\tau$ is not diagonal, $R_{K,M} > 0$ for all $K \ge 1$ and $M \ge 2$. In particular, a single photon on $M$ rails has mean photon number $1/M$ per mode, so $Q_2(\Phi,\bar{N}) \ge R_{1,M} > 0$ for every $\bar{N} > 0$, choosing any $M \ge \max\{2, 1/\bar{N}\}$.

\subsection{Single-Photon Encoding}
\noindent For a single photon ($K = 1$), the code space is spanned by the states $\ket{i}$, $i = 1, \dots, M$, with the photon in mode $i$. Then $G_{ii} = \langle \lambda \rangle$ and $G_{ij} = \langle \sqrt{\lambda} \rangle^2$ for $i \neq j$, so that $P_{\text{succ}} = \langle \lambda \rangle$ and
\begin{equation}\label{eq:mr_k1_state}
    \tau = \frac{(1-c)\, \mathbbm{1} + c J}{M}, \qquad c = \frac{\langle \sqrt{\lambda} \rangle^2}{\langle \lambda \rangle} = 1 - \frac{\text{Var}(\sqrt{\lambda})}{\langle \lambda \rangle},
\end{equation}
where $J$ is the $M \times M$ matrix with all entries equal to one. Equivalently, conditioned on the detection of the photon, the $M$ channel uses act on the code space as the Hadamard channel $\rho \mapsto C \circ \rho$, with $C = (1-c)\mathbbm{1} + cJ$ and $\circ$ the entrywise product, and the success probability $\langle \lambda \rangle$ does not depend on the input. This is the structure found in Ref.~\cite{Mele2024LossDephasing} for the loss-dephasing channel, with the dephasing factor between two rails replaced by $c$: the fluctuations of the transmissivity act as an effective dephasing between the rails. The marginal $q$ is uniform, and the eigenvalues of $\tau$ are $r_+ = [1 + (M-1)c]/M$ and $r_- = (1-c)/M$, the latter with multiplicity $M-1$. Hence
\begin{equation}\label{eq:mr_k1_rate}
    R_{1,M} = \frac{\langle \lambda \rangle}{M} \left[ \log_2 M + r_+ \log_2 r_+ + (M-1)\, r_- \log_2 r_- \right].
\end{equation}
For $M = 2$ the channel above is the qubit dephasing channel $\rho \mapsto (1-\epsilon)\rho + \epsilon Z \rho Z$ with $\epsilon = (1-c)/2$, and
\begin{equation}\label{eq:mr_dualrail}
    R_{1,2} = \frac{\langle \lambda \rangle}{2} \left[ 1 - H_2\!\left( \frac{1+c}{2} \right) \right],
\end{equation}
which is Eq.~\eqref{eq:Q2_mr_dualrail} of the main text. For example, if $\lambda \in \{0, 1\}$ with probability $1/2$ each, then $c = 1/2$ and $R_{1,2} = [1 - H_2(3/4)]/4 \simeq 0.047$, whereas a maximally entangled accepted state would give $1/4$.

\noindent\textbf{A bound depending only on the mean transmissivity.} Since $\sqrt{\lambda} \ge \lambda$ on $[0,1]$ and, by the Cauchy--Schwarz inequality, $\langle \sqrt{\lambda} \rangle^2 \le \langle \lambda \rangle$, one has $\langle \lambda \rangle \le c \le 1$. The function $x \mapsto 1 - H_2((1+x)/2)$ is increasing on $[0,1]$, so Eq.~\eqref{eq:mr_dualrail} implies
\begin{equation}\label{eq:mr_mean_only}
    Q_2(\Phi,\bar{N}) \ge \frac{\langle \lambda \rangle}{2} \left[ 1 - H_2\!\left( \frac{1 + \langle \lambda \rangle}{2} \right) \right] > 0 \qquad \text{for } \langle \lambda \rangle > 0, \ \bar{N} \ge \tfrac{1}{2}.
\end{equation}
For the erasure-like channel, $\lambda \in \{0,1\}$, one has $c = \langle \lambda \rangle$ and the right-hand side coincides with $R_{1,2}$. This bound certifies that $Q_2(\Phi,\bar{N}) > 0$, and hence $Q_2(\Phi) > 0$, for every fading channel other than the complete-loss channel, in agreement with Sec.~\ref{section: distillability}, but it is quantitatively weak at high loss: since $1 - H_2((1+x)/2) = x^2/(2 \ln 2) + O(x^4)$, its right-hand side behaves as $\langle \lambda \rangle^3/(4 \ln 2)$ for $\langle \lambda \rangle \ll 1$.
\subsection{Proof of Theorem 1}
\label{subsec:two_way_discrete}

\noindent We now prove Theorem~\ref{thm:exact_two_way}, which determines exactly the two-way quantum and secret-key capacities of a fading channel with a finite set of transmissivities. We restate it for convenience.

\smallskip
\noindent\textbf{Theorem~\ref{thm:exact_two_way}} (restated)\textbf{.} {\itshape Let
\begin{equation*}
    \Phi=\sum_{n=1}^{d}p_n\mathcal E_{\lambda_n},
    \qquad
    p_n>0,
    \qquad
    \sum_{n=1}^{d}p_n=1,
\end{equation*}
be a finite convex combination of pure-loss channels with
transmissivities $\lambda_n\in[0,1]$, acting independently on each channel use, with a realization that is revealed to neither party. Then, with no constraint on the input energy,
\begin{equation*}
    Q_2(\Phi)=K(\Phi)
    =
    \sum_{n=1}^{d}
    p_n\log_2\frac{1}{1-\lambda_n},
\end{equation*}
which is infinite if some $\lambda_n = 1$. In particular, the two-way quantum and secret-key capacities of a finite convex combination of pure-loss channels are given by the corresponding convex combinations of the capacities of the individual pure-loss channels, namely $Q_2\!\left(\sum_{n=1}^{d} p_n \mathcal E_{\lambda_n}\right)=\sum_{n=1}^{d} p_n\, Q_2(\mathcal E_{\lambda_n})$ and, analogously, $K\!\left(\sum_{n=1}^{d} p_n \mathcal E_{\lambda_n}\right)=\sum_{n=1}^{d} p_n\, K(\mathcal E_{\lambda_n})$.}
\smallskip

\noindent As in Sec.~\ref{sec:improved Q2} and Ref.~\cite{PLOB} (Supplementary Note 7), the capacities without an energy constraint are the limits $\bar{N}\to\infty$ of the capacities $Q_2(\Phi,\bar{N})$ and $K(\Phi,\bar{N})$ at mean input energy $\bar{N}$ per channel use.

\paragraph*{Idea of the proof.}
The only difficulty is that Alice and Bob do not know which pure-loss
channel has acted. Alice adds a large displacement to the mode carrying
the entanglement, so that different transmissivities produce well-separated
output amplitudes. Bob measures only a small fraction of the received
light to identify the transmissivity and removes the corresponding
displacement from the rest. He then groups together the states associated
with the same estimated transmissivity and distills each group separately.
As the displacement increases, the identification error vanishes; as the
measured fraction decreases, the additional loss introduced by the
measurement also vanishes. This recovers the average of the pure-loss
capacities.

\paragraph*{Protocol.}
We first combine terms with equal transmissivity, so that the
$\lambda_n$ are distinct. Fix $N,\alpha>0$ and $t\in(0,1)$, and define
\begin{equation}
    \eta_n=(1-t)\lambda_n,
    \qquad
    \mu_n=\sqrt{2t\lambda_n}\,\alpha .
\end{equation}
Alice prepares the two-mode squeezed vacuum
\begin{equation}\label{eq:n1_tmsv}
    \ket{\psi_N}_{AS}
    =
    \frac{1}{\sqrt{N+1}}
    \sum_{k=0}^{\infty}
    \left(\frac{N}{N+1}\right)^{k/2}
    \ket{k}_A\ket{k}_S,
\end{equation}
whose two modes have mean photon number $N$. For real $\alpha$, we write
\begin{equation}
    D(\alpha)=\exp[\alpha(a^\dagger-a)] .
\end{equation}
The protocol is summarized in Table~\ref{tab:discrete_fading_protocol}.

\begin{table}[htbp]
\caption{Protocol for a finite fading mixture.}
\label{tab:discrete_fading_protocol}
\centering
\small
\renewcommand{\arraystretch}{1.15}
\begin{ruledtabular}
\begin{tabular}{p{0.20\linewidth}p{0.73\linewidth}}
\textbf{Step} & \textbf{Operation} \\
\hline
\textbf{1. Send}
&
Alice keeps mode $A$, applies $D(\alpha)$ to mode $S$, and sends $S$
through $\Phi$. The mean photon number of the transmitted mode is
$N+\alpha^2$.
\\[5pt]

\textbf{2. Tap}
&
Bob mixes the received mode with vacuum on a beam splitter with real
coefficients, transmissivity $1-t$ towards output mode $B$ and reflectivity $t$
towards output mode $C$. Thus, a fraction $1-t$ of the received optical
power is retained in $B$, while a fraction $t$ is sent to $C$.
Bob keeps $B$ and homodynes
$\hat q=(a+a^\dagger)/\sqrt{2}$ on $C$, obtaining an outcome $Y$.
\\[5pt]

\textbf{3. Identify}
&
Bob estimates which component acted by choosing the closest expected
homodyne value,
\begin{equation*}
    \hat X
    =
    \operatorname*{arg\,min}_{1\le n\le d}
    |Y-\mu_n| .
\end{equation*}
\\[5pt]

\textbf{4. Correct}
&
Bob applies
$D(-\sqrt{\eta_{\hat X}}\,\alpha)$ to $B$ and announces $\hat X$ to
Alice. The homodyne outcome $Y$ is discarded.
\\[5pt]

\textbf{5. Distill}
&
After many repetitions, Alice and Bob group the shared states according
to the announced value $\hat X=n$ and apply entanglement hashing
separately to each group (at a fixed input energy, after the local
truncation described in the proof).
\\
\end{tabular}
\end{ruledtabular}
\end{table}

\begin{proof}
The converse is known~\cite[Eq.~(22)]{PLOB}, \cite[Eq.~(75)]{PLB19}:
\begin{equation}\label{eq:n1_converse}
    Q_2(\Phi)
    \leq
    K(\Phi)
    \leq
    \sum_{n=1}^{d}
    p_n\log_2\frac{1}{1-\lambda_n}.
\end{equation}
The bound holds for protocols with any energy constraint, and hence in the limit $\bar{N}\to\infty$ (Ref.~\cite{PLOB}, Methods).
It remains to prove the matching lower bound on $Q_2$.
The case $d=1$ is the pure-loss result~\cite{PLOB}, so we assume
$d\geq2$.

\paragraph*{1. Bob identifies the transmissivity.}
Let $X\in\{1,\ldots,d\}$ denote the component that acts, with
$\Pr(X=n)=p_n$. Pure-loss channels satisfy
\begin{equation}\label{eq:n1_displacement}
    \mathcal E_\lambda
    \!\left(
        D(\alpha)\rho D^\dagger(\alpha)
    \right)
    =
    D(\sqrt{\lambda}\alpha)
    \mathcal E_\lambda(\rho)
    D^\dagger(\sqrt{\lambda}\alpha).
\end{equation}
Therefore, conditioned on $X=n$, the homodyne outcome is Gaussian:
\begin{equation}\label{eq:n1_homodyne}
    Y\mid X=n
    \sim
    \mathcal N
    \left(
        \mu_n,\,
        t\lambda_nN+\frac12
    \right).
\end{equation}

Since the $\lambda_n$ are distinct,
\begin{equation}
    \delta
    :=
    \min_{n\neq n'}
    |\sqrt{\lambda_n}-\sqrt{\lambda_{n'}}|
    >0 .
\end{equation}
The distance between any two homodyne means is at least
$\sqrt{2t}\,\alpha\delta$. Hence, a nearest-mean decision can be wrong
only if
\begin{equation}
    |Y-\mu_n|
    \geq
    \alpha\sqrt{\frac{t}{2}}\,\delta .
\end{equation}
Since the variance in Eq.~\eqref{eq:n1_homodyne} is at most
$tN+1/2$, the Gaussian tail bound gives
\begin{equation}\label{eq:n1_error}
    e_\alpha
    :=
    \Pr(\hat X\neq X)
    \leq
    2\exp
    \left[
        -\frac{t\alpha^2\delta^2}{4tN+2}
    \right].
\end{equation}
Thus $e_\alpha\to0$ as $\alpha\to\infty$, for every fixed $N$ and $t$.

\paragraph*{2. The ideal output.}
Consider, only for the analysis, an ideal decoder that performs the same
measurement as Bob but uses the true value $X=n$ to choose the inverse
displacement. After the tap, the mode retained by Bob has transmissivity
\begin{equation}
    \eta_n=(1-t)\lambda_n .
\end{equation}
If Bob removes the correct displacement and the homodyne outcome is
discarded, the state left on $AB$ is therefore
\begin{equation}\label{eq:n1_rho_n}
    \rho_{N,\eta_n}
    :=
    (\operatorname{id}\otimes\mathcal E_{\eta_n})
    (\ket{\psi_N}\bra{\psi_N}).
\end{equation}
This is the usual finite-energy Choi state of the pure-loss channel
$\mathcal E_{\eta_n}$.

Keeping a classical copy of the component label on both sides, the
output of the ideal decoder is
\begin{equation}\label{eq:n1_ideal_state}
    \Omega
    =
    \sum_{n=1}^{d}
    p_n
    \ket{n,n}\bra{n,n}
    \otimes
    \rho_{N,\eta_n}.
\end{equation}

Let $\widetilde{\Omega}_\alpha$ denote the output of the actual decoder.
To compare the two outputs, let $\omega_n$ be the state of $A$, $B$ and
the tapped mode $C$ before the homodyne measurement when $X=n$, let
$\Pi_k$ be the homodyne effect of the decision region of $k$, and let
$\sigma_{n,k}=\text{Tr}_C[(\mathbbm{1}\otimes\Pi_k)\,\omega_n]$. With
$D_k=D(-\sqrt{\eta_k}\,\alpha)$ acting on $B$,
\begin{equation*}
    \widetilde{\Omega}_\alpha
    =
    \sum_{n,k}p_n\ket{k,k}\bra{k,k}\otimes D_k\sigma_{n,k}D_k^\dagger,
    \qquad
    \Omega
    =
    \sum_{n,k}p_n\ket{n,n}\bra{n,n}\otimes D_n\sigma_{n,k}D_n^\dagger,
\end{equation*}
where the second expression equals Eq.~\eqref{eq:n1_ideal_state} because
$\sum_k\sigma_{n,k}=\text{Tr}_C\,\omega_n$ and $D_n$ does not depend on
the outcome. The terms with $k=n$ coincide: the two decoders agree
whenever $\hat X=X$. The remaining terms have total trace $e_\alpha$ in
each sum. Hence
\begin{equation}\label{eq:n1_trace_distance}
    \frac12
    \left\|
        \widetilde{\Omega}_\alpha-\Omega
    \right\|_1
    \leq
    e_\alpha .
\end{equation}

\paragraph*{3. Sorting and distillation.}
Suppose first that Alice and Bob share $m$ independent copies of the
ideal state $\Omega$. Since both parties know the classical label $n$,
they can simply sort the copies according to this label. By the law of
large numbers, the number of copies with label $n$ is asymptotically
$m p_n$.

For each $n$, Alice and Bob therefore share approximately $m p_n$
copies of $\rho_{N,\eta_n}$. The reverse coherent information~\cite{GarciaPatron2009} of this
state is
\begin{equation}\label{eq:n1_rci}
    r(N,\eta_n)
    =
    S((\rho_{N,\eta_n})_A)
    -
    S(\rho_{N,\eta_n}).
\end{equation}
Alice's marginal is thermal with mean photon number $N$, while the
environment of the pure-loss channel is thermal with mean photon number
$(1-\eta_n)N$; since the joint state of $A$, $B$ and the environment is
pure, $S(\rho_{N,\eta_n})$ is the entropy of the environment. Thus
\begin{equation}\label{eq:n1_rci_explicit}
    r(N,\eta_n)
    =
    g(N)-g((1-\eta_n)N),
\end{equation}
where
\begin{equation}
    g(x)
    =
    (x+1)\log_2(x+1)-x\log_2x .
\end{equation}
By the hashing inequality~\cite{DevetakWinter2005}, which extends to
bosonic states of bounded energy such as $\rho_{N,\eta_n}$ (Supplementary
Information of Ref.~\cite{PLOB}), this is an achievable
entanglement-distillation rate for each group, with classical
communication from Bob to Alice only. Distilling
the $d$ groups separately therefore gives the total rate
\begin{equation}\label{eq:n1_ideal_rate}
    R_{N,t}
    =
    \sum_{n=1}^{d}
    p_n
    \left[
        g(N)-g((1-\eta_n)N)
    \right].
\end{equation}

It remains to return from the ideal decoder to the actual one.
Since $\Omega$ does not depend on $\alpha$, the protocol for
$\Omega^{\otimes m}$ can be fixed first: for every $R<R_{N,t}$ it produces
$\lfloor mR\rfloor$ ebits with an error $\varepsilon_m\to0$ (the group
sizes concentrate around $mp_n$, hashing succeeds for every sufficiently
large group, and surplus ebits are discarded).
For $m$ independent repetitions, Eq.~\eqref{eq:n1_trace_distance} and
the telescoping inequality give
\begin{equation}
    \frac12
    \left\|
        \widetilde{\Omega}_\alpha^{\otimes m}
        -
        \Omega^{\otimes m}
    \right\|_1
    \leq
    m e_\alpha .
\end{equation}
Since $e_\alpha$ decays exponentially in $\alpha^2$, we can choose a
sequence $\alpha=\alpha_m$ such that
\begin{equation}
    m e_{\alpha_m}\longrightarrow0 .
\end{equation}
For instance, it is enough to take $\alpha_m^2$ proportional to
$\log m$ with a sufficiently large constant.

Alice and Bob now apply to the actual outputs exactly the same sorting
and hashing protocol used for the ideal outputs. Trace distance cannot
increase under this protocol, so the final state produced from
$\widetilde{\Omega}_{\alpha_m}^{\otimes m}$ approaches the same maximally
entangled state as for $\Omega^{\otimes m}$, with an error at most
$\varepsilon_m+m e_{\alpha_m}$.

This protocol uses the input energy $N+\alpha_m^2$, which grows with $m$,
so by itself it does not bound the capacities as defined in
Sec.~\ref{sec:improved Q2}, which are the limits of the
energy-constrained ones. We therefore show that the same rate is
approached at a fixed input energy. Fix $\alpha$, and let
$\mathcal{C}_L$ be the channel that projects a mode onto the span of
$\ket{0},\dots,\ket{L}$ and sends the rest to the vacuum: it is completely
positive and trace preserving, acts locally, and does not increase the
mean photon number. Alice and Bob apply $\mathcal{C}_L$ to their modes and
then the hashing protocol directly to the independent copies of the
resulting finite-dimensional state, with classical communication from
Bob to Alice only. For a state $\omega$ with classical labels on both
sides write $I_\leftarrow(\omega)=S(\omega_A)-S(\omega_{AB})$, where $A$
and $B$ include the labels; the achievable rate is
$I_\leftarrow((\mathcal{C}_L\otimes\mathcal{C}_L)(\widetilde{\Omega}_\alpha))$.
In $\Omega$ the labels are perfectly correlated and their entropy
appears in both terms and cancels, so $I_\leftarrow(\Omega)=R_{N,t}$, and
the energy-constrained continuity bound of Ref.~\cite{Winter2016} shows
that $I_\leftarrow((\mathcal{C}_L\otimes\mathcal{C}_L)(\Omega))\to R_{N,t}$ as
$L\to\infty$. At a fixed $L$ each party holds a system of dimension
$(L+1)d$, and the finite-dimensional continuity bound of the same work,
applied to $-S(B|A)$, together with Eq.~\eqref{eq:n1_trace_distance} and
the monotonicity of the trace distance under
$\mathcal{C}_L\otimes\mathcal{C}_L$, shows that
$I_\leftarrow((\mathcal{C}_L\otimes\mathcal{C}_L)(\widetilde{\Omega}_\alpha))$
tends to $I_\leftarrow((\mathcal{C}_L\otimes\mathcal{C}_L)(\Omega))$ as
$\alpha\to\infty$. The cutoff must be fixed before
$\alpha$, so that the continuity bound used at the end is dimensional and
does not see the energy of the misidentified branches. The mean input
energy is $N+\alpha^2$ for every $m$, and the distilled ebits are
converted into transmitted qubits by teleportation, which uses the free
forward communication. Hence, for every $L$,
$\lim_{\bar{N}\to\infty}Q_2(\Phi,\bar{N})\geq I_\leftarrow((\mathcal{C}_L\otimes\mathcal{C}_L)(\Omega))$,
and letting $L\to\infty$ gives
$\lim_{\bar{N}\to\infty}Q_2(\Phi,\bar{N})\geq R_{N,t}$, that is,
\begin{equation}\label{eq:n1_achievable}
    Q_2(\Phi)
    \geq
    \sum_{n=1}^{d}
    p_n
    \left[
        g(N)
        -
        g\!\left(
            [1-(1-t)\lambda_n]N
        \right)
    \right].
\end{equation}

\paragraph*{4. Taking the limits.}
For every $a>0$,
\begin{equation}
    \lim_{N\to\infty}
    \big[g(N)-g(aN)\big]
    =
    -\log_2a .
\end{equation}
At fixed $t>0$,
\begin{equation}
    1-(1-t)\lambda_n\geq t>0,
\end{equation}
and therefore Eq.~\eqref{eq:n1_achievable} gives, as $N\to\infty$,
\begin{equation}
    Q_2(\Phi)
    \geq
    \sum_{n=1}^{d}
    p_n
    \log_2
    \frac{1}{1-(1-t)\lambda_n}.
\end{equation}
Finally, letting $t\to0$,
\begin{equation}
    Q_2(\Phi)
    \geq
    \sum_{n=1}^{d}
    p_n
    \log_2
    \frac{1}{1-\lambda_n}.
\end{equation}
Together with the converse in Eq.~\eqref{eq:n1_converse}, this proves
Eq.~\eqref{eq:exact_two_way}.
\end{proof}

\section{Analytical Solution for the Erasure-like Channel}\label{app:erasure_derivation}

\noindent We provide the detailed derivation of the quantum mutual information $I(\Phi_p,\rho) = S(\rho) + S(\Phi_p(\rho)) - S(\tilde{\Phi}_p(\rho))$ for the binary fading channel $\Phi_p = p \mathcal{I} + (1-p) \mathcal{E}_0$. Our goal is to determine the input state $\rho$ that maximizes this quantity under a mean energy constraint $\text{Tr}[\hat{n}\rho] \le E$.
Since the channel $\Phi_p$ is phase-insensitive, the mutual information is concave and invariant under phase rotations. As discussed in Appendix~\ref{appendix: diagonal optimizer}, this allows us to restrict the optimization to the set of Fock-diagonal states without loss of generality.
We parametrize a generic Fock-diagonal state $\rho$ as a convex combination of the vacuum state and a normalized state $\sigma$ supported on the orthogonal subspace ($\text{span}\{\ket{n}\}_{n \ge 1}$):
\begin{equation}
    \rho = p_0 \ket{0}\bra{0} + (1-p_0) \sigma.
\end{equation}
This decomposition is merely a rewriting of the state and does not introduce any restriction. The parameter $p_0 \in [0,1]$ represents the vacuum population, while $\sigma$ captures the distribution of the excited states. The energy constraint on $\rho$ translates to a constraint on the mean energy of $\sigma$:
\begin{equation}
    \text{Tr}[\hat{n}\rho] = (1-p_0) \text{Tr}[\hat{n}\sigma] \le E \implies \text{Tr}[\hat{n}\sigma] \le N := \frac{E}{1-p_0}.
\end{equation}

\subsection{Input Entropy}
\noindent For a Fock-diagonal state $\rho = \sum_n \rho_{nn} \ket{n}\bra{n}$, the von Neumann entropy is the Shannon entropy of its diagonal elements. Substituting the decomposition for $\rho$, we obtain:
\begin{equation}
    S(\rho) = H_2(p_0) + (1-p_0)S(\sigma),
\end{equation}
where $H_2(x) = -x\log_2 x - (1-x)\log_2(1-x)$ is the binary entropy function.

\subsection{Output Entropy}
\noindent The channel maps the input $\rho$ to the output state $\Phi_p(\rho) = p\rho + (1-p)\ket{0}\bra{0}$. The vacuum component is shifted due to the erasure term, while the excited components are simply scaled by $p$. The spectrum of the output state is:
\begin{equation}
    \text{spec}(\Phi_p(\rho)) = \{ p\rho_{00} + (1-p), p\rho_{11}, p\rho_{22}, \dots \}.
\end{equation}
Using the decomposition $\rho_{00} = p_0$ and the expression for $S(\rho)$, the output entropy is:
\begin{align}
    S(\Phi_p(\rho)) &= p S(\rho) - p \log_2 p + p\rho_{00}\log_2(p\rho_{00}) - (p\rho_{00} + 1-p)\log_2(p\rho_{00} + 1-p) \nonumber \\
    &= p H_2(p_0) + p(1-p_0)S(\sigma) - p \log_2 p + p p_0\log_2(p p_0) - (p p_0 + 1-p)\log_2(p p_0 + 1-p).
\end{align}

\subsection{Exchange Entropy}
\noindent To compute the exchange entropy $S(\tilde{\Phi}_p(\rho))$, we construct the Stinespring dilation and, from that, we can compute the Kraus operators of $\Phi$. The channel admits a Kraus representation with operators $\{N_k\}_{k=0}^\infty$ defined as:
\begin{equation}
    N_0 = \sqrt{p}\mathbbm{1}, \quad N_k = \sqrt{1-p}\ket{0}\bra{k-1} \text{ for } k \ge 1.
\end{equation}
Using the general formula for the complementary channel $\tilde{\Phi}(\rho) = \sum_{i,j} \ket{i}_E\bra{j} \text{Tr}[\rho N_j^\dagger N_i]$~\cite{WildeBook}, we find that for a Fock-diagonal input, the output state $\tilde{\Phi}_p(\rho)$ is block-diagonal. Specifically, in the environment Fock basis $\{\ket{k}_E\}$, it takes the form:
\begin{equation}
	\tilde{\Phi}_p(\rho)=
	\begin{pmatrix}
		p & \gamma & 0 & \dots \\
		\gamma & (1-p)\rho_{00} & 0 & \dots \\
		0 & 0 & (1-p)\rho_{11} & \dots \\
		\vdots & \vdots & \vdots & \ddots
	\end{pmatrix},
\end{equation}
where $\gamma = \rho_{00}\sqrt{p(1-p)}$. The only non-trivial block is the $2 \times 2$ matrix in the subspace spanned by the environment vacuum $\ket{0}_E$ and first excitation $\ket{1}_E$. Its eigenvalues $\lambda_\pm$ are given by:
\begin{equation}
    \lambda_\pm (\rho_{00})= \frac{1}{2} \left[ p + (1-p)\rho_{00} \pm \sqrt{p^2 - 2p(1-p)\rho_{00} + (1+2p-3p^2)\rho_{00}^2} \right].
\end{equation}
The remaining eigenvalues are simply $(1-p)\rho_{kk}$ for $k \ge 1$. Summing these entropy contributions and substituting the parametrized form of $\rho$:
\begin{align}
    S(\tilde{\Phi}_p(\rho)) &= (1-p)S(\rho) - (1-p)\log_2(1-p) + (1-p)\rho_{00}\log_2((1-p)\rho_{00}) - \sum_{\alpha \in \{+,-\}} \lambda_\alpha \log_2 \lambda_\alpha \nonumber \\
    &= (1-p)H_2(p_0) + (1-p)(1-p_0)S(\sigma) - (1-p)\log_2(1-p) \nonumber \\
    &\quad + (1-p)p_0\log_2((1-p)p_0) - \sum_{\alpha \in \{+,-\}} \lambda_\alpha(p_0) \log_2 \lambda_\alpha(p_0).
\end{align}

\subsection{Closed-Form Mutual Information and Optimization}
\noindent Combining the entropy terms derived above, the mutual information $I(\Phi_p, \rho)$ simplifies to:
\begin{align}\label{eq:mutual_info_general}
    I(\Phi_p, p_0, \sigma) &= 2p(1-p_0)S(\sigma) + 2p H_2(p_0) - p\log_2 p + (1-p)\log_2(1-p) \nonumber \\
    &\quad + p p_0\log_2(p p_0) - (p p_0 + 1-p)\log_2(p p_0 + 1-p) \nonumber \\
    &\quad - (1-p)p_0\log_2((1-p)p_0) + \sum_{\alpha \in \{+,-\}} \lambda_\alpha(p_0) \log_2 \lambda_\alpha(p_0).
\end{align}
Notice that for a fixed $p_0$, the mutual information depends on the state $\sigma$ only through the term $2p(1-p_0)S(\sigma)$, which is strictly increasing with $S(\sigma)$. To maximize the capacity, we must therefore maximize the entropy $S(\sigma)$ subject to the constraints imposed on $\sigma$.
Since $\sigma$ must be supported on the subspace $n \ge 1$ and has a fixed mean energy $\text{Tr}[\hat{n}\sigma] = N = E/(1-p_0)$, the optimal state $\sigma_{\text{opt}}$ is uniquely determined by the maximum entropy principle. This optimization problem is equivalent to finding the maximum entropy state of a harmonic oscillator with mean photon number $N' = N - 1$. Consequently, the optimal state $\sigma_{\text{opt}}$ is the thermal state of the harmonic oscillator shifted to the subspace $n \ge 1$ (a shifted geometric distribution):
\begin{equation}
    \sigma_{\text{opt}} = \frac{1}{N} \sum_{n=1}^\infty \left(\frac{N-1}{N}\right)^{n-1} \ket{n}\bra{n}.
\end{equation}
Its von Neumann entropy is exactly given by the bosonic entropy function evaluated at the shifted energy $N-1$:
\begin{equation}
    S(\sigma_{\text{opt}}) = g(N-1) = N\log_2 N - (N-1)\log_2(N-1).
\end{equation}
Substituting this derived optimal entropy into the expression for the mutual information, with $N=E/(1-p_0)$, the problem reduces to a scalar maximization over the parameter $p_0$:
\begin{equation}
    C_\text{E}(\Phi_p, E) = \max_{p_0 \in [0, 1]} I(\Phi_p, p_0, \sigma_{\text{opt}}),
\end{equation}
where all terms in $I$ are now explicit functions of $p_0$, $p$, and $E$.

\section{Computation of Entropic Functionals for Discretized Fading Channels}\label{app:entropic_functionals}

\noindent In this appendix, we detail the procedure to compute the quantum mutual information $I(\Phi, \rho) = S(\rho) + S(\Phi(\rho)) - S(\tilde{\Phi}(\rho))$ for the discretized fading channel. We consider the channel $\Phi_{\{p_n,\lambda_n\}}$ defined as a convex combination of pure-loss channels:
\begin{equation}\label{eq:app_channel_def}
    \Phi_{\{p_n,\lambda_n\}} = \sum_{n=1}^{d} p_n \mathcal{E}_{\lambda_n},\qquad
\sum_{n=0}^{d-1} p_n=1,
\end{equation}
where $\mathcal{E}_{\lambda}$ denotes the pure-loss channel with transmissivity $\lambda\in[0,1]$. The input states we treat are Fock-diagonal and constructed by the iterative prescription used in the numerical optimization described in the main text; we denote the $k$-th iterate by $\rho_E^{(k)}$.  All entropy expressions are the von Neumann entropy $S(\rho)=-\text{Tr}[\rho\log_2\rho]$.
Below we first recall the thermal (geometric) distribution and then compute the input entropy. Next we compute the action of each component channel on a Fock-diagonal state (giving closed forms for the output populations) and derive the output entropy of the mixture~\eqref{eq:app_channel_def}. Finally we derive the block structure of the complementary channel on Fock-diagonal inputs and give the corresponding expression for the exchange entropy. We close with a short, practical bound on the truncation error when cutting the Fock space at finite dimension.

\subsection{Notation and preliminary: thermal state and binomial kernel}
\noindent We begin by recalling the definition of a thermal state with mean photon number $E\ge0$:
\begin{equation}
    \rho_{\text{th},E} = \frac{1}{E+1} \sum_{n=0}^{\infty} \left( \frac{E}{E+1} \right)^n \ket{n} \bra{n}.
\end{equation}
and its von Neumann entropy is the well-known function
\[
g(E):=S(\rho_{\mathrm{th},E})=(E+1)\log_2(E+1)-E\log_2 E,
\]
with the convention $g(0)=0$.  For reference we also define the binomial redistribution kernel which appears when an $n$-photon Fock component passes through an attenuator of transmissivity $\lambda$:
\begin{equation}\label{eq:binomial}
B_m(n,\lambda):=\binom{n}{m}\lambda^m(1-\lambda)^{\,n-m},\qquad 0\le m\le n,
\end{equation}
and the action of the attenuator on a Fock-diagonal state $\rho=\sum_{n\ge0}\rho_{nn}|n\rangle\langle n|$ results in a Fock-diagonal state whose populations are
\begin{equation}\label{eq:attenuator-populations}
\big[\mathcal{E}_{\lambda}(\rho)\big]_{mm}
=\sum_{n=m}^\infty B_m(n,\lambda)\,\rho_{nn}.
\end{equation}
All manipulations below exploit this Fock-diagonality.

\subsection{Input State and Input Entropy}

\noindent At the $k$-th step of the optimization algorithm, the ansatz state $\rho_E^{(k)}$ is constructed as a mixture of the first $k$ Fock states and a ``shifted" thermal tail. Explicitly:
\begin{equation}\label{eq:app_input_state}
    \rho_E^{(k)} = \sum_{n=0}^{k-1} w_n \ket{n}\bra{n} + w_{\ge k} \, \sigma_N^{(k)},
\end{equation}
where we defined the weights $w_n = \pi_n \prod_{i=0}^{n-1} (1-\pi_i)$ for the discrete part, in terms of the parameters $\pi_i$ of Sec.~\ref{sec:optimization_method}, and the total weight of the tail $w_{\ge k} = \prod_{n=0}^{k-1} (1-\pi_n)$. The matrix elements of the tail $\sigma_N^{(k)}$ are defined by:
\begin{equation}
    [\sigma_N^{(k)}]_{mm} = 
    \begin{cases} 
        0 & \text{if } m < k, \\
        \frac{1}{N-k+1} \left(\frac{N-k}{N-k+1} \right)^{m-k} & \text{if } m \ge k,
    \end{cases}
\end{equation}
which is a renormalized geometric (thermal) tail with mean photon number adjusted to $N$ on the conditional subspace $m\ge k$.
Here, the parameter $N$ is determined by the energy constraint $\text{Tr}[\hat{n}\rho_E^{(k)}] = E$, which yields:
\begin{equation}
    N = \frac{E - \sum_{n=0}^{k-1} n w_n }{w_{\ge k}}.
\end{equation}
Since the state $\rho_E^{(k)}$ is diagonal and the supports of the discrete part and the tail are orthogonal, the entropy $S(\rho_E^{(k)})$ can be computed as the sum of the Shannon entropy of the mixing probabilities and the weighted entropy of the components:
\begin{equation}\label{eq:app_input_entropy}
    S(\rho_E^{(k)}) = H(\{w_0, \dots, w_{k-1}, w_{\ge k}\}) + w_{\ge k} S(\sigma_N^{(k)}).
\end{equation}
Explicitly, substituting the weights defined above:
\begin{equation}
\begin{split}
    S(\rho_E^{(k)}) &= - \sum_{n=0}^{k-1} w_n \log_2 (w_n) - w_{\ge k} \log_2 (w_{\ge k}) + w_{\ge k} S(\sigma_N^{(k)}),
\end{split}
\end{equation}
where the entropy of the shifted thermal tail is related to the standard bosonic entropy function by:
\begin{equation}
    S(\sigma_N^{(k)}) = g(N-k).
\end{equation}
This requires $N \ge k$, the feasibility condition discussed in Sec.~\ref{sec:optimization_method}.

\subsection{Output State and Output Entropy}

\noindent The action of the fading channel on the input state yields the output state:
\begin{equation}
    \Phi_{\{p_n,\lambda_n\}}(\rho_E^{(k)}) = \sum_{j=1}^{d} p_j \mathcal{E}_{\lambda_j}(\rho_E^{(k)}).
\end{equation}
Since the input is Fock-diagonal and the channel is phase-insensitive, the output is also Fock-diagonal. We compute the diagonal elements $[\mathcal{E}_{\lambda}(\rho_E^{(k)})]_{mm}$ for a generic transmissivity $\lambda$. Due to the structure of the input state, we distinguish two cases for the index $m$.

\paragraph*{(i) Populations for levels $m<k$.} For $m < k$, the population receives contributions from both the discrete part (decaying from higher Fock states $n \ge m$) and the tail:
\begin{equation}\label{eq:app_out_diag_small_m}
    \left[ \mathcal{E}_{\lambda} (\rho_E^{(k)})\right]_{mm} = \sum_{n=m}^{k-1} B_m(n,\lambda) w_n + w_{\ge k} \left( \frac{N-k+1}{N-k}\right)^k \left( [\rho_{\text{th},\lambda(N-k)}]_{mm} - \sum_{n=m}^{k-1} B_m(n,\lambda) [\rho_{\text{th},N-k}]_{nn} \right),
\end{equation}
where $B_m(n,\lambda) = \binom{n}{m}\lambda^m(1-\lambda)^{n-m}$ are the binomial transition probabilities of the pure loss channel.

\paragraph*{(ii) Populations for levels $m\ge k$.} For $m \ge k$, the contribution from the discrete part of the input state (which contains Fock numbers $n < k$) is identically zero, because a pure-loss channel cannot increase the photon number. Therefore, the output diagonal elements are determined solely by the evolution of the shifted thermal tail.
Crucially, since the evolution of a thermal state through a pure-loss channel yields another thermal state with scaled mean photon number, and the truncation of the first $k$ terms does not affect the terms $m \ge k$ of the output, the relation holds exactly:
\begin{equation}\label{eq:app_out_diag_large_m}
    \left[ \mathcal{E}_{\lambda} (\rho_E^{(k)})\right]_{mm} = w_{\ge k} \left( \frac{N-k+1}{N-k}\right)^k [\rho_{\text{th},\lambda (N-k)}]_{mm}.
\end{equation}
The total output entropy is then the Shannon entropy of the diagonal distribution averaged over the fading realization:
\begin{equation}\label{app:output entropy generalized}
    S(\Phi(\rho_E^{(k)})) = - \sum_{m=0}^{\infty} \Lambda_m \log_2 (\Lambda_m),
\end{equation}
where $\Lambda_m = \sum_{j=1}^{d} p_j [\mathcal{E}_{\lambda_j} (\rho_E^{(k)})]_{mm}$. All terms in~\eqref{app:output entropy generalized} are explicitly computable once the finite set $\{p_n,\lambda_n\}_{n=1}^{d}$ and the input parameters $(\pi_n)_{n<k}, N$ are fixed; in practice the infinite sums are truncated using the error bound discussed below.

\subsection{Exchange Entropy}

\noindent The exchange entropy is defined as $S(\tilde{\Phi}(\rho_E^{(k)}))$, where $\tilde{\Phi}$ is the complementary channel of $\Phi$. For a fading channel defined by the ensemble $\{p_n, \lambda_n\}$, the output of the complementary channel can be written in a block-diagonal form with respect to the environment basis $\{\ket{i}_E\}$:
\begin{equation}
    \tilde{\Phi}_{\{p_n,\lambda_n\}}(\rho) = \sum_{i=0}^{\infty} \ket{i}_E\bra{i} \otimes \tilde{\Phi}_i(\rho),
\end{equation}
where each $\tilde{\Phi}_i(\rho)$ is a $d \times d$ matrix acting on the auxiliary system $A$ (which tracks the fading realization). The matrix elements of the block $\tilde{\Phi}_i$ are given by:
\begin{equation}
    [\tilde{\Phi}_i(\rho)]_{nn'} = \sqrt{p_n p_{n'}} \sum_{m=i}^{\infty} \sqrt{B_i(m, 1-\lambda_n) B_i(m, 1-\lambda_{n'})} \, \rho_{mm}.
\end{equation}
We can simplify this expression by utilizing the explicit form of the binomial coefficients. Let us define the geometric mean of the transmissivities $\bar{\lambda}_{nn'} = \sqrt{\lambda_n \lambda_{n'}}$ and the term $\mu_{nn'} = \sqrt{(1-\lambda_n)(1-\lambda_{n'})}$. We observe that:
\begin{equation}
\begin{split}
    \sum_{m=i}^{\infty} \sqrt{B_i(m, 1-\lambda_n) B_i(m, 1-\lambda_{n'})} \, \rho_{mm} 
    &= \sum_{m=i}^{\infty} \binom{m}{i} \mu_{nn'}^i (\bar{\lambda}_{nn'})^{m-i} \rho_{mm} \\
    &= \left( \frac{\mu_{nn'}}{1-\bar{\lambda}_{nn'}} \right)^i \sum_{m=i}^{\infty} \binom{m}{i} (1-\bar{\lambda}_{nn'})^i (\bar{\lambda}_{nn'})^{m-i} \rho_{mm} \\
    &= \left( \frac{\mu_{nn'}}{1-\bar{\lambda}_{nn'}} \right)^i \left[ \mathcal{E}_{1-\bar{\lambda}_{nn'}} (\rho) \right]_{ii}.
\end{split}
\end{equation}
In the last step, we recognized the sum as the $i$-th diagonal element of the state $\rho$ evolved through a fictitious pure-loss channel with transmissivity $1-\bar{\lambda}_{nn'}$. 
Thus, the $i$-th block of the complementary state is:
\begin{equation}\label{eq:app_exchange_block}
\tilde{\Phi}_i(\rho_E^{(k)}) = \sum_{n,n'=0}^{d-1} \sqrt{p_n p_{n'}} \left[ \frac{\sqrt{(1-\lambda_{n})(1-\lambda_{n'})} }{1-\sqrt{\lambda_n \lambda_{n'}}}\right]^i \left[ \mathcal{E}_{1-\sqrt{\lambda_n \lambda_{n'}}} (\rho_E^{(k)})\right]_{ii} \ket{n}_A \bra{n'}.
\end{equation}
The expression of $\mathcal{E}_{1-\sqrt{\lambda_n \lambda_{n'}}} (\rho_E^{(k)})$ is evaluated from~\eqref{eq:app_out_diag_small_m} and~\eqref{eq:app_out_diag_large_m}. Finally, the entropy can be computed by diagonalizing each $d \times d$ block $\tilde{\Phi}_i(\rho_E^{(k)})$ numerically and summing the entropies of their eigenvalues.
\subsection{Truncation Error Bound}\label{app:truncation_error_bound}

\noindent Since the Fock space is infinite-dimensional, any numerical evaluation of entropic functionals requires truncating the Hilbert space at a finite photon number $\bar{n}$. We now derive a rigorous upper bound on the error introduced by this truncation, ensuring the reliability of our numerical results.
Let $\rho = \sum_{n=0}^\infty \gamma_n \ket{\psi_n}\bra{\psi_n}$ be the exact (infinite-dimensional) diagonalized state we wish to analyze, where $\gamma_n$ are the eigenvalues. We introduce a truncation cutoff $\bar{n}$ and define the truncation error parameter $\varepsilon$ as the total probability mass in the discarded tail:
\begin{equation}
    \varepsilon := 1 - \sum_{n=0}^{\bar{n}-1} \gamma_n = \sum_{n=\bar{n}}^{\infty} \gamma_n.
\end{equation}
We can decompose the full state $\rho$ into a sum of a non-normalized truncated part $\tilde{\rho}_{\bar{n}}$ (supported on the subspace of the first $\bar{n}$ eigenvectors) and a normalized tail state $\rho_{\text{tail}}$ weighted by $\varepsilon$:
\begin{equation}
    \rho = \tilde{\rho}_{\bar{n}} + \varepsilon \rho_{\text{tail}},
\end{equation}
where the normalized tail is defined as:
\begin{equation}
    \rho_{\text{tail}} = \frac{1}{\varepsilon} \sum_{n=\bar{n}}^{\infty} \gamma_n \ket{\psi_n}\bra{\psi_n}.
\end{equation}
Since $\tilde{\rho}_{\bar{n}}$ and $\rho_{\text{tail}}$ have orthogonal supports, the von Neumann entropy $S(\rho) = -\text{Tr}[\rho \log_2 \rho]$ decomposes as:
\begin{equation}
    S(\rho) = S(\tilde{\rho}_{\bar{n}}) - \varepsilon\log_2 \varepsilon + \varepsilon S(\rho_{\text{tail}}),
\end{equation}
where $S(\tilde{\rho}_{\bar{n}}) = - \sum_{n=0}^{\bar{n}-1} \gamma_n \log_2 \gamma_n$ is the entropy contribution computed numerically from the truncated spectrum.
We are interested in bounding the error between the true entropy $S(\rho)$ and the computed quantity. Rearranging the terms:
\begin{equation}
    \Delta S := S(\rho) - (S(\tilde{\rho}_{\bar{n}}) - \varepsilon \log_2 \varepsilon) = \varepsilon S(\rho_{\text{tail}}).
\end{equation}
To upper bound this error, we observe that for a given mean energy, the entropy is maximized by a thermal state. Thus, $S(\rho_{\text{tail}}) \le g(E_{\text{tail}})$, where $g(N)$ is the bosonic entropy function and $E_{\text{tail}}$ is the mean energy of the tail state.
The energy of the tail can be computed from the total energy $E(\rho)$ and the energy of the truncated part $E(\tilde{\rho}_{\bar{n}})$:
\begin{equation}
    E_{\text{tail}} = \frac{E(\rho) - E(\tilde{\rho}_{\bar{n}})}{\varepsilon} = \frac{\Delta E_{\bar{n}}}{\varepsilon}.
\end{equation}
Therefore, we obtain the bound:
\begin{equation}\label{eq:truncation error bound}
    \Delta S \le \varepsilon g\left(\frac{\Delta E_{\bar{n}}}{\varepsilon}\right).
\end{equation}
Since the tail state is supported on Fock states with $n \ge \bar{n}$, its mean energy satisfies $E_{\text{tail}} \ge \bar{n}$. In the regime of high truncation cutoff ($\bar{n} \gg 1$), we have $E_{\text{tail}} \gg 1$, and the expansion $g(N) = \log_2 N + \log_2 e + O(N^{-1})$ applies. The error bound then behaves asymptotically as:
\begin{equation}
    \Delta S \lesssim \varepsilon \left[ \log_2\!\left(\frac{\Delta E_{\bar{n}}}{\varepsilon}\right) + \log_2 e \right].
\end{equation}
This confirms that as long as the probability mass of the tail $\varepsilon$ decays faster than the inverse of the tail's energy (which is guaranteed for states with finite moments, like Gaussian states and our optimized Fock-diagonal states), the error vanishes. More precisely, $\Delta S \to 0$ whenever $\varepsilon \log_2(\Delta E_{\bar{n}}/\varepsilon) \to 0$, which holds for every state of finite energy. The same bound applies to the output entropy, and to the exchange entropy, whose tail consists of blocks of dimension $d$, with the additional term $\varepsilon\log_2 d$. In practice the bound is small: for the states of Fig.~\ref{fig:a5_optimal_states}, a cutoff of $2000$ levels leaves a discarded mass below $10^{-78}$ in both the output and the environment, so that the certified interval on $I_c$ is narrower than machine precision.
Finally, for the output and exchange entropies, the relevant energies are simply scaled by the channel transmissivity. For a fading channel with mean transmissivity $\braket{\lambda}$:
\begin{align}
    E(\Phi(\rho)) &= \braket{\lambda}E(\rho), \\
    E(\tilde{\Phi}(\rho)) &= (1-\braket{\lambda})E(\rho).
\end{align}
These relations allow us to apply the same bound to all entropic terms in the mutual information.

\subsection{Summary: mutual information and numerical recipe}

Putting the pieces together, for the chosen input $\rho=\rho_E^{(k)}$ and channel $\Phi = \Phi_{\{p_n,\lambda_n\}} = \sum_{n=0}^{d-1} p_n \mathcal{E}_{\lambda_n}$ we compute:
\begin{enumerate}
  \item $S(\rho_E^{(k)})$ from~\eqref{eq:app_input_entropy}.
  \item For each channel component $\lambda_n$ compute the output populations~\eqref{eq:app_out_diag_small_m},\eqref{eq:app_out_diag_large_m}, then form the mixture populations $\Lambda_m = \sum_{j=1}^{d} p_j [\mathcal{E}_{\lambda_j} (\rho_E^{(k)})]_{mm}$ by averaging on the distribution and evaluate $S(\Phi(\rho_E^{(k)}))= - \sum_m \Lambda_m \log_2 \Lambda_m$ by~\eqref{app:output entropy generalized}, summing up to the truncation cutoff.
  \item For the complementary channel compute each block $\tilde\Phi_i(\rho_E^{(k)})$ by~\eqref{eq:app_exchange_block}, diagonalize the $d\times d$ matrix to obtain $S(\tilde\Phi_i(\rho_E^{(k)}))$ and sum over $i$ up to the truncation cutoff to obtain $S(\tilde\Phi(\rho_E^{(k)}))$ up to a desired precision.
  \item The mutual information follows as
\begin{equation}
I(\Phi,\rho_E^{(k)})\simeq S(\rho_E^{(k)})- \sum_{m=0}^{i_\text{max}} 
\Lambda_m \log_2 \Lambda_m-\sum_{i=0}^{\tilde{i}_{\max}} S(\tilde\Phi_i(\rho_E^{(k)}))
\end{equation}
  where $i_{\max},\tilde{i}_{\max}$ are chosen so that the error bound of App.~\ref{app:truncation_error_bound} is smaller than the desired precision.
\end{enumerate}

\section{Optimality of Fock-diagonal inputs}\label{appendix: diagonal optimizer}

\begin{lemma}
	Let $\Phi$ be a phase-insensitive quantum channel. That is, for any input state $\rho$ and any phase angle $\theta \in \mathbb{R}$:
	\begin{equation}\label{eq:phase_covariance}
		\Phi(e^{-i\theta\hat{N}} \rho e^{i\theta\hat{N}}) = e^{-i\theta\hat{N}} \Phi(\rho) e^{i\theta\hat{N}},
	\end{equation}
	where $\hat{N} = a^\dagger a$ is the photon number operator. Then, the maximum of the entanglement-assisted classical capacity, characterized by the mutual information $I(\Phi, \rho)$, is attained by a state $\rho$ that is diagonal in the Fock basis. Moreover, a state $\rho$ is strictly sub-optimal whenever $\Phi(\rho)$ is not Fock-diagonal; in particular, if $\Phi$ does not destroy all off-diagonal coherences, i.e., if for any non-Fock-diagonal $\rho$ there exists $m \neq n$ such that $\langle m | \Phi(\rho) | n \rangle \neq 0$, then every non-Fock-diagonal state is strictly sub-optimal. If in addition $\Phi$ does not map two distinct Fock-diagonal states to the same output, there is at most one maximizer.
\end{lemma}

\begin{proof}
    The proof relies on the concavity of the quantum mutual information with respect to the input state~\cite{WildeBook}. For any ensemble $\{p_x, \rho_x\}$, the following inequality holds:
	\begin{equation}\label{eq:concavity}
		I\left(\Phi, \sum_x p_x \rho_x\right) \ge \sum_x p_x I(\Phi, \rho_x).
	\end{equation}
	Let us consider the phase twirling operation, which averages a state over all possible phase rotations:
	\begin{equation}
	    \bar{\rho} = \int_0^{2\pi} \frac{d\theta}{2\pi} \, e^{-i\theta\hat{N}} \rho \, e^{i\theta\hat{N}}.
	\end{equation}
    It is easy to verify that for any input $\rho$, the resulting state $\bar{\rho}$ is diagonal in the Fock basis (i.e., all off-diagonal elements in the number basis vanish due to the integration of $e^{i\theta(n-m)}$).
    We now show that the mutual information is invariant under phase rotations for a phase-insensitive channel. The quantum mutual information is defined as:
	\begin{equation}
		I(\Phi,\rho) = S(\rho) + S(\Phi(\rho)) - S_e(\rho),
	\end{equation}
	where $S_e(\rho) = S((\mathcal{I}_R \otimes \Phi_A)(\ket{\psi_\rho}\bra{\psi_\rho}_{RA}))$ is the entropy exchange, calculated on the output of the channel $\Phi$ acting on the system $A$ of a purification $\ket{\psi_\rho}_{RA}$.
	
	Let $U_\theta = e^{-i\theta\hat{N}}$. We analyze the three terms of the mutual information for the rotated state $\rho_\theta = U_\theta \rho U_\theta^\dagger$:
    \begin{enumerate}
        \item \textbf{Input Entropy:} The von Neumann entropy is invariant under unitary transformations:
        \begin{equation}
            S(\rho_\theta) = S(U_\theta \rho U_\theta^\dagger) = S(\rho).
        \end{equation}
        
        \item \textbf{Output Entropy:} Using the phase-covariance property of the channel (Eq.~\ref{eq:phase_covariance}) and the unitary invariance of entropy:
        \begin{equation}
            S(\Phi(\rho_\theta)) = S(U_\theta \Phi(\rho) U_\theta^\dagger) = S(\Phi(\rho)).
        \end{equation}
        
        \item \textbf{Entropy Exchange:} Let $\ket{\psi_\rho}$ be a purification of $\rho$. Then $(\mathbbm{1}_R \otimes U_\theta)\ket{\psi_\rho}$ is a purification of $\rho_\theta$. The joint output state becomes:
        \begin{align}
            \rho_{RB}' &= (\mathcal{I} \otimes \Phi) \left[ (\mathbbm{1} \otimes U_\theta)\ket{\psi_\rho}\bra{\psi_\rho}(\mathbbm{1} \otimes U_\theta^\dagger) \right] \nonumber \\
            &= (\mathbbm{1} \otimes U_\theta) \left[ (\mathcal{I} \otimes \Phi)(\ket{\psi_\rho}\bra{\psi_\rho}) \right] (\mathbbm{1} \otimes U_\theta^\dagger).
        \end{align}
        Since the joint output state of the rotated input is unitarily equivalent to the original joint output, their entropies are identical:
        \begin{equation}
            S_e(\rho_\theta) = S(\rho_{RB}') = S_e(\rho).
        \end{equation}
    \end{enumerate}
    
\noindent Combining these results, we obtain $I(\Phi, \rho_\theta) = I(\Phi, \rho)$ for any $\theta$. Finally, applying the concavity of mutual information to the continuous convex combination represented by the integral:
	\begin{align}
		I(\Phi, \bar{\rho}) &= I\left(\Phi, \int_0^{2\pi} \frac{d\theta}{2\pi} U_\theta \rho U_\theta^\dagger \right) \nonumber \\
        &\ge \int_0^{2\pi} \frac{d\theta}{2\pi} I(\Phi, U_\theta \rho U_\theta^\dagger) \nonumber \\
        &= \int_0^{2\pi} \frac{d\theta}{2\pi} I(\Phi, \rho) = I(\Phi, \rho).
	\end{align}
	This inequality shows that for any arbitrary state $\rho$, the Fock-diagonal state $\bar{\rho}$ yields a mutual information that is at least as large.
\noindent We now show that this inequality is strict whenever $\Phi(\rho)$ is not Fock-diagonal, and hence, under the hypothesis of the lemma, whenever $\rho$ is not Fock-diagonal. Using the identity $I(\Phi,\rho) = S(\Phi(\rho)) + 
    S(\Phi(\rho) | \tilde{\Phi}(\rho))$, where both terms are concave 
    in $\rho$, it suffices to show that the concavity of $S(\Phi(\rho))$ 
    is strict. Introducing the classical random variable $\theta$ 
    with uniform distribution over $[0, 2\pi)$, the concavity inequality 
    for $S(\Phi(\rho))$ can be rewritten as:
    \begin{equation}
        S(\Phi(\bar{\rho})) \geq S(\Phi(\rho) | \theta) \iff I(B:\theta) \geq 0,
    \end{equation}
    where $B$ denotes the output system and $I(B:\theta)$ is the 
    classical-quantum mutual information between the output and $\theta$. 
    The equality $I(B:\theta) = 0$ holds if and only if $B$ and $\theta$ are 
    uncorrelated, i.e., all output states $\Phi(\rho_\theta) = 
    U_\theta \Phi(\rho) U_\theta^\dagger$ are identical. This occurs 
    if and only if $\Phi(\rho)$ is Fock-diagonal. If $\Phi$ does not destroy all off-diagonal coherences, that is, if $\Phi(\rho)$ is not Fock-diagonal for any $\rho$ that is not Fock-diagonal, this holds if and only if $\rho$ itself is Fock-diagonal. Therefore, under this hypothesis, $I(\Phi, \bar{\rho}) > I(\Phi, \rho)$ for every $\rho$ that is not Fock-diagonal.
\noindent The same argument gives uniqueness. For two distinct Fock-diagonal states
$\rho_1 \neq \rho_2$ and the ensemble $\{1/2, \rho_1; 1/2, \rho_2\}$, the gap in the
concavity inequality for $S(\Phi(\cdot))$ is again the Holevo quantity of
$\{1/2, \Phi(\rho_1); 1/2, \Phi(\rho_2)\}$, which vanishes if and only if
$\Phi(\rho_1) = \Phi(\rho_2)$. Hence, if $\Phi$ does not map two distinct Fock-diagonal
states to the same output, $S(\Phi(\cdot))$ is strictly concave on the Fock-diagonal states;
since the conditional entropy of the output, $S(B|E) = S(\rho) - S_e(\rho)$, is concave in $\rho$,
$I(\Phi, \cdot)$ is strictly concave as well on the convex set of Fock-diagonal states with
$\text{Tr}[a^\dagger a\,\rho] \le E$, so that there is at most one maximizer there. All the
entropies involved are finite on that set, since
$S(\Phi(\rho)) \le g(\braket{\lambda} E)$ and, by subadditivity on a purification of
$\rho$, $S_e(\rho) \le S(\rho) + S(\Phi(\rho)) \le g(E) + g(\braket{\lambda} E)$.
The two hypotheses are independent: the first constrains the action of $\Phi$ on the off-diagonal
bands, the second its action on the photon-number distributions.
\end{proof}

\begin{figure*}[t]
    \centering
    \includegraphics[width=0.98\linewidth]{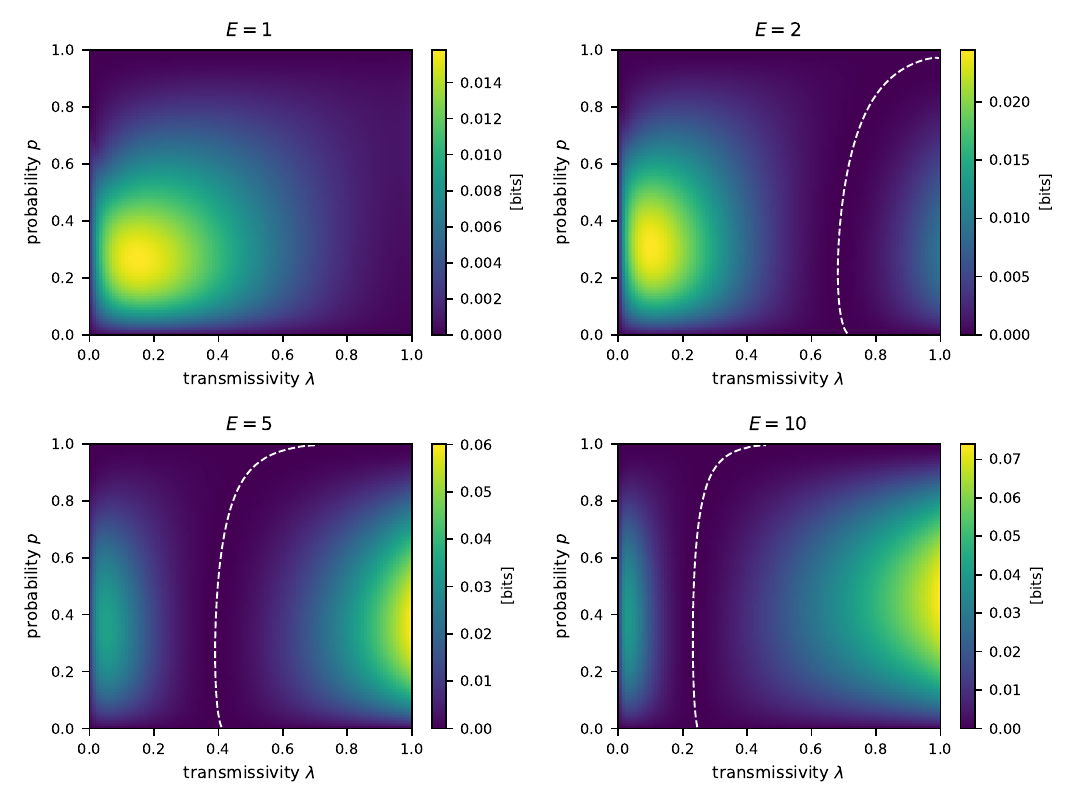}

    \caption{Achievable quantum mutual information gain over the thermal benchmark, $I(\Phi^{(0)}_{p,\lambda},\rho^{(1)}) - I(\Phi^{(0)}_{p,\lambda},\rho_{\text{th}})$, for the erasure-lossy channel, where $\rho^{(1)}$ is the first-order variational state. The panels compare different energy constraints: $E=1$ (top-left), $E=2$ (top-right), $E=5$ (bottom-left), and $E=10$ (bottom-right). The region of non-Gaussian advantage persists and evolves with increasing energy. The gain vanishes on the dashed curves, on which the thermal state is stationary within the first-order family; there, second-order states still outperform it. At $E=1$ there is no such curve: the first-order state outperforms the thermal state at all computed points.}
    \label{fig:gain_grid_2x2}
\end{figure*}

\noindent Therefore, we can restrict the optimization of $I(\Phi, \rho)$ to 
the set of Fock-diagonal states without loss of generality. For a bosonic fading channel with at least one positive transmissivity, strict sub-optimality admits a simple sufficient condition on the input. Since $\Phi$ is a mixture of pure-loss channels, the normally ordered moments of the output are $\langle a^{\dagger j} a^k\rangle_{\mathrm{out}} = \braket{\lambda^{(j+k)/2}}\,\langle a^{\dagger j} a^k\rangle$, whereas a Fock-diagonal state has $\langle a^{\dagger j} a^k\rangle = 0$ for all $j \neq k$. Hence, if the input has a finite nonvanishing moment $\langle a^{\dagger j} a^k\rangle$ with $j \neq k$, its output is not Fock-diagonal and, by the lemma, the input is strictly sub-optimal. This is the case of every non-thermal Gaussian state $\rho_G$, for which $\langle a \rangle \neq 0$ or $\langle a^2 \rangle \neq 0$~\cite{Weedbrook2012}: $I(\Phi, \rho_G) < I(\Phi, \bar{\rho}_G) \leq I(\Phi, \rho^*)$, where $\bar{\rho}_G$ is Fock-diagonal with the same energy as $\rho_G$ and $\rho^*$ is the maximizer over all Fock-diagonal states at that energy. Every non-thermal Gaussian state is thus ruled out as a co-optimizer, while the thermal state, the only Fock-diagonal Gaussian state at fixed energy~\cite{Weedbrook2012}, is optimal only where it is the Fock-diagonal maximizer. States of mean photon number below $E$ are strictly sub-optimal as well whenever $C_\text{E}(\Phi,\cdot)$ is unbounded: $C_\text{E}$ is non-decreasing and, by concavity of $I(\Phi,\cdot)$, concave in the energy, hence strictly increasing, so that a state of energy $N < E$ has $I(\Phi,\rho) \le C_\text{E}(\Phi,N) < C_\text{E}(\Phi,E)$. This is the case for every fading channel with finitely many components and some $\lambda_n > 0$, since coherent-state inputs give $C_\text{E}(\Phi,E) \ge \sum_n p_n g(\lambda_n E) - H(\{p_n\})$. Together with the argument above for non-thermal Gaussian states, the thermal state of mean photon number $E$ is the only Gaussian state that can achieve $C_\text{E}(\Phi,E)$. For the erasure-like channel both hypotheses of the lemma hold, since $\Phi_p(\rho) = p\,\rho + (1-p)\ket{0}\bra{0}$ preserves every coherence and is injective for $p > 0$; the maximizer is then unique, and the thermal state can be optimal only on the curve $\Gamma$ of Sec.~\ref{sec:erasure_channel}, where it is found numerically to be optimal.

\section{Plots Gallery and Numerical Analysis}
\label{app:plots_gallery}

\noindent In this appendix, we present a comprehensive gallery of numerical results that support and expand upon the findings discussed in the main text. The following sections provide a detailed visual analysis of the gains over the thermal benchmark and the structure of the optimized states across different fading channel models, starting from simple discrete binary mixtures up to continuous atmospheric turbulence models.

\subsection{Erasure-like channel}
\noindent We begin by analyzing the erasure-like channel, defined by the map:
\begin{equation}
    \Phi_p = p\mathcal{I} + (1-p)\mathcal{E}_0.
\end{equation}
This channel represents a fundamental communication link suffering from complete signal dropouts with probability $1-p$. The numerical evaluation of this channel allows us to observe the core mechanism of non-Gaussian advantage. As shown in Fig.~\ref{fig:optimal_p0_and_gain}, the optimization algorithm correctly identifies that introducing a specific vacuum population $p_0$ is essential to combat the erasure noise. The corresponding relative gain demonstrates that non-Gaussian encodings are particularly crucial in the low-transmission (high-noise) and moderate-energy regimes.

\begin{figure}[t]
    \centering
    \includegraphics[width=0.49\linewidth]{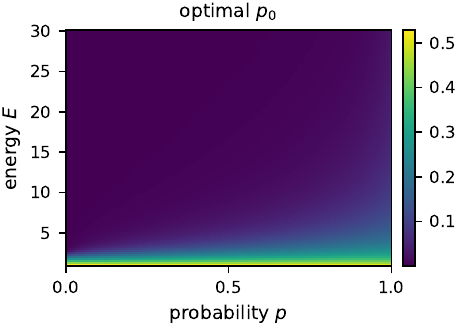}
    \hfill
    \includegraphics[width=0.49\linewidth]{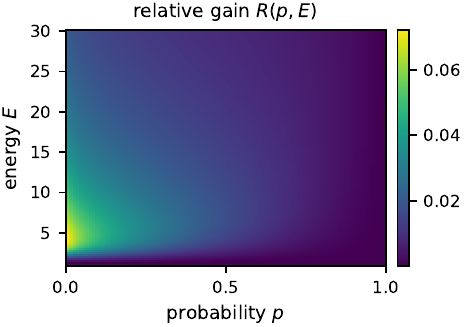}

    \caption{\textbf{Left:} Optimal vacuum population $p_0$ for the erasure-like channel input as a function of fading probability $p$ and energy $E$. Lower $p_0$ is preferred at higher energies to distinguish from the vacuum noise. \textbf{Right:} Relative gain in capacity $R(p,E)$ obtained by using the optimized non-Gaussian state compared to the thermal state. The gain is most significant for low transmission probabilities (noisy channels) and moderate energies. Both panels are computed from the closed form of App.~\ref{app:erasure_derivation}.}
    \label{fig:optimal_p0_and_gain}
\end{figure}

\subsection{Erasure-lossy channel}
\noindent We now generalize the previous model to the erasure-lossy channel, where successful transmission events occur with a finite transmissivity $\lambda$:
\begin{equation} 
    \Phi^{(0)}_{p,\lambda} = p \mathcal{E}_\lambda + (1-p) \mathcal{E}_0.
\end{equation}
This model captures the essence of imperfect transmission interspersed with total signal loss. In Fig.~\ref{fig:gain_grid_2x2}, we report the achievable gain in mutual information over the thermal benchmark provided by our Fock-diagonal ansatz across different energy constraints. To understand the physical structure of these optimized states, Fig.~\ref{fig:bar_diagram_evolution_3x3} tracks the evolution of the optimized photon number distribution during the iterative optimization process. The results confirm that the region of non-Gaussian advantage is not a low-energy artifact, but persists and dynamically evolves as the available energy increases. The gain of the first-order state vanishes on the dashed curves of Figs.~\ref{fig:gain_erasure_lossy}, \ref{fig:gain_grid_2x2} and~\ref{fig:gain_grid_identity_lossy}, on which the thermal state is stationary within the first-order family. The thermal state is nevertheless not optimal on these curves: at the points we tested, second-order states improve on it, by up to about $10^{-3}$ bits. With $k \le 2$, the certified gain is strictly positive at every point of the grids of Figs.~\ref{fig:gain_grid_2x2} and~\ref{fig:gain_grid_identity_lossy}. It falls below $10^{-10}$ bits only for the identity-lossy channel with $\lambda \ge 0.915$, which approaches the identity channel, for which the thermal state is optimal.

\label{app:hierarchy_numerics}\noindent At the point of Figs.~\ref{fig:convergence} and~\ref{fig:bar_diagram_evolution_3x3} the hierarchy has converged by $k = 8$, the further increments being below $10^{-5}$ bits. In a sample of ten points of the two binary families of Sec.~\ref{sec:binary-lossy}, nine reach almost all of their gain before $k = 10$, while at the tenth the gain is still growing at $k = 10$.

\subsection{Identity-lossy channel}\label{app:identity_lossy}
\noindent As a complementary scenario to the erasure models, we consider the identity-lossy channel:
\begin{equation}
\Phi^{(1)}_{p,\lambda} = p \mathcal{I}+ (1-p) \mathcal{E}_\lambda.
\end{equation}
This family describes intermittent ``clear-air" turbulence, where the signal is occasionally attenuated but never completely erased. Figure~\ref{fig:gain_grid_identity_lossy} shows the achievable gain in mutual information over the thermal benchmark for this channel. The gain is largest as $\lambda \to 0$, where the channel approaches the erasure-like channel, and decreases to zero as $\lambda \to 1$, where it becomes the identity channel, for which the thermal state is optimal.

\begin{figure*}[t]
    \centering
    \includegraphics[width=0.98\linewidth]{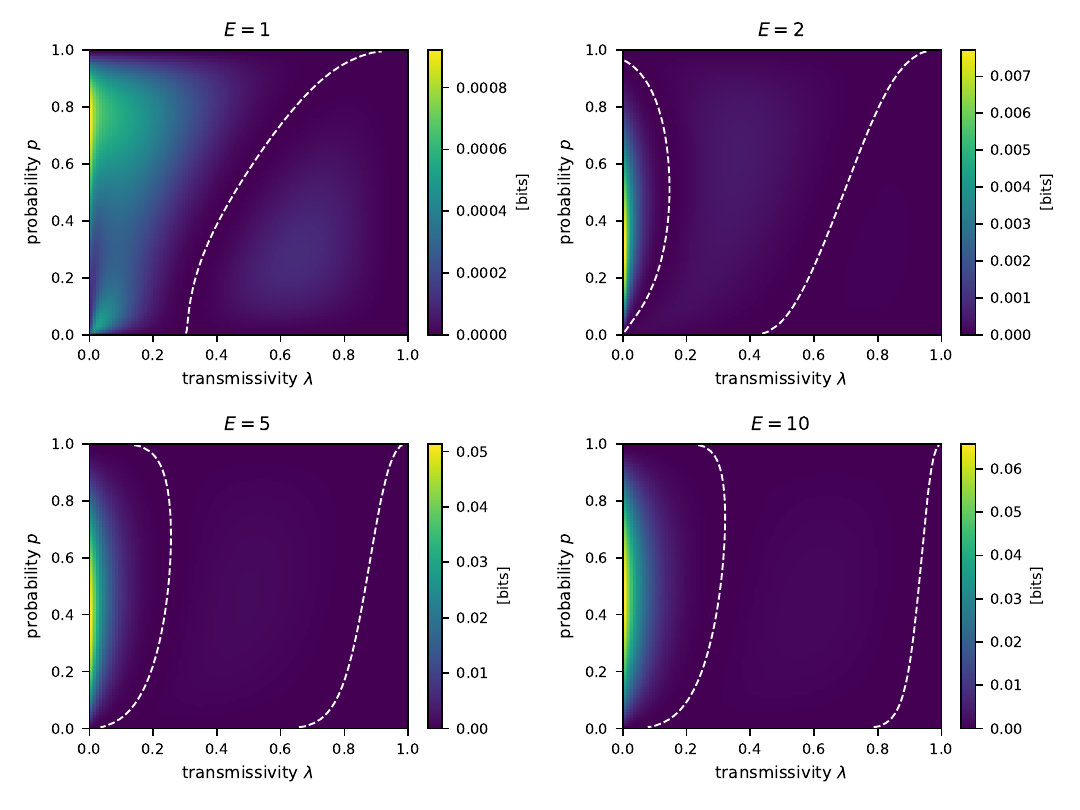}

    \caption{Achievable mutual-information gain over the thermal benchmark, $I(\Phi^{(1)}_{p,\lambda},\rho^{(1)}) - I(\Phi^{(1)}_{p,\lambda}, \rho_{\text{th}})$, for the identity-lossy channel $ \Phi^{(1)}_{p,\lambda}$. Panels correspond to different energy constraints ($E=1, 2, 5, 10$). The gain is largest as $\lambda \to 0$, where the channel approaches the erasure-like channel, and vanishes as $\lambda \to 1$. It also vanishes on the dashed curves, on which the thermal state is stationary within the first-order family; on these curves, second-order states still outperform it.}
    \label{fig:gain_grid_identity_lossy}
\end{figure*}

\subsection{Coherent information and channel activation}
\noindent To further investigate the ultimate quantum limits of the channel, we focus on the single-shot coherent information $I_c$, which provides a rigorous lower bound for the quantum capacity. Figure~\ref{fig:coherent_info_comparison_grid} displays the certified lower bound on $Q_1$ given by the optimized states alongside its gain over the thermal benchmark. The orange contours in the right column enclose the parameter regions where our optimized states enable positive communication rates while the thermal benchmark yields a non-positive coherent information.

\label{app:ic_numerics}\noindent In the maps of Figs.~\ref{fig:activation_Q1}, \ref{fig:coherent_info_comparison_grid}, \ref{fig:lnw_E1} and~\ref{fig:lnw_E10}, $I_c^{\text{th}}$ is maximized over 48 values of $N$ between $10^{-4}E$ and $E$; at the points classified as activation, over 4000 values down to $10^{-6}E$ and, below that, with the bound on $I_c(\Phi, \rho_{\text{th},N})$ of App.~\ref{app:weak_inputs}. Where $\braket{\sqrt{\lambda}}^2 + \braket{\lambda} > 1$, $I_c^{\text{th}} > 0$ by Lemma~\ref{lem:weak_inputs}, even when the positive values lie below the sampled $N$.

\label{app:higher_orders}\noindent \emph{Higher orders.} The maps of Figs.~\ref{fig:activation_Q1}, \ref{fig:coherent_info_comparison_grid}, \ref{fig:lnw_E1} and~\ref{fig:lnw_E10} use states of order $k \le 2$. To see how much this restricts the activation region, we have continued the hierarchy up to $k = 8$ along the line $p = 0.6$, $E = 10$ of Fig.~\ref{fig:activation_Q1}, with the same search and the same certified evaluation. On this line $p\lambda(1+p) < 1$ for every $\lambda$, and for every $\lambda \le 0.95$ considered the thermal states have a negative coherent information at all the 239 sampled values of $N$ between $10^{-6}E$ and $E$. At $\lambda = 0.95$, the point of Fig.~\ref{fig:a5_optimal_states}, the coherent information of the optimized states grows from $0.11$ bits at $k = 5$ to $0.22$ bits at $k = 8$. States of order $k \le 8$ give a positive coherent information down to $\lambda \approx 0.8$: $0.14$ bits at $\lambda = 0.90$, $0.07$ bits at $\lambda = 0.85$ and $10^{-3}$ bits at $\lambda = 0.80$, whereas with $k \le 2$ the positive values start at $\lambda \approx 0.945$. At $\lambda = 0.75$ and $0.70$ the search finds no state with a positive coherent information. At high order the optimized states differ from those of Fig.~\ref{fig:a5_optimal_states}: their coherent information coincides, to the precision given, with that of a state supported on at most four Fock levels, such as $\{2, 4, 6\}$ at $\lambda = 0.94$ or $\{2, 5\}$ at $\lambda = 0.85$, with a mean photon number between $2$ and $5$ and a negligible vacuum weight. The activation region of the maps is therefore a conservative estimate, and at $k = 8$ the hierarchy has not converged.

\begin{figure*}[t] 
    \centering
    \includegraphics[width=0.98\linewidth]{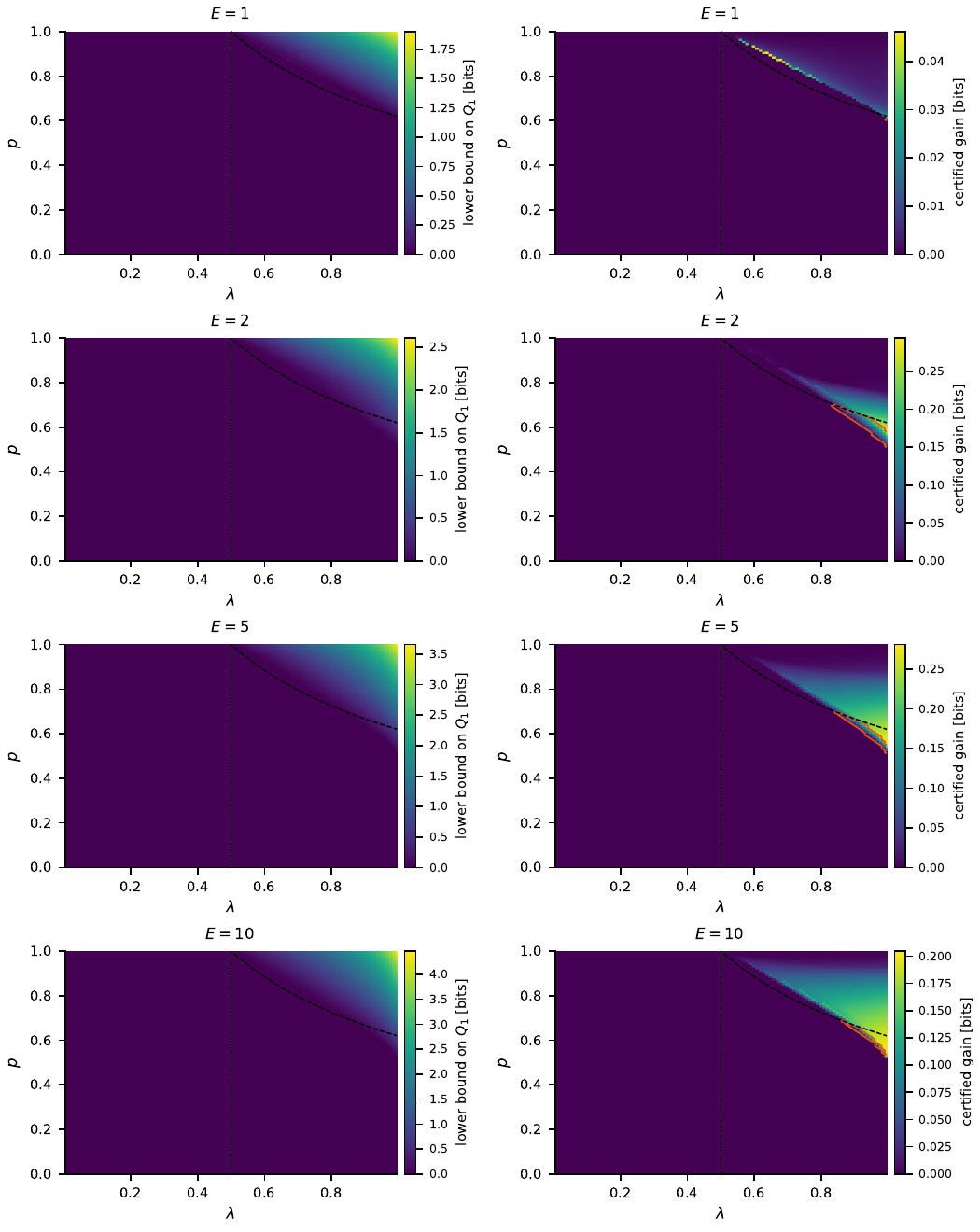}
    
    \caption{Single-shot quantum capacity analysis for the erasure-lossy channel $\Phi_{p,\lambda}^{(0)}$ across increasing energy regimes ($E=1, 2, 5, 10$ from top to bottom). 
    The optimized state belongs to the variational family with $k \le 2$. The thermal benchmark is $I_c^{\text{th}} = \max_{N \le E} I_c(\Phi_{p,\lambda}^{(0)}, \rho_{\text{th},N})$ (Sec.~\ref{sec:quantum-capacity}); above the black dashed line, $p\lambda(1+p) = 1$, weak thermal inputs already have a positive coherent information (Lemma~\ref{lem:weak_inputs}). \textbf{Left Column:} The certified lower bound on $Q_1$ given by the tested inputs, $\max\{0, I_c(\Phi_{p,\lambda}^{(0)}, \rho_{\text{opt}}), I_c^{\text{th}}\}$; to the left of the white dashed line, $\lambda \le 1/2$, the channel is antidegradable and $I_c \le 0$ for every input.
    \textbf{Right Column:} The certified gain of this lower bound over $\max\{0, I_c^{\text{th}}\}$; the orange contours enclose the activation regions. The bright areas highlight parameter regions where the optimized Fock-diagonal states outperform the thermal benchmark; activation proper corresponds to the narrower set where the thermal benchmark is non-positive ($I_c^{\text{th}} \le 0 < I_c^{\text{opt}}$).}
    \label{fig:coherent_info_comparison_grid}
\end{figure*}

\begin{figure*}[t]
    \centering
    \includegraphics[width=0.98\linewidth]{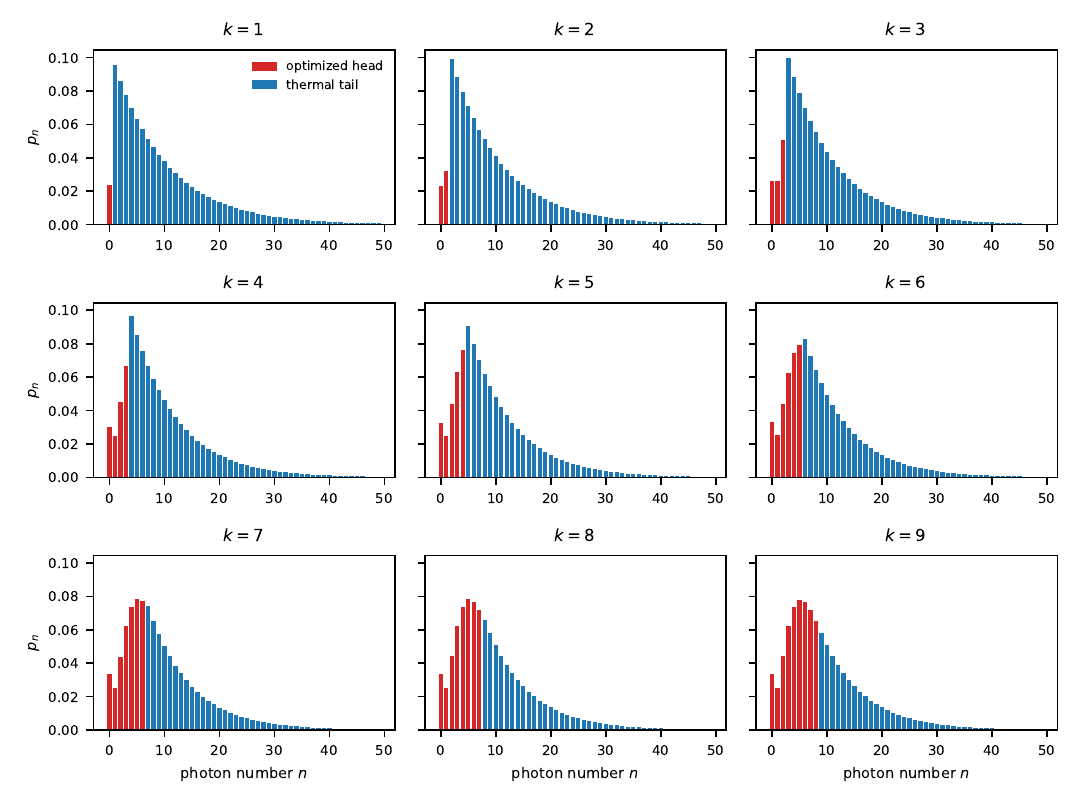}
    
    \caption{Evolution of the optimized photon number distribution $p_n$ as the iterative algorithm progresses from step $k=1$ to $k=9$, for the channel and the energy of Fig.~\ref{fig:convergence}.
    At each step $k$, the first $k$ populations are optimized independently, while the tail ($n \ge k$) follows a thermal distribution.
    The deviation from the monotonic geometric distribution (characteristic of Gaussian thermal states) becomes evident as $k$ increases, revealing the non-Gaussian structure of the optimized state within the order-$k$ ansatz. The optimized populations of the first $k$ levels are shown in red, the thermal tail in blue.}
    \label{fig:bar_diagram_evolution_3x3}
\end{figure*}

\begin{figure}[t]
    \centering
        \includegraphics[width=\linewidth]{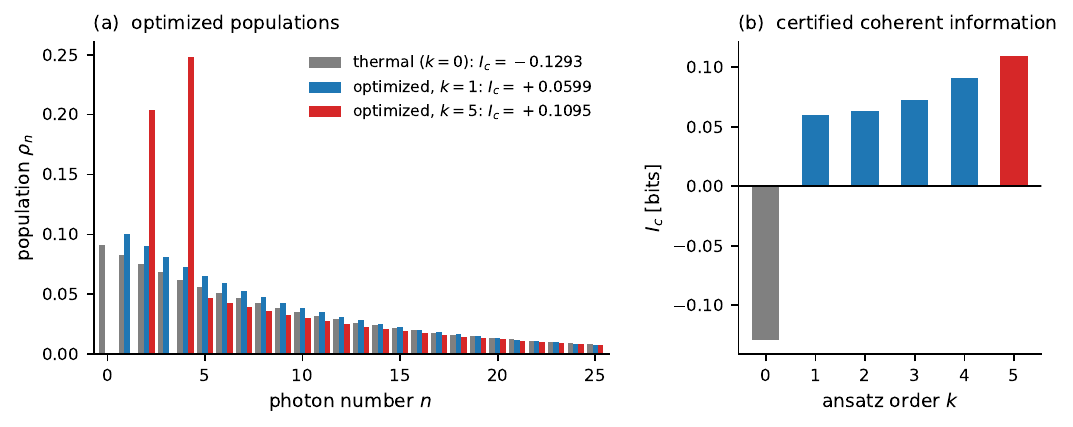}
    \caption{Input states that realize the activation of the quantum capacity, for the binary
    fading channel $\Phi_{p,\lambda}^{(0)}$ at $p = 0.6$, $\lambda = 0.95$, $E = 10$.
    \textbf{(a)} Populations $\rho_n$ of the thermal state of mean photon number $E$ ($k = 0$, grey), of the optimized
    state at first order ($k = 1$, blue) and at fifth order ($k = 5$, red). The optimized states
    remove the vacuum population, which the complete-loss branch of the channel would make
    indistinguishable from an erasure, and concentrate the non-Gaussianity in the lowest Fock
    levels; the remaining weight follows the thermal tail fixed by the energy constraint.
    \textbf{(b)} Coherent information along the hierarchy up to $k = 5$: it is negative for the
    thermal state of mean photon number $E$, as for all the sampled thermal states with
    $N \le E$, and strictly positive from the first order on, so that each order realizes
    activation. Higher
    orders are discussed in App.~\ref{app:higher_orders}.}
    \label{fig:a5_optimal_states}
\end{figure}

\subsection{General fading (Log-Negative Weibull distribution)}
\label{app:general_fading}

\noindent Finally, we apply our numerical framework to a physically realistic continuous fading model. In free-space quantum communications affected by beam wandering, the transmissivity statistics are accurately described by the Log-Negative Weibull (LNW) distribution. We map the non-Gaussian advantage for the single-shot coherent information $I_c$ over the parameters $(R, \gamma)$ of the density introduced in Sec.~\ref{sec:prelim}, with cutoff $\lambda_0 = 1$, for two representative energies. The window $R \in [0.5, 7]$, $\gamma \in [0.1, 20]$, logarithmic in $\gamma$, contains the curve on which the thermal coherent information vanishes, except at the smallest values of $\gamma$, where it moves to larger $R$. The distribution is replaced by the $20$-point Gauss quadrature rule of the law of $\theta = \arcsin\sqrt{\lambda}$: the channel becomes $\sum_q p_q \mathcal{E}_{\lambda_q}$, with $\lambda_q = \sin^2\theta_q$ and positive weights summing to one, so that it is itself a fading channel. We work in $\theta$ because the output and the complementary output depend on $\lambda$ through powers of $\sqrt{\lambda}$ and $\sqrt{1-\lambda}$, which are analytic in $\theta$, so that the rule converges quickly also where the density concentrates near $\lambda = 1$; the nodes and the weights follow from the representation $\braket{\lambda^a} = \int_0^\infty dw\, e^{-w} \exp[-a(2w/R^2)^{\gamma/2}]$, valid for $\lambda_0 = 1$, which holds because $(R^2/2)[\ln(\lambda_0/\lambda)]^{2/\gamma}$ is exponentially distributed under the density of Sec.~\ref{sec:prelim}. Doubling the number of nodes changes the certified gain, and the coherent information at the points of activation, by at most about $10^{-6}$ bits, whereas $20$ bins of equal width in $\lambda$, each represented by its conditional mean, misestimate the coherent information of the same states by up to $4 \times 10^{-2}$ bits. The certified values refer to this discretized channel; the discretization error is estimated, not bounded. The thermal benchmark is $I_c^{\mathrm{th}} = \max_{N \le E} I_c(\Phi, \rho_{\text{th},N})$, as in Sec.~\ref{sec:quantum-capacity}.
Figures~\ref{fig:lnw_E1} and~\ref{fig:lnw_E10} display, for $E=1$ and $E=10$ respectively, three complementary quantities: the certified gain of the lower bound on $Q_1$, a categorical activation map, and the thermal benchmark $I_c^{\mathrm{th}}$ for reference. The activation map distinguishes three qualitatively distinct regimes: (i) no positive coherent information is certified for the tested inputs; (ii) $I_c^{\mathrm{th}} \leq 0 < I_c^{\mathrm{opt}}$, the activation region in which non-Gaussian states strictly enable quantum communication where thermal states fail; and (iii) $I_c^{\mathrm{th}} > 0$, so that quantum communication is achievable with Gaussian inputs; by Lemma~\ref{lem:weak_inputs}, this is the case wherever $\braket{\sqrt{\lambda}}^2 + \braket{\lambda} > 1$ (dashed line).
Two observations follow from these maps. First, once the thermal benchmark is optimized over $N \le E$, activation is rare for this model with the states of order $k \le 2$ used here. In the window, the boundary between regimes~(i) and~(iii) follows the curve $\braket{\sqrt{\lambda}}^2 + \braket{\lambda} = 1$ of Lemma~\ref{lem:weak_inputs} at both energies, to the resolution of the grid. At $E = 1$ we find no activation: wherever the optimized states give a positive coherent information, some thermal state of lower mean photon number has a positive one too. At $E = 10$, activation is found at 5 of the $10^4$ grid points, with a coherent information of up to about 0.15 bits; they lie at $R \simeq 1.35$ and $\gamma \ge 16$, on the upper edge of the window and just on the side of that curve where $\braket{\sqrt{\lambda}}^2 + \braket{\lambda} \le 1$. Second, where the thermal benchmark is positive, the optimized states improve on it only slightly, by $3 \times 10^{-4}$ and $3 \times 10^{-3}$ bits on average at $E = 1$ and $E = 10$. For this realistic model, states of order $k \le 2$ thus bring a small advantage for the coherent information, and with them activation appears only at the boundary. For the binary family, higher orders enlarge the activation region considerably (App.~\ref{app:higher_orders}); for this model they do not: along three lines of the window at fixed $\gamma$, states of order $k \le 8$, with the energy constraint imposed as an inequality, extend the activation only slightly beyond the points of Fig.~\ref{fig:lnw_E10}, and only at large $\gamma$.

\begin{figure*}[t]
    \centering
        \includegraphics[width=0.98\linewidth]{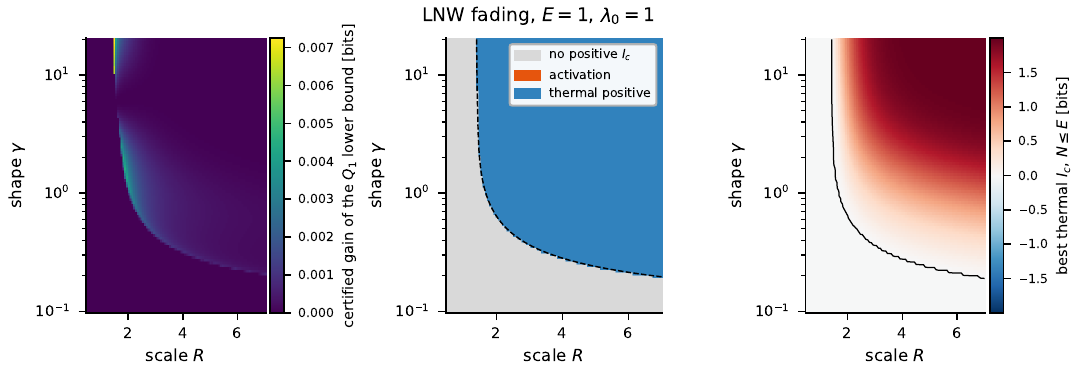}
    
    \caption{Coherent information of the LNW fading channel at low energy ($E=1$, $\lambda_0=1$), over the window $R \in [0.5, 7]$, $\gamma \in [0.1, 20]$ of the parameters of Sec.~\ref{sec:prelim}. The optimized state belongs to the variational family with $k = 1$, and with $k = 2$ near the curve where the thermal state of mean photon number $E$ has vanishing coherent information. The thermal benchmark $I_c^{\mathrm{th}}$ is the best thermal state with $N \le E$ (Sec.~\ref{sec:quantum-capacity}). \textbf{Left:} certified gain of the lower bound on $Q_1$ over $\max\{0, I_c^{\mathrm{th}}\}$. \textbf{Center:} no certified positive coherent information for the tested inputs (grey), activation, $I_c^{\mathrm{th}} \leq 0 < I_c^{\mathrm{opt}}$ (orange), $I_c^{\mathrm{th}} > 0$ (blue); dashed, $\braket{\sqrt{\lambda}}^2 + \braket{\lambda} = 1$. \textbf{Right:} $I_c^{\mathrm{th}}$, with its zero contour.}
    \label{fig:lnw_E1}
\end{figure*}

\begin{figure*}[t]
    \centering
        \includegraphics[width=0.98\linewidth]{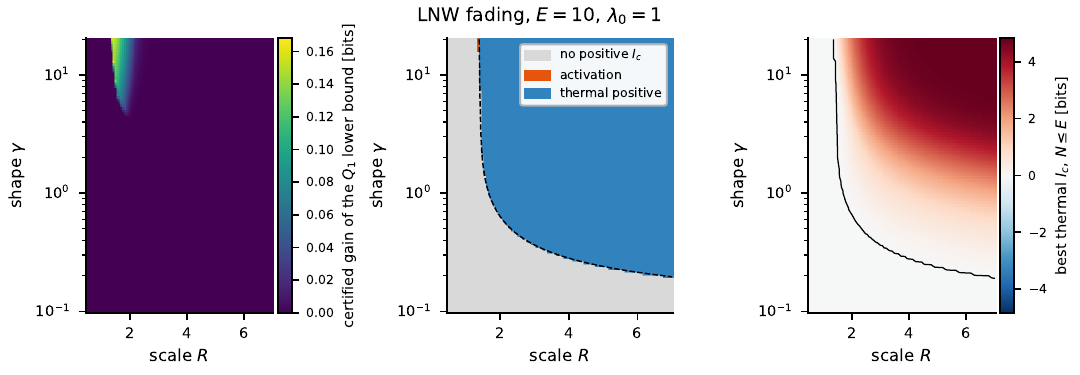}
    
    \caption{Same as Fig.~\ref{fig:lnw_E1}, at high energy ($E = 10$).}
    \label{fig:lnw_E10}
\end{figure*}
\end{document}